\documentclass[11pt]{article}
\usepackage{amsmath}
\usepackage{amsfonts}
\usepackage{url}
\usepackage{amssymb}
\usepackage{amsthm}
\usepackage{graphicx}
\usepackage{fullpage}
\usepackage{verbatim}
\usepackage{appendix} %
\usepackage{rotating} %
\usepackage[sc]{mathpazo} % Palatino with real small caps
\usepackage[utf8]{inputenc}
\usepackage[T1]{fontenc} % Use 8-bit encoding that has 256 glyphs
\usepackage{multirow}
\usepackage{booktabs}
\usepackage{paralist}
\usepackage{setspace}
\usepackage{xr}
\usepackage{tikz} % Figures in tex format
\newcommand{\abs}[1]{\lvert#1\rvert}
\newcommand{\Abs}[1]{\left\lvert#1\right\rvert} % abs w left and right
\newcommand{\cvt}[1]{c_{#1}} % two-sided CI
\newcommand{\cvo}[1]{c_{#1}^{\textnormal{os}}} % one-sided CI

\usepackage[longnamesfirst]{natbib}
\shortcites{connors96rhc}% don't give full citation, there are too many authors
\theoremstyle{plain}
\newtheorem{theorem}{Theorem}[section]
\newtheorem{lemma}{Lemma}[section]
\newtheorem{corollary}{Corollary}[section]
\newtheorem{assumption}{Assumption}[section]

\theoremstyle{definition}

\title{A Simple Adjustment for Bandwidth Snooping\thanks{We thank Joshua
    Angrist, Matias Cattaneo, Victor Chernozhukov, Kirill Evdokimov, Bo
    Honor\'e, Chris Sims, numerous seminar and conference participants, four
    anonymous referees and the editor for helpful comments and suggestions. We
    also thank Matias Cattaneo for sharing the Progresa dataset. All remaining
    errors are our own. The research of the first author was supported by
    National Science Foundation Grant SES-1628939. The research of the second
    author was supported by National Science Foundation Grant SES-1628878.}}
\author{Timothy B. Armstrong\thanks{email: timothy.armstrong@yale.edu}\\
  Yale University \and
  Michal Koles\'{a}r\thanks{email: mkolesar@princeton.edu}\\
  Princeton University}

\begin{document}

\maketitle

\begin{abstract}
  \setstretch{1.3} Kernel-based estimators such as local polynomial estimators
  in regression discontinuity designs are often evaluated at multiple bandwidths
  as a form of sensitivity analysis. However, if in the reported results, a
  researcher selects the bandwidth based on this analysis, the associated
  confidence intervals may not have correct coverage, even if the estimator is
  unbiased. This paper proposes a simple adjustment that gives correct coverage
  in such situations: replace the normal quantile with a critical value that
  depends only on the kernel and ratio of the maximum and minimum bandwidths the
  researcher has entertained. We tabulate these critical values and quantify the
  loss in coverage for conventional confidence intervals. For a range of
  relevant cases, a conventional 95\% confidence interval has coverage between
  70\% and 90\%, and our adjustment amounts to replacing the conventional
  critical value 1.96 with a number between 2.2 and 2.8. Our results also apply
  to other settings involving trimmed data, such as trimming to ensure overlap
  in treatment effect estimation. We illustrate our approach with three
  empirical applications.
\end{abstract}

\clearpage

\section{Introduction}\label{introduction_sec}

Kernel and local polynomial estimators of objects such as densities and
conditional means involve a choice of bandwidth. To assess sensitivity of the
results to the chosen bandwidth, it is common to compute estimates and
confidence intervals for several bandwidths, or plot them against a continuum of
bandwidths. For example, in regression discontinuity designs---a leading
application of non-parametric methods in econometrics---this approach is
recommended in several surveys \citep{imbens_regression_2008,LeLe10,DiLe11} and
implemented widely in applied work.\footnote{\label{fn:rd-snooping-examples}For
  prominent examples see, for instance, \citet{vdk02}, \citet{LeMi08},
  \citet{LuMi07}, or \citet{CaDoMa09}.} However, such practice leads to a
well-known problem that if the bandwidth choice is influenced by these results,
the confidence interval at the chosen bandwidth may undercover, even if the
estimator is unbiased.

This problem does not only arise when the selection rule is designed to make the
results of the analysis look most favorable (for example by choosing a bandwidth
that minimizes the $p$-value for some test). Undercoverage can also occur from
honest attempts to report a confidence interval with good statistical
properties. In settings in which one does not know the smoothness of the
estimated function, it is typically \emph{necessary} to examine multiple
bandwidths to obtain confidence intervals that are optimal (see
Section~\ref{sec:adaptation_example} for details and
\citet{armstrong_adaptive_2015} for a formal statement). We use the term
``bandwidth snooping'' to refer to any situation in which a researcher considers
multiple bandwidths in reporting confidence intervals.

This paper proposes a simple adjustment to account for bandwidth snooping:
replace the usual critical value based on a quantile of a standard normal
distribution with a critical value that depends only on the kernel, order of the
local polynomial, and the ratio of the maximum and minimum bandwidths that the
researcher has tried. We tabulate these adjusted critical values for a several
popular kernels, and show how our adjustment can be applied in regression
discontinuity designs, as well as estimation of average treatment effects under
unconfoundedness after trimming, and estimation of local average treatment
effects.

To explain the adjustment in a simple setting, consider the problem of
estimating the conditional mean $E[Y_{i}\mid X_{i}=x]$ at a point $x$, which we
normalize to zero. Given and i.i.d.\ sample $\{(X_i,Y_i)\}_{i=1}^{n}$, a kernel
$k$, and a bandwidth $h$, the Nadaraya-Watson kernel estimator is given by
\begin{equation*}
  \hat{\theta}(h)=\frac{\sum_{i=1}^n Y_{i}k(X_i/h)}{\sum_{i=1}^n k(X_i/h)},
\end{equation*}
and it is approximately unbiased for the pseudo-parameter
\begin{equation*}
  \theta(h)=\frac{E[Y_{i}k(X_i/h)]}{E[k(X_i/h)]}.
\end{equation*}
Under appropriate smoothness conditions, if we take $h\to 0$ with the sample
size, $\hat\theta(h)$ will converge to $\theta(0):=\lim_{h\to 0}
\theta(h)=E(Y_i\mid X_i=0)$. Given a standard error $\hat{\sigma}(h)/\sqrt{nh}$,
the $t$-statistic $\sqrt{nh}(\hat\theta(h)-\theta(h))/\hat\sigma(h)$ is
approximately standard normal. Letting $z_{1-\alpha/2}$ denote the $1-\alpha/2$
quantile of the standard normal distribution, the standard confidence interval
$[\hat\theta(h)\pm z_{1-\alpha/2}\hat\sigma(h)/\sqrt{nh}]$, is therefore an
approximate $100\cdot (1-\alpha)\%$ confidence interval for $\theta(h)$. If the
bias $|\theta(h)-\theta(0)|$ is small enough relative to the standard error,
such as when the bandwidth $h$ ``undersmooths'', the standard confidence
interval is also an approximate confidence interval for $\theta(0)$, the
conditional mean at zero.

However, if the selected bandwidth $\hat{h}$ is based on examining $\hat\theta(h)$
over $h$ in some interval $[\underline h,\overline h]$, the standard confidence
interval around $\hat{\theta}(\hat{h})$ may undercover even if there is no bias.
To address this problem, we propose confidence intervals that cover $\theta(h)$
simultaneously for all $h$ in a given interval $[\underline h,\overline h]$ with
a prespecified probability. In particular, we derive a critical value
$c_{1-\alpha}$ such that as $n\to\infty$,
\begin{equation}\label{unif_ci_eq}
P\left(\theta(h)\in [\hat\theta(h)\pm c_{1-\alpha}\hat\sigma(h)/\sqrt{nh}]
\text{ for all $h\in[\underline h,\overline h]$}\right)
\to 1-\alpha.
\end{equation}
In other words, our critical values allow for a uniform confidence band for
$\theta(h)$. Thus, the confidence interval for the selected bandwidth,
$[\hat\theta(\hat{h})\pm c_{1-\alpha}\hat\sigma(\hat{h})/\sqrt{n\hat{h}}]$, will
achieve correct coverage of $\theta(\hat{h})$ no matter what selection rule was
used to pick $\hat{h}$.

Our main contribution is to give a coverage result of the
form~\eqref{unif_ci_eq} for a large class of kernel-based estimators
$\hat{\theta}(h)$, as well as a similar statement showing coverage of
$\theta(0)$. The latter follows under additional conditions that allow the bias
to be mitigated through undersmoothing or bias-correction. These conditions are
essentially the same as those needed for pointwise coverage: if $\hat\theta(h)$
is ``undersmoothed and/or bias corrected enough'' that the pointwise CI has good
pointwise coverage of $\theta(0)$ at each $h\in[\underline h,\overline h]$, our
uniform CI will cover $\theta(0)$ uniformly over this set. In particular, we
show how our approach can be combined with a popular bias-correction method
proposed by \citet{cct14}.

Since our confidence bands cover $\theta(h)$ under milder smoothness conditions
than those needed for coverage of $\theta(0)$, they are particularly well-suited
for sensitivity analysis. Suppose that a particular method for bias correction
or undersmoothing implies that, in a given data set, the bias is asymptotically
negligible if $h\le 3$. If one finds that the confidence bands for, say, $h=2$
and $h=3$ do not overlap even after our correction, then one can conclude that
the assumptions needed for this form of bias correction are not supported by the
data. Our confidence bands can thus be used to formalize certain conclusions
about confidence intervals being ``sensitive'' to bandwidth choice.\footnote{An
  alternative approach to sensitivity analysis is to reject a particular null
  hypothesis regarding $\theta(0)$ only when one rejects the corresponding
  hypothesis test based on $\hat{\theta}(h)$ for all values $h$ that one has
  examined. The CI for $\theta(0)$ is then given by the union of the CIs based
  on $\hat \theta(h)$ as $h$ varies over all values that one has examined. This
  form of sensitivity analysis does not require a snooping correction, but is
  typically very conservative. See Section~\ref{sensitivity_analysis_sec} of the
  appendix for further discussion.}

In many applications, $\theta(h)$, taken as a function indexed by the bandwidth,
is an interesting parameter in its own right, in which case our confidence bands
are simply confidence bands for this function. As we discuss in detail in
Section~\ref{applications_sec}, this situation arises, for instance, in
estimation of local average treatment effects for different sets of compliers,
or in estimation of average treatment effects under unconfoundedness with
limited overlap. In the latter case, $h$ corresponds to a trimming parameter
such that observations with propensity score within distance $h$ to $0$ or $1$
are discarded, and $\theta(h)$ corresponds to average treatment effects for the
remaining subpopulation with moderate values of the propensity score.

A key advantage of our approach is that the critical value $c_{1-\alpha}$
depends only on the ratio $\overline h/\underline h$ and the kernel $k$ (in the
case of local polynomial estimators, it also depends on the order of the
polynomial and whether the point is on the boundary of the support). In
practice, researchers often report a point estimate $\hat\theta(\hat{h})$ and a
standard error $\hat\sigma(\hat{h})/\sqrt{n\hat{h}}$. As long as the kernel and
order of the local polynomial are also reported, a reader can use our critical
values to construct a confidence interval that takes into account a
specification search over a range $[\underline h,\overline h]$ that the reader
believes the original researcher used. Alternatively, one can assess the
sensitivity of the conclusions of the analysis to bandwidth specification search
by, say, computing the largest value of $\overline h/\underline h$ for which the
robust confidence interval does not include a particular value. As an example to
give a sense of the magnitudes involved, we find that, with the uniform kernel
and a local constant estimator, the critical value for a two sided uniform
confidence band with $1-\alpha=0.95$ and $\overline h/\underline h=3$ is about
$2.6$ (as opposed to $1.96$ with no correction). If one instead uses the
pointwise-in-$h$ critical value of $1.96$ and searches over
$h\in[\underline h,\overline h]$ with $\overline h/\underline h=3$, the true
coverage (of $\theta(h)$) will be approximately 80\%. The situation for the
triangular kernel is more favorable, with a critical value of around $2.25$ for
the case with $\overline h/\underline h=3$, and with the coverage of the
pointwise-in-$h$ procedure around $91\%$.

We also derive analytic results showing that the critical values grow very
slowly with $\overline h/\underline h$, at the rate
$\sqrt{\log\log(\overline h/\underline h)}$. Thus, from a practical
standpoint, examining a wider range of bandwidths carries only a very small
penalty (relative to examining a moderate range): while using our correction is
important for obtaining correct coverage, the critical values increase quite
slowly once $\overline h/\underline h$ is above $5$. A Monte Carlo study in the
supplemental appendix confirms that these critical values lead to uniform
coverage of $\theta(h)$ that is close to the nominal level. Uniform coverage of
$\theta(0)$ is also good so long as our method is combined with bias correction
or undersmoothing.

We illustrate our results with three empirical applications. First, we apply our
method to the regression discontinuity study of the effect of Progresa from
\citet{cct14}, and we find that the significance of the results is sensitive to
bandwidth snooping. The second empirical application is the regression
discontinuity study of \citet{lee08}. Here, in contrast, we find that, while the
confidence regions are somewhat larger when one allows for examination of
estimates at multiple bandwidths, the overall conclusions of that study are
robust to a large amount of bandwidth snooping. Finally, we consider an
application to estimating treatment effects under unconfoundedness from
\citet{connors96rhc}. Here, we find that the results are again
quite robust to the choice of trimming parameter, providing additional evidence
supporting the study's conclusions.

The rest of the paper is organized as follows. Section~\ref{lit_review_sec}
discusses related literature. Section~\ref{simple_derivation_sec} gives a
heuristic derivation of our asymptotic distribution results in a simplified
setup. Section~\ref{general_result_sec} states our main asymptotic distribution
result under general high-level conditions.
Section~\ref{practical_implementation_sec} gives a step-by-step explanation of
how to find the appropriate critical value in our tables and implement the
procedure. Section~\ref{applications_sec} works out applications of our results
to several econometric models. Section~\ref{empirical_sec} presents an
illustration of our approach in three empirical applications.
Section~\ref{conclusion_sec} concludes. Proofs and auxiliary results, as well as
additional tables and figures and a Monte Carlo study, are given in the appendix
and a supplemental appendix.\footnote{The supplemental appendix is available at
  \url{http://arxiv.org/abs/1412.0267}} Since
Section~\ref{simple_derivation_sec} and the beginning of
Section~\ref{general_result_sec} are concerned primarily with theoretical
aspects of our problem, readers who are primarily interested in implementation
can skip Section~\ref{simple_derivation_sec} and the beginning of
Section~\ref{general_result_sec} up to
Section~\ref{practical_implementation_sec}.

\subsection{Related literature}\label{lit_review_sec}

The idea of controlling for multiple inference by constructing a uniform
confidence band has a long tradition in the statistics literature---see
\citet[Chapter 9]{lehmann_testing_2005} for an overview and early contributions,
and \citet{white_reality_2000} for an application to econometrics. On a
technical level, our results borrow from the literature on Gaussian
approximations to empirical processes and extreme value limits for suprema of
Gaussian processes. To obtain an approximation of the kernel estimator by a
Gaussian process, we use an approximation of
\citet{sakhanenko_convergence_1985}. For the case
$\overline h/\underline h\to \infty$, we then use extreme value theory, and our
derivation is similar in spirit to \citet{bickel_global_1973}, who consider
kernel estimation of a density under a fixed sequence of bandwidths $h=h_{n}$
and derive confidence bands that are uniform in the point $x$ at which the
density is evaluated. For the case with bounded $\overline h/\underline h$,
classical empirical process results such as those given in
\citet{van_der_vaart_weak_1996} could be used instead of the
\citet{sakhanenko_convergence_1985} approximation, which we use in our proof
since it covers both cases. In both cases, our results require the (to our
knowledge novel) insight that the approximating Gaussian process is stationary
when indexed by $t=\log h$, and depends only on the kernel used to compute the
estimator. This leads to simple critical values that depend only on the kernel
and bandwidth ratio. In other settings in which snooping does not lead to a
pivotal asymptotic distribution, one could use the general bootstrap approach of
\citet{chernozhukov_gaussian_2013}, which allows one to obtain uniform
confidence bands without obtaining an asymptotic distribution.

In addition to \citet{bickel_global_1973}, numerous authors have used extreme
value limiting theorems for suprema of Gaussian processes to derive confidence
bands for a density or conditional mean function that are uniform in the point
$x$ at which the function is evaluated, with the bandwidth sequence $h=h_n$
fixed \citep[see, among
others][]{johnston_probabilities_1982,hardle_asymptotic_1989,liu_simultaneous_2010}.
In the special case where the Nadaraya-Watson estimator with uniform kernel is
used, extreme value limiting results in \citet{armstrong_multiscale_2016} lead
to confidence bands that are uniform in both $x$ and $h$. In contrast, our case
corresponds to fixing $x$ and requiring that coverage be uniform over $h$.

An important area of application of multiple tests involving tuning parameters
is adaptive inference and testing (in our context, this amounts to constructing
a confidence band for $\theta(0)$ that is close to as small as possible for a
range of smoothness classes for the data generating process). While we do not
consider this problem in this paper, \citet{armstrong_adaptive_2015} uses our
approach to obtain adaptive one-sided confidence intervals under a monotonicity
condition (see Section~\ref{sec:adaptation_example} below). For the problem of
global estimation and uniform confidence bands \citet{gine_confidence_2010}
propose an approach based on a different type of shape restriction. The latter
approach has been generalized in important work by
\citet{chernozhukov_anti-concentration_2014}.

\section{Derivation of the correction in a simple case}\label{simple_derivation_sec}

This section presents a heuristic derivation of the correction in the simple
problem of inference on the conditional mean described in the introduction. To
further simplify the exposition, consider an idealized situation in which
$Y_i=g(X_i)+\sigma\varepsilon_i$, $\sigma^{2}$ is known, $\varepsilon_{i}$ are
i.i.d.\ with variance one, and the regressors are non-random and given by
$X_i=(i+1)/(2n)$ for $i$ odd and $X_i=-i/(2n)$ for $i$ even. In this case, the
Nadaraya-Watson kernel estimator with a uniform kernel,
$k(x)=I(\abs{x}\leq 1/2)$, reduces to
\begin{equation*}
  \hat{\theta}(h)
  =\frac{\sum_{i=1}^{n}k(X_{i}/h)Y_{i}}{\sum_{i=1}^{n}k(X_{i}/h)}=
  \frac{\sum_{i=1}^{nh}Y_{i}}{nh},
\end{equation*}
where, for the second equality and throughout the rest of this example, we
assume that $nh$ is an even integer for notational convenience. Consider first
the problem of constructing a confidence interval for
\begin{equation*}
  \theta(h)= E(\hat\theta(h))= \frac{\sum_{i=1}^{nh} g(X_{i})}{nh}
\end{equation*}
that will have coverage $1-\alpha$ no matter what bandwidth $h$
we pick, so long as $h$ is in some given range $[\underline{h},\overline{h}]$.
For a given bandwidth $h$, a two-sided $t$-statistic is given by
\begin{equation*}
  \sqrt{nh}\frac{\abs{\hat\theta(h)-\theta(h)}}{\sigma}
  =\Abs{\frac{\sum_{i=1}^{nh}\varepsilon_{i}}{\sqrt{nh}}}.
\end{equation*}
In order to guarantee correct coverage, instead of using a critical value equal
to the $1-\alpha/2$ quantile of a normal distribution, we will need to use a
critical value equal to the $1-\alpha$ quantile of the distribution of the
maximal $t$-statistic in the range $[\underline{h},\overline{h}]$. If
$n\underline{h}\to\infty$, we can approximate the partial sum
$n^{-1/2}\sum_{i=1}^{nh}\varepsilon_{i}$ by a Brownian motion $\mathbb{B}(h)$,
so that in large samples, we can approximate the distribution of the maximal
$t$-statistic as
\begin{equation}\label{eq:brownian-motion}
  \sup_{\underline{h}\leq h\leq \overline{h}}
  \sqrt{nh}\frac{\abs{\hat\theta(h)-\theta(h)}}{\sigma}\approx
  \sup_{\underline{h}\leq h\leq \overline{h}}
  \abs{\mathbb{B}(h)/\sqrt{h}}\overset{d}{=}
  \sup_{1\leq h\leq \overline{h}/\underline{h}}
  \abs{\mathbb{B}(h)/\sqrt{h}}.
\end{equation}
Thus, the sampling distribution of the maximal $t$-statistic will in large
samples only depend on the ratio of maximum and minimum bandwidth that we
consider, $\overline{h}/\underline{h}$, and its quantiles can easily be
tabulated (see the columns corresponding to uniform kernel in
Table~\ref{tab:cvs-condensed}). As $\overline h/\underline h\to \infty$, the
recentered distribution of $\sup_{1\leq h\leq
  \overline{h}/\underline{h}}\abs{\mathbb{B}(h)/\sqrt{h}}$, scaled by
$\sqrt{2\log\log (\overline h/\underline h)}$, can be approximated by the
extreme value distribution by the \citet{darling_limit_1956} theorem. Thus, as
$\overline h/\underline h\to\infty$, the critical values increase very slowly,
at the rate $\sqrt{\log\log (\overline h/\underline h)}$.

To guarantee that the resulting confidence interval achieves coverage for
$\theta(0)=g(0)$, the conditional mean at zero, we also need to ensure that the
bias $\abs{\theta(h)-\theta(0)}$ is small relative to the standard error
$\sigma/\sqrt{nh}$, uniformly over $h\in[\underline{h},\overline{h}]$. If the
conditional mean function is twice differentiable with a bounded second
derivative and $\overline{h}/\underline{h}$ is bounded, a sufficient condition
is that $n\overline{h}^{5}\to 0$, that is, we ``undersmooth''.

In the next section, we show that the approximation of the distribution of the
maximal $t$-statistic by a scaled Brownian motion in~\eqref{eq:brownian-motion}
still obtains even if the restrictive assumptions in this section are dropped,
and holds for more general problems than inference for the conditional mean at a
point. The only difference will be that if the kernel is not uniform, then we
need to approximate the distribution of the maximal $t$-statistic by a different
Gaussian process.

\section{General setup and main result}\label{general_result_sec}

This section describes our general setup, states our main asymptotic
distribution result, and derives critical values based on this result. Readers
who are interested only in implementing our procedure can skip to
Section~\ref{practical_implementation_sec}, which explains how to use our tables
to find critical values and implement our procedure. We state our result using
high level conditions, which can be verified for particular applications. For
applications in Section~\ref{applications_sec}, we verify these conditions in
Appendix~\ref{sec:techn-deta-appl}.

We consider a sample $\{X_i,W_i\}_{i=1}^n$, which we assume throughout the paper
to be i.i.d. Here, $X_i$ is a real-valued random variable, and we are interested
in a kernel estimate at a particular point, which we normalize to be $x=0$ for
notational convenience. We consider confidence intervals that are uniform in $h$
over some range $[\underline h_n,\overline h_n]$, where we now make explicit the
dependence of $\underline h_n$ and $\overline h_n$ on $n$. To keep statements of
theoretical results simple, all of our results are pointwise in the underlying
distribution (we show that, for any data generating process satisfying certain
assumptions, coverage of the uniform-in-$h$ CI converges to $1-\alpha$).
However, versions of these results in which coverage is shown to converge to
$1-\alpha$ uniformly in some class of underlying distributions could be derived
from similar arguments, using uniform versions of the bounds in our assumptions.
Our main condition imposes an influence function representation involving a
kernel function.

\begin{assumption}\label{inf_func_assump}
  For some function $\psi(W_i,h)$ and a kernel function $k$ with
  $E\psi(W_i,h)k(X_i/h)=0$ and $\frac{1}{h}var(\psi(W_i,h)k(X_i/h))=1$,
  \begin{align*}
    \frac{\sqrt{nh}(\hat{\theta}(h)-\theta(h))}{\hat\sigma(h)}=\frac{1}{\sqrt{nh}}
    \sum_{i=1}^n \psi(W_i,h)k(X_i/h) +o_P\left(1/\sqrt{\log\log (\overline
        h_n/\underline h_n)}\right)
  \end{align*}
  uniformly over $h\in[\underline h_n,\overline h_n]$.
\end{assumption}

Most of the verification of Assumption~\ref{inf_func_assump} is standard. For
most kernel and local polynomial based estimators, these calculations are
available in the literature, with the only additional step being that the
remainder term must be bounded uniformly over
$h\in[\underline h_n,\overline h_n]$, and with a
$o_P(1/\sqrt{\log\log(\overline{h}_{n}/\underline h_n)})$ rate of approximation.
Supplemental Appendix~\ref{kern_tail_sec} provides some results that can be used
to obtain this uniform bound. For example, in the case of the Nadaraya-Watson
kernel estimator
$\hat\theta(h)=\sum_{i=1}^n Y_{i} k(X_i/h)/\sum_{i=1}^{n}k(X_i/h)$,
Assumption~\ref{inf_func_assump} holds with
$\psi(W_i,h)=(Y_i-\theta(h))/\sqrt{var\left\{[Y_i-\theta(h)]k(X_i/h)/h\right\}}$.
In the local polynomial case, the kernel function $k$ corresponds to the
equivalent kernel, and depends on the order of the polynomial and whether the
estimated conditional quantities are at the boundary (see
Section~\ref{reg_discont_sec} and Supplemental Appendix~\ref{loc_poly_bound_sec}
for details, including a discussion of how our results can be extended to cover
cases in which the boundary of the support of $X_{i}$ is local to $0$).

We also impose some regularity conditions on $k$ and the data generating
process. In applications, these will typically impose smoothness conditions on
the conditional mean and variance of certain variables conditional on $X_{i}$.

\begin{assumption}\label{dgp_kern_assump}
\begin{itemize}
\item[(i)] The kernel function $k$ is symmetric with finite support $[-A,A]$,
  bounded with a bounded, uniformly continuous first derivative on $(0,A)$, and
  satisfies $\int k(u)\, du\ne 0$.
\item[(ii)] $|X_i|$ has a density $f_{|X|}$ with $f_{|X|}(0)>0$,
  $\psi(W_i,h)k(X_i/h)$ is bounded uniformly over $h\le \overline h_n$ with
  $var\left(\psi(W_i,0)\mid |X_i|=0\right)>0$, and, for some deterministic
  function $\ell(h)$ with $\ell(h)\log\log(h^{-1})\to 0$ as $h\to 0$, the
  absolute values of the following expressions are bounded by $\ell(t)$:
  $f_{|X|}(t)-f_{|X|}(0)$,
  $E\left[\psi(W_i,0)\mid
    \abs{X_i}=t\right]-E\left[\psi(W_i,0)\mid\abs{X_i}=0\right]$,
  $(\psi(W_i,t)-\psi(W_i,0))k(X_i/t)$, and\\
  $var(\psi(W_i,0)\mid \abs{X_i}=t) -var(\psi(W_i,0)\mid |X_i|=0)$.
  \item[(iii)] Taken as classes of functions varying over $h>0$, $w\mapsto
    \psi(w,h)$ and $x\mapsto k(x/h)$ have polynomial uniform covering numbers
    (as defined in Appendix~\ref{sec:proof-main-result}).
\end{itemize}
\end{assumption}

Assumption~\ref{dgp_kern_assump} will typically require some smoothness on
$\theta(h)$ as a function of $h$ (since it places smoothness on certain
conditional means, etc.). For inference on $\theta(h)$, rather than $\theta(0)$,
the amount of smoothness required is very mild relative to smoothness conditions
typically imposed when considering bias-variance tradeoffs. In particular,
Assumption~\ref{dgp_kern_assump} only requires that certain quantities are
slightly smoother than $t\mapsto 1/\log\log(t^{-1})$, which does not require
differentiability and holds, e.g., for $t\mapsto t^\gamma$ for any $\gamma>0$.
Thus, our confidence bands for $\theta(h)$ are valid under very mild conditions
on the smoothness of $\theta(h)$, which is useful in settings where the possible
lack of smoothness of $\theta(h)$ leads one to examine $\hat\theta(h)$ across
multiple bandwidths.

Assumption~\ref{inf_func_assump} and~\ref{dgp_kern_assump} are tailored toward
statistics involving conditional means, rather than densities or derivatives of
conditional means and densities (for density estimation, we would have
$\psi(W_i,h)=1$, which is ruled out by the assumptions
$var\left[\psi(W_i,0)\mid |X_i|=0\right]>0$ and $E\psi(W_i,h)k(X_i/h)=0$; for
estimating derivatives of conditional means or densities, the scaling would be
$\sqrt{nh^{1+\nu}}$ where $\nu$ is the order of the derivative). This is done
only for concreteness and ease of notation, and the results can be generalized
to these cases as well by verifying the high level conditions in
Theorems~\ref{ev_limit_thm_general} and~\ref{gaussian_limit_thm} in
Appendix~\ref{sec:proof-main-result}, which is used in proving
Theorem~\ref{highlevel_asym_dist_thm} below. The only requirement is that a
scaled version of $\hat\theta(h)-\theta(h)$ be approximated by the Gaussian
process $\mathbb{H}$ given in Theorem~\ref{highlevel_asym_dist_thm} below. For
estimating derivatives, the kernel $k$ in the process $\mathbb{H}$ will
correspond to the equivalent kernel, and it will depend on the order of the
derivative as well as the order of the local polynomial.

Finally, note that Assumption~\ref{dgp_kern_assump} requires that
$\psi(W_i,h)k(X_i/h)$ be bounded, which typically requires a bounded outcome
variable in applications. We conjecture that this assumption could be relaxed at
the expense of imposing stronger assumptions on $\underline h_n$ and $\overline
h_n$ (see Section~\ref{gauss_approx_sec}).

We are now ready to state the main asymptotic approximation result.

\begin{theorem}\label{highlevel_asym_dist_thm}
  Let $\cvt{1-\alpha}(t,k)$ be the $1-\alpha$ quantile of
  $\sup_{1\le h\le t} \left|\mathbb{H}(h)\right|$, where $\mathbb{H}(h)$ is a
  mean zero Gaussian process with covariance kernel
  $cov\left(\mathbb{H}(h),\mathbb{H}(h')\right) =\frac{\int k(u/h)k(u/h')\,
    du}{\sqrt{hh'}\int k(u)^2\, du} =\sqrt{\frac{h'}{h}}\frac{\int k(u
    (h'/h))k(u)\, du}{\int k(u)^2\, du}$. Suppose that $\underline h_n\to 0$,
  $\overline h_n=\mathcal{O}_P(1)$, and
  $n\underline h_n/[(\log \log n)(\log\log\log n)]^2\to\infty$. Then, under
  Assumptions~\ref{inf_func_assump} and~\ref{dgp_kern_assump},
  \begin{equation*}
        P\left(
      \theta(h)\in \left\{\hat\theta(h)\pm \hat\sigma(h)\cdot
        \cvt{1-\alpha}(\overline h_n/\underline h_n,k)/\sqrt{nh}\right\}
      \text{ all }h\in [\underline h_n\le h\le \overline h_n]
    \right)
    \stackrel{n\to\infty}{\to} 1-\alpha.
  \end{equation*}
  If, in addition, $\overline h_n/\underline h_n\to\infty$, the above display
  also holds with $\cvt{1-\alpha}(\overline h_n/\underline h_n,k)$ replaced by
  \begin{equation}\label{eq:alternative-cv}
    \frac{-\log\left(-\frac{1}{2}\log (1-\alpha)\right)+b(\overline h_n/\underline
      h_n, k)}{\sqrt{2\log\log(\overline h_n/\underline h_n)}} +\sqrt{2\log\log
      (\overline h_n/\underline h_n)},
  \end{equation}
  where $b(t,k)=\log c_1(k)+(1/2)\log\log\log t$ if $k(A)\ne 0$ and
  $b(t,k)=\log c_2(k)$ if $k(A)=0$, with
  $c_1(k)=\frac{Ak(A)^2}{\sqrt{\pi}\int k(u)^2\, du}$ and
  $c_2(k)=\frac{1}{2\pi}\sqrt{\frac{\int \left[k'(u)u+\frac{1}{2}
        k(u)\right]^2\, du}{\int k(u)^2\, du}}$.
%to do: check
\end{theorem}

Theorem~\ref{highlevel_asym_dist_thm} shows coverage of $\theta(h)$. Often,
however, $\theta(0)$ is of interest. We now state a corollary showing coverage
of $\theta(0)$ under an additional condition.

\begin{corollary}\label{theta0_corollary}
  If
  $\sup_{h\in [\underline h_n,\overline
    h_n]}\frac{\sqrt{nh}|\theta(h)-\theta(0)|}{\hat\sigma(h)}=o_P\big(
      (\log\log (\overline h_n/\underline h_n))^{-1/2}\big)$, and the conditions
  of Theorem~\ref{highlevel_asym_dist_thm} hold, then
  \begin{equation*}
    P\left(
      \theta(0)\in \left\{\hat\theta(h)\pm \hat\sigma(h)\cdot
        \cvt{1-\alpha}(\overline h_n/\underline h_n,k)/\sqrt{nh}\right\}
      \text{ all }h\in [\underline h_n\le h\le \overline h_n]
    \right)
    \stackrel{n\to\infty}{\to} 1-\alpha.
\end{equation*}
\end{corollary}
Corollary~\ref{theta0_corollary} uses the additional condition that the bias
$\theta(h)-\theta(0)$ is negligible relative to the standard error
$\hat\sigma(h)/\sqrt{nh}$ uniformly over the range of bandwidths
considered.\footnote{If $\overline h_n/\underline h_n$ is bounded, the condition
  in Corollary~\ref{theta0_corollary} is the same as the condition that is
  needed for pointwise-in-$h$ coverage of conventional CIs that do not adjust
  for snooping. If $\overline h_n/\underline h_n\to\infty$, there is an
  additional $\log\log$ term in the rate at which the bias must decrease, which
  arises for technical reasons. However, this term is small enough that this
  condition is still guaranteed by bias-correction methods such as the one
  proposed by \citet{cct14}.} Typically, it is ensured by bias-correction or
undersmoothing, as long as the smoothness conditions in
Assumption~\ref{dgp_kern_assump} are appropriately
strengthened.\footnote{Alternatively, if a bound $\overline b(h)$ on the bias
  $|\theta(h)-\theta(0)|$ is available, one can allow the bias to be of the same
  order of magnitude as standard deviation by adding and subtracting
  $\hat\sigma(h)\cdot \cvt{1-\alpha}(\overline h_n/\underline
  h_n,k)/\sqrt{nh}+\overline b(h)$. This has the advantage of allowing for
  weaker conditions on the bandwidth sequence, including cases where the
  undersmoothing condition does not hold. See
  \citet{chernozhukov_anti-concentration_2014}, \citet{schennach15} and
  \citet{donoho_statistical_1994} for applications of this idea in different
  settings.} In Section~\ref{reg_discont_sec}, we discuss how, in a regression
discontinuity setting, our approach can be applied with a bias-correction
proposed by \citet{cct14}, and illustrate this approach in empirical examples in
Section~\ref{empirical_sec}. Critical values for constructing one-sided
confidence intervals robust to bandwidth snooping are analogous to the two-sided
case---see Supplemental Appendix~\ref{sec:critical-values} for details.

If the bandwidth choice is a priori tied to a pre-specified set, it is possible
to further tighten the critical values. For example,
\citet{imbens_regression_2008} suggest examining estimates at half and twice the
original bandwidth $\hat{h}$, which yields the set
$\{\hat{h}/2,\hat{h},2\hat{h}\}$. One can extend our approach to obtain critical
values under such discrete snooping. However, such critical values will depend
on the entire discrete set, and will often not be much tighter than
$\cvt{1-\alpha}(\overline{h}_{n}/\underline{h}_{n},k)$ with $\overline{h}_{n}$
and $\underline{h}_{n}$ given by the biggest and smallest bandwidths in the set
(so long as the triangular or Epanechnikov kernel is used). For example, for the
discrete bandwidth set $\{\hat{h}/2,\hat{h},2\hat{h}\}$ the critical value for
the triangular kernel can be shown to equal 2.23, while
$\cvt{1-\alpha}(4,k)=2.26$.

In addition to providing the critical values $\cvt{1-\alpha}$,
Theorem~\ref{highlevel_asym_dist_thm} provides a further approximation
in~\eqref{eq:alternative-cv} to the quantiles of
$\sup_{\underline h_n\le h\le \overline h_n}
\sqrt{nh}\abs{\hat\theta(h)-\theta(h)}/\hat\sigma(h)$ based on an extreme value
limiting distribution, provided that $\overline h_n/\underline h_n\to\infty$. In
the case where $k$ is the uniform kernel, $\psi(W_i,h)$ does not depend on $h$
and $E[\psi(W_i,h)|X_i=x]=0$ and $var[\psi(W_i,h)|X_i=x]=1$ for all $x$, the
latter result reduces to a well-known theorem of \citet{darling_limit_1956}
\citep[see also][]{einmahl_darling-erdos_1989}. For the case where $k$ is not
the uniform kernel, or where $\psi$ depends on $h$, this result is, to our
knowledge, new. We do not recommend using the critical value
in~\eqref{eq:alternative-cv} value in practice, as critical values based on
extreme value results have been known to perform poorly in related settings
\citep[see][]{hall91}. Instead we recommend using the critical value
$\cvt{1-\alpha}(\overline h_n/\underline h_n,k)$, which does not suffer from
these issues because it is based directly on the Gaussian process approximation,
and it remains valid even for fixed $\overline h_n/\underline h_{n}$ (see
Figure~\ref{fig:extreme-value} in the supplemental appendix for a comparison of
these critical values). Thus, we report only this critical value in
Table~\ref{tab:cvs-condensed} below.

The main practical value of the approximation in~\eqref{eq:alternative-cv} is
that it demonstrates that critical value grows very slowly with
$\overline h_n/\underline h_n$, at rate
$\sqrt{\log\log(\overline h_n/\underline h_n)}$, so that the cost of examining a
wider range of bandwidths relative to examining a moderate range is rather
small. Indeed, while using our correction is important for maintaining correct
coverage, as can be seen from Table~\ref{tab:cvs-condensed}, once
$\overline h_n/\underline h_{n}$ is above 5, widening the range of bandwidths
that one examines increases the critical value by only a small amount.

To outline how Theorem~\ref{highlevel_asym_dist_thm} obtains, consider again the
problem of estimating a nonparametric mean at a point described in the
introduction. Here the influence function is given by
$\psi(W_i,h)k(X_i/h)$ where
$\psi(W_i,h)=(Y_i-\theta(h))/\sqrt{var\left\{[Y_i-\theta(h)]k(X_i/h)/h\right\}}$
so that, for small $h$, we can approximate the t-statistic as
\begin{equation*}
\frac{\sqrt{nh}(\hat\theta(h)-\theta(h))}{\hat\sigma(h)}
\approx \frac{\sum_{i=1}^n [Y_i-\theta(h)]k(X_i/h)}{
  \sqrt{n \cdot var\left\{[Y_i-\theta(h)]k(X_i/h)\right\}}}.
\end{equation*}
Thus, we expect that the supremum of the absolute value of this display over $h\in [\underline h,\overline h]$ is approximated by $\sup_{h\in[\underline h,\overline h]} \left|\mathbb{H}_n(h)\right|$ where
$\mathbb{H}_n(h)$ is a Gaussian process with covariance function
\begin{equation}\label{H_cov_approx_eq}
cov\left(\mathbb{H}_n(h),\mathbb{H}_n(h')\right)
=\frac{cov\left\{[Y_i-\theta(h)]k(X_i/h),[Y_i-\theta(h')]k(X_i/h')\right\}}{\sqrt{var\left\{[Y_i-\theta(h)]k(X_i/h)\right\}}\sqrt{var\left\{[Y_i-\theta(h')]k(X_i/h')\right\}}}.
\end{equation}
The conditions in Assumption~\ref{dgp_kern_assump} ensure that $E(Y_i|X_i=x)$,
$var(Y_i|X_i=x)$ and the density $f_X(x)$ of $X_i$ do not vary too much as $x\to
0$, so that, for $h$ and $h'$ close to zero,
\begin{multline*}
cov\left\{[Y_i-\theta(h)]k(X_i/h),[Y_i-\theta(h')]k(X_i/h')\right\}
\approx E\left\{[Y_i-E(Y_i|X_i)]^2k(X_i/h)k(X_i/h')\right\}  \\
=\int var(Y_i|X_i=x)k(x/h)k(x/h')f_X(x)\, dx
\approx var(Y_i|X_i=0)f_X(0)\int k(x/h)k(x/h')\, dx  \\
= var(Y_i|X_i=0)f_X(0)h'\int k\left(u(h'/h)\right)k(u)\, du.
\end{multline*}
Using this approximation for the variance terms in the denominator of (\ref{H_cov_approx_eq}) as well as the covariance in the numerator gives the approximation
\begin{equation*}
  cov\left(\mathbb{H}_n(h),\mathbb{H}_n(h')\right)
  \approx \frac{h'\int k\left(u(h'/h)\right)k(u)\, dx}
  {\sqrt{h'\int k(u)^2\, dx}\sqrt{h\int k(u)^2\, dx}}
  =\frac{\sqrt{h'/h}\int k\left(u(h'/h)\right)k(u)\, dx}
  {\int k(u)^2\, dx}.
\end{equation*}
Thus, letting $\mathbb{H}(h)$ be the Gaussian process with the covariance on the
right hand side of the above display, we expect that the distribution of
$\sup_{h\in[\underline h,\overline
  h]}\frac{\sqrt{nh}\abs{\hat\theta(h)-\theta(h)}}{\hat\sigma(h)}$ is
approximated by the distribution of
$\sup_{h\in [\underline h,\overline h]}\left|\mathbb{H}(h)\right|$. Since the
covariance kernel given above depends only on $h'/h$,
$\sup_{h\in [\underline h,\overline h]}\left|\mathbb{H}(h)\right|$ has the same
distribution as
$\sup_{h\in [\underline h,\overline h]}\left|\mathbb{H}(h/\underline h)\right|
=\sup_{h\in [1,\overline h/\underline h]}\left|\mathbb{H}(h)\right|$. As it
turns out, this approximation will work under relatively mild conditions so long
as $\underline h\to 0$ even if $\overline h$ does not approach zero, because, in
this case, the bandwidth that achieves the supremum will still converge in
probability to zero, yielding the first part of the theorem. For the second part
of the theorem, we show that
$\sup_{h\in[\underline h,\overline
  h]}\frac{\sqrt{nh}|\hat\theta(h)-\theta(h)|}{\hat\sigma(h)}$ increases
proportionally to $\sqrt{2\log\log(\overline h/\underline h)}$, and that a
further scaling by $\sqrt{2\log\log(\overline h/\underline h)}$ gives an extreme
value limiting distribution. To further understand the intuition for this, note
that $\mathbb{H}(h)$ is stationary when indexed by $t=\log h$ (since the
covariance at $h=e^t$ and $h'=e^{t'}$ depends only on $h'/h=e^{t'-t}$), so,
setting $T=\log(\overline h/\underline h)$, we expect the supremum over
$[\log 1,\log (\overline h/\underline h)]=[0,T]$ to follow an extreme value
limiting with scaling
$\sqrt{2\log T}=\sqrt{2\log\log (\overline h/\underline h)}$ so long as
dependence dies away quickly enough with $T$, following classical results
\citep[see][for a textbook exposition of these
results]{leadbetter_extremes_1983}.

\subsection{Practical implementation}\label{practical_implementation_sec}

For convenience, this section gives step-by-step instructions for finding the
appropriate critical value in our tables and implementing our procedure. We also
provide some analysis of the magnitudes involved in the correction and the
undercoverage that can occur from searching over multiple bandwidths without
implementing our correction.

Table~\ref{tab:cvs-condensed} gives the critical values
$\cvt{1-\alpha}(\overline h_n/\underline h_n,k)$ for several kernel functions
$k$, $\alpha=0.05$ and selected values of $\overline h_n/\underline h_n$.
Critical values for $\alpha=0.01$ and $\alpha=0.10$, are given in
Table~\ref{tab:cvs-ts} in the supplemental appendix. The
critical values can also be obtained using our R package \texttt{BWSnooping},
which can be downloaded from \url{https://github.com/kolesarm/BWSnooping}. For
local polynomial estimators, the critical value depends on the order of the
local polynomial, as well as whether the point of interest is at the boundary
(including the case of regression discontinuity) or in the interior of the
support of $X_i$. We report values for Nadaraya-Watson (local constant) and
local linear estimators. Note that the critical values for Nadaraya-Watson
kernel regression are the same whether or not the point of interest is in the
interior or at the boundary. For local linear regression in the interior, the
equivalent kernel is the same as the original kernel, and therefore the critical
value is the same as that for Nadaraya-Watson kernel regression. For local
linear regression at the boundary, including inference in regression
discontinuity designs, the critical value is different because the equivalent
kernel is different (see Supplemental Appendix~\ref{loc_poly_bound_sec} for
details).

Using these tables, our procedure can be described in the
following steps:
\begin{enumerate}
\item Compute an estimate $\hat\sigma(h)$ of the standard deviation of
  $\sqrt{nh}(\hat\theta(h)-\theta(h))$, where $\hat\theta(h)$ is a kernel-based
  estimate.

\item Let $\underline h$ and $\overline h$ be the smallest and largest values of
  the bandwidth $h$ considered, respectively, and let $\alpha$ be the nominal
  level. Appropriate choice of $\underline h$ and $\overline h$ will depend on
  the application; Section~\ref{applications_sec} discusses this choice for the
  applications we consider. Look up the critical value
  $\cvt{1-\alpha}(\overline h_n/\underline h_n,k)$ in
  Table~\ref{tab:cvs-condensed} for $\alpha=0.05$, or in Table~\ref{tab:cvs-ts}
  for $\alpha=0.01$ and $\alpha=0.10$.
\item Report uniform confidence band
  $\left\{\hat\theta(h)\pm (\hat\sigma(h)/\sqrt{nh}) \cvt{1-\alpha} (\overline
    h_n/\underline h_n,k) \mid h\in[\underline h,\overline h]\right\}$ for
  $\theta(h)$. Or, report
  $\hat\theta(\hat{h})\pm (\hat\sigma(\hat{h})/\sqrt{n\hat{h}}) \cvt{1-\alpha}
  (\overline h_n/\underline h_n,k)$ for a chosen bandwidth $\hat{h}$ as a
  confidence interval for $\theta(\hat{h})$ that takes into account ``snooping''
  over $h\in[\underline h,\overline h]$.
\end{enumerate}

It is common practice to report an estimate $\hat\theta(\hat{h})$ and a standard
error $se(\hat{h})\equiv \hat\sigma(\hat{h})/\sqrt{n\hat{h}}$ for a value of $\hat{h}$
chosen by the researcher. If one suspects that results reported in this way were
obtained after examining the results for $h$ in some set
$[\underline h,\overline h]$ (say, by looking for the value of $h$ for which the
corresponding test of $H_0:\theta(h)=0$ has the smallest $p$-value), one can
compute a ``bandwidth snooping adjusted'' confidence interval as described in
step 3, so long as the kernel function is reported (as well as the order of the
local polynomial).

Figure~\ref{fig:cvs-twosided} plots our critical values as a function of
$\overline h/\underline h$ for $1-\alpha=0.95$. By construction, the critical
value is given by the standard normal quantile $1.96$ when
$\overline h/\underline h=1$, and increases from there. For the kernels and
range of $\overline h/\underline h$ considered, the correction typically amounts
to replacing the standard normal quantile $1.96$ with a number between $2.2$ and
$2.8$, depending on the kernel and range of bandwidths considered.

Our results can also be used to quantify undercoverage from entertaining
multiple bandwidths without using our correction.
Figure~\ref{fig:coverage-twosided} plots the true uniform asymptotic coverage of
a nominal 95\% confidence interval over a range $[\underline h,\overline h]$ for
different values of $\overline h/\underline h$. This amounts to finding
$1-\tilde \alpha$ such that the pointwise critical value $1.96$ is equal to
$\cvt{1-\tilde{\alpha}}(\overline h_n/\underline h_n,k)$. For
$\overline h/\underline h$ below 10, the true coverage is typically somewhere
between $70\%$ and $90\%$, depending on the kernel and the exact value of
$\overline h/\underline h$.

\section{Applications}\label{applications_sec}

This section applies the main results from Section~\ref{general_result_sec} to
three econometric models. In the first example, $\theta(0)$ is of primary
interest, while in the other examples, $\theta(h)$ is an interesting economic
object in its own right. Technical details for this section are relegated to
Appendix~\ref{sec:techn-deta-appl}.

\subsection{Regression discontinuity}\label{reg_discont_sec}

We are interested in a regression discontinuity (RD) parameter, where the
discontinuity point is normalized to $x=0$ for convenience of notation. We
consider both ``sharp'' and ``fuzzy'' regression discontinuity. Using arguments
in the discussion preceding Theorem~\ref{highlevel_asym_dist_thm}, the results
in this section could also be generalized to cover ``kink'' designs
\citep{clpw15}, where the focus is on estimating derivatives of conditional
means at a point---in the interest of space, we do not pursue this extension
here.

For fuzzy RD, we observe $\{(X_i,D_i,Y_i)\}_{i=1}^n$, and the parameter of
interest is given by
$\theta(0)=\frac{\lim_{x\downarrow 0} E(Y_i|X_i=x)-\lim_{x\uparrow 0}
  E(Y_i|X_i=x)}{\lim_{x\downarrow 0} E(D_i|X_i=x)-\lim_{x\uparrow 0}
  E(D_i|X_i=x)}$. For sharp RD, we observe $\{(X_i,Y_i)\}_{i=1}^n$, and the
parameter of interest is given by
$\theta(0)=\lim_{x\downarrow 0} E(Y_i|X_i=x)-\lim_{x\uparrow 0} E(Y_i|X_i=x)$.
For ease of exposition, we focus on the commonly used local linear estimator
\citep[see, e.g.,~][]{porter_estimation_2003}.\footnote{We cover the extension to
  local polynomial regression of higher order in
  Appendix~\ref{loc_poly_bound_sec}.} Given a kernel function $k^*$ and a
bandwidth $h$, let $\hat\alpha_{\ell,Y}(h)$ and $\hat\beta_{\ell,Y}(h)$ denote
the intercept and slope from a weighted linear regression of $Y_{i}$ on $X_{i}$
in the subsample with $X_{i}<0$, weighted by $k(X_{i}/h)$. That is,
$\hat\alpha_{\ell,Y}(h)$ and $\hat\beta_{\ell,Y}(h)$ minimize
\begin{equation*}
\sum_{i=1}^n \left(Y_i-\alpha_{\ell,Y}-\beta_{\ell,Y} X_i\right)^2I(X_i< 0)k^*(X_i/h).
\end{equation*}
Let $(\hat\alpha_{u,Y}(h),\hat\beta_{u,Y}(h))$ denote the regression
coefficients from a regression in the subsample with $X_{i}\geq 0$. For the
fuzzy case, define $(\hat\alpha_{\ell,D}(h),\hat\beta_{\ell,D}(h))$ and
$(\hat\alpha_{u,D}(h),\hat\beta_{u,D}(h))$ analogously with $D_i$ replacing
$Y_i$. The sharp RD local linear estimator is then given by
$\hat\theta(h)=\hat\alpha_{u,Y}(h)-\hat\alpha_{\ell,Y}(h)$. The fuzzy RD
estimator is given by
$\hat\theta(h)=\frac{\hat\alpha_{u,Y}(h)-\hat\alpha_{\ell,Y}(h)}
{\hat\alpha_{u,D}(h)-\hat\alpha_{\ell,D}(h)}$.

We define $\theta(h)$ as the statistic constructed from the population versions
of these estimating equations, which leads to $\hat\theta(h)$ being
approximately unbiased for $\theta(h)$. Let
$(\alpha_{\ell,Y}(h),\beta_{\ell,Y}(h))$ minimize
\begin{equation*}
E \left(Y_i-\alpha_{\ell,Y}-\beta_{\ell,Y} X_i\right)^2I(X_i< 0)k^*(X_i/h),
\end{equation*}
and let $(\alpha_{u,Y}(h),\beta_{u,Y}(h))$,
$(\alpha_{\ell,D}(h),\beta_{\ell,D}(h))$ and $(\alpha_{u,D}(h),\beta_{u,D}(h))$
be defined analogously. We define
$\theta(h)=\frac{\alpha_{u,Y}(h)-\alpha_{\ell,Y}(h)}
{\alpha_{u,D}(h)-\alpha_{\ell,D}(h)}$ for fuzzy RD, and
$\theta(h)=\alpha_{u,Y}(h)-\alpha_{\ell,Y}(h)$ for sharp RD\@. Under appropriate
smoothness conditions, $\theta(h)$ will converge to $\theta(0)$ as $h\to 0$.

Theorem~\ref{reg_disc_thm} in Appendix~\ref{sec:techn-deta-appl} shows that
under appropriate conditions, Theorem~\ref{highlevel_asym_dist_thm} applies with
$k(u)$ given by the equivalent kernel $k(u)=(\mu_{k^*,2}-\mu_{k^*,1}|u|)k^*(u)$,
where $\mu_{k^*,j}=\int_{u=0}^\infty u^j k^*(u)$ for $j=1,2$ (rather than the
original kernel $k^{*}$). For convenience, we report critical values for
$k(u)=(\mu_{k^*,2}-\mu_{k^*,1}|u|)k^*(u)$ for some common choices of $k^*$ in
Table~\ref{tab:cvs-condensed} for $\alpha=0.05$ and Table~\ref{tab:cvs-ts} in
the supplemental appendix for $\alpha=0.01$ and $\alpha=0.10$.

In most RD applications, $\theta(0)$, rather than $\theta(h)$, is of primary
interest. Let $h^{*}_{ll}$ denote a bandwidth that minimizes the mean-squared
error $E[(\hat{\theta}(h)-\theta(0))^{2}]$ of the local linear estimator (or an
asymptotic approximation of it), such as the \citet{imbens_optimal_2012}
bandwidth selector. Then, as is well-known, the bias of
$\hat{\theta}({h}^{*}_{ll})$ will not be asymptotically negligible, and
confidence intervals around $\hat{\theta}({h}^{*}_{ll})$ will have poor coverage
of $\theta(0)$, even without any snooping.

In an important paper, \citet[CCT]{cct14} show that one can address this issue
by recentering the confidence interval by subtracting an estimate of the
asymptotic bias, and rescaling it to account for the additional noise induced by
the bias estimation. CCT show that the remaining bias is asymptotically
negligible so that this alternative confidence interval will achieve proper
coverage of $\theta(0)$, provided the conditional mean functions are smooth
enough on each side of the cutoff. If the pilot bandwidth used to estimate the
bias equals ${h}^{*}_{ll}$, this procedure is equivalent to constructing the
usual confidence interval around a local quadratic estimator with bandwidth
${h}^{*}_{ll}$. Since the MSE optimal bandwidth for local quadratic regression
is of larger order than the optimal bandwidth for local linear regression, this
method of constructing confidence intervals can also be viewed as a particular
undersmoothing procedure. Consequently, if one uses a local quadratic estimator
and $\overline{h}=\mathcal{O}({h}^{*}_{ll})$, Corollary~\ref{theta0_corollary}
applies, so that our adjusted confidence intervals will also achieve correct
coverage of the RD parameter $\theta(0)$. We apply this method in two empirical
examples in Section~\ref{empirical_sec}, and investigate its finite-sample
properties in a Monte Carlo exercise in Supplemental Appendix~\ref{mc_sec}.

In the remainder of this subsection, we discuss two cases in which our
computing our adjusted confidence interval is relevant. We also discuss the
choice of $\overline{h}$ and $\underline{h}$.

\subsubsection{Sensitivity Analysis}\label{section:sensitivity-analysis-example}
A researcher implements the CCT bias-correction method by calculating the local
quadratic estimator of the sharp RD parameter
$\theta(0)=\lim_{x\downarrow 0} E(Y_i|X_i=x)-\lim_{x\uparrow 0} E(Y_i|X_i=x)$ at
the bandwidth $h={h}^{*}_{ll}$. To check the robustness of the results, the
researcher also evaluates the estimator at a bandwidth
$h_{\text{smaller}}<{h}^{*}_{ll}$. Suppose that the CI evaluated at
${h}^{*}_{ll}$ contains zero, while the CI evaluated at $h_{\text{smaller}}$
does not (in any given sample, this may happen even if both estimators are
exactly unbiased). Arguing that the bias of the estimator at
$h_{\text{smaller}}$ is negligible under weaker assumptions, the researcher may
be tempted to conclude that $\theta(0)=E(Y_i|X_i=0)$ is nonzero, and that the
conclusions of this hypothesis test are valid under even weaker assumptions than
the original assumptions needed for validity of the CCT confidence interval.
Unfortunately, this is not true for the actual hypothesis test that the
researcher has performed (looking at both ${h}^{*}_{ll}$ and
$h_{\text{smaller}}$), since the $\alpha$ probability of type I error has
already been ``used up'' on the test based on ${h}^{*}_{ll}$. By replacing
$z_{1-\alpha/2}$ with the critical value
$\cvt{1-\alpha}({h}^{*}_{ll}/h_{\text{smaller}},k)$, the researcher can conclude
that $\theta(0)\ne 0$ under the original assumptions, so long as at least one of
the two confidence intervals does not contain zero.
Appendix~\ref{sensitivity_analysis_sec} provides further discussion of cases in
which the uniform-in-$h$ confidence bands can be useful in sensitivity analysis.

\subsubsection{Adaptive inference}\label{sec:adaptation_example}
Suppose that it is known from the economics of the problem that the conditional
mean function $E(Y_i|X_i=x)$ is weakly decreasing. Then a Nadaraya-Watson (local
constant) estimator $\hat{\theta}_{NW}(h)$ of the sharp RD parameter must be
biased downward for any bandwidth $h$. Because any downward bias will make the
one-sided confidence interval
$[\hat\theta_{NW}(h)-z_{1-\alpha}\hat\sigma_{NW}(h)/\sqrt{nh},\infty)$ only more
conservative, it is asymptotically valid for any $h$ regardless of how fast
$h\to 0$ with $n$ (even if $h$ does not decrease with $n$ at all), so long as
$nh\to\infty$ so that a central limit theorem applies to $\hat\theta_{NW}(h)$.

One may wish to use this fact to ``snoop'' by reporting the most favorable
confidence interval, namely,
$[\sup_{h\in [\underline h,\overline
  h]}(\hat\theta_{NW}(h)-z_{1-\alpha}\hat\sigma_{NW}(h)/\sqrt{nh}),\infty)$ for
some $[\underline h,\overline h]$. Because it involves entertaining multiple
bandwidths, this is not a valid confidence interval. Replacing $z_{1-\alpha}$
with one-sided version of our critical value, $\cvo{1-\alpha}$ (see Supplemental
Appendix~\ref{sec:critical-values}), leads to a confidence interval
$[\sup_{h\in [\underline h,\overline h]}(\hat\theta(h)-
\cvo{1-\alpha}\hat\sigma(h)/\sqrt{nh}),\infty)$, which will have correct
asymptotic coverage.

In fact, this confidence interval enjoys an optimality property of being
adaptive to certain levels of smoothness of the conditional mean, that is, it is
almost as tight as the tightest confidence interval if the smoothness of the
conditional mean was known. More formally, suppose $E(Y_i|X_i=x)$ approaches
$E(Y_i|X_i=0)$ at the rate $x^\beta$ for some $\beta\in (0,1]$. Then, so long as
$\overline h\to 0$ slowly enough and $\underline h\to 0$ quickly enough, the
lower endpoint of this confidence interval will shrink toward
$\theta(0)=E(Y_i|X_i=0)$ at the same rate as a confidence interval constructed
using prior knowledge of $\beta$, up to a term involving $\log\log n$.
Furthermore, no confidence region can achieve this rate simultaneously for
$\beta$ in a nontrivial interval without giving up this $\log\log n$ term. Since
the $\log\log n$ term comes from the multiple bandwidth adjustment in our
critical values, this shows that such an adjustment (or something like it), is
necessary for this form of adaptation. In particular, one cannot estimate the
optimal bandwidth accurately enough to do away with our correction
\citep[see][for details]{armstrong_adaptive_2015}.

\subsubsection{Choice of $\underline{h}$ and
  $\overline{h}$}\label{sec:choice-bw-limits-rd}

Let us discuss some general considerations for a choice of the smallest and
largest bandwidth in the context of sensitivity analysis in RD (for adaptive
inference under monotonicity, the appropriate choice depends on the range of
smoothness levels of the conditional mean, see \citet{armstrong_adaptive_2015}
for details). A conservative approach is to set $\underline h$ to the smallest
value such that enough observations are used for the central limit theorem to
give a good approximation (say, 50 effective observations). If one is interested
in inference on $\theta(0)$, using the CCT bias-correction discussed above,
$\overline{h}$ can be set to be of the same order as ${h}^{*}_{ll}$, such as
$\overline{h}=3{h}^{*}_{ll}/2$ or $\overline{h}=2{h}^{*}_{ll}$. Alternatively,
one can take an even more conservative approach of setting $\overline h$ to
include all of the data, so long as one keeps in mind that CIs with $h$ much
larger than ${h}^{*}_{ll}$ may not contain $\theta(0)$ due to bias. Given that
the critical value increases slowly with $\overline h/\underline h$ for moderate
to large values of $\overline h/\underline h$, the resulting critical value will
not be much larger than under a more moderate choice of $\underline h$ and
$\overline h$.

Implementations of the MSE optimal bandwidth such as those in
\citet{imbens_optimal_2012} and \citet{cct14} typically yield a random
bandwidth, so if $\overline{h}$ depends on it, it will also be random. While we
state our results for nonrandom $[\underline h,\overline h]$, our results can be
extended to this case without the need for additional corrections so long as
$\overline h/\overline{h}^{*}\stackrel{p}{\to} 1$ and
$\underline h/\underline{h}^{*}\stackrel{p}{\to} 1$ for some nonrandom sequences
$\overline h^{*}$ and $\underline h^{*}$ satisfying our conditions.
\citet{imbens_optimal_2012} exhibit a nonrandom sequence $h^*_{IK}$ such that
their bandwidth selector $\hat{h}^{*}_{IK}$ satisfies
$\hat{h}^{*}_{IK}/h^*_{IK}\stackrel{p}{\to} 1$ under certain conditions, so that
one can take, for example, $\overline h=\hat{h}^{*}_{IK}$ or
$\overline h=2 \hat{h}^{*}_{IK}$ under these conditions.\footnote{While this
  argument applies to certain data-dependent bandwidth selectors, one cannot use
  arbitrary data-dependent rules to choose $[\overline h,\underline h]$ in our
  setup. As an extreme example, choosing $\overline h=\underline h$ to minimize
  the $p$-value for a particular hypothesis (and then arguing that a snooping
  correction is not needed since $\overline h=\underline h$) is clearly not
  compatible with our setup. Rather, one would have to define
  $[\underline h,\overline h]$ to be the range over which the $p$-value was
  minimized.} Note, however, that bandwidth selectors such as $\hat{h}^{*}_{IK}$
can, in practice, exhibit substantial variability and dependence on tuning
parameters chosen by the user.\footnote{For example, as
  \citet{imbens_optimal_2012} point out, the bandwidth that is ``optimal''
  according to their definition is infinite in certain cases. Their procedure
  uses tuning parameters to ensure that the selected bandwidth goes to zero in
  these cases, resulting in a bandwidth sequence that depends on tuning
  parameters asymptotically. This can lead to substantial differences between
  different implementations of their approach such as the original
  implementation in \citet{imbens_optimal_2012} and the implementation in
  \citet{cct14}.} To the extent that a data-dependent bandwidth is highly
variable in finite samples, or can be ``gamed'' using tuning parameters, it is
safer to make use a conservative choice of $[\underline h,\overline h]$ that
contains the data-dependent bandwidth with probability one regardless of the
tuning parameters.

\subsection{Trimmed average treatment effects under unconfoundedness}\label{te_unconf_sec}

We extend our setting to obtain uniform confidence bands for average treatment
effects (ATEs) on certain subpopulations under unconfoundedness. Here the
adjustment is slightly different, but it can still be computed using our tables
along with quantities that are routinely reported in applied research.

Let $Y_{i}(0)$ and $Y_{i}(1)$ denote the potential outcomes associated with a
binary treatment $D_{i}$ that is as good as randomly assigned, conditional on
covariates $X_{i}$, so that $E(Y_i(d)|X_i,D_i)=E(Y_i(d)|X_i)$. We observe and
i.i.d.\ sample $\{(X_i,D_i,Y_i)\}_{i=1}^n$, where
$Y_i=Y_{i}(1)D_{i}+Y_{i}(0)(1-D_{i})$ denotes the observed outcome. Let
$\tau(x)=E(Y_i(1)-Y_i(0)\mid X_i=x)=\mu_1(x)-\mu_0(x)$ denote the average
treatment effect for individuals with $X_{i}=x$, where
$\mu_d(x)=E(Y_i|X_i=x,D_i=d)$. Let $e(x)=P(D_i=1|X_i=x)$ denote the propensity
score.

Typically, we are interested in the ATE for the whole population,
$\theta(0)=E[Y_{i}(1)-Y_{i}(0)]$. However, since effects for individuals with
propensity score $e(X_{i})$ close to zero or one cannot be estimated very
precisely, in samples with limited overlap (i.e.~in which the number of such
individuals is high), estimates of the ATE $\theta(0)$ will be too noisy. To
deal with this problem, it is common in empirical applications to trim the
sample by discarding observations with extreme values of the propensity
score.\footnote{\label{fn:trimmed-ate-snooping-examples}For prominent examples,
  see \citet{HeIcTo97}, \citet{GGS05}, or \citet{BaGB15}.} Doing so, however,
changes the estimand. In particular, if the sample is restricted to individuals
with moderate values of the propensity score,
$\mathcal{X}_h=\{X_{i}\colon h\le e(X_i)\le 1-h\}$ for some $0\le h< 1/2$, then,
as discussed in \citet{crump_dealing_2009}, the estimand changes from
$\theta(0)$ to
\begin{equation*}
  \theta(h)
  =E(Y_i(1)-Y_i(0)|X_i\in \mathcal{X}_h)=E(\tau(X_i)|X_i\in\mathcal{X}_h),
\end{equation*}
One therefore faces the trade-off between increasing $h$ from $0$ to increase
the precision of the estimator at the cost of making the estimand $\theta(h)$
arguably less interesting. See \citet{crump_dealing_2009},
\citet{hill_robust_2013} and \citet{khan_irregular_2010} for a detailed
discussion of these issues. \citet{crump_dealing_2009} propose a rule for
picking the trimming parameter $h$ that minimizes the variance of the resulting
estimator. In practice, one may want to resolve this trade-off in other ways.
Our approach of reporting a uniform confidence band allows the researcher to
avoid the issue of which trimmed estimate to report and simply report a range of
estimates. With the reported confidence band for $\theta(h)$, the reader can
pick their preferred trimming value, assess treatment effect heterogeneity by
examining how $\theta(h)$ varies with $h$, or obtain a confidence interval for
$\theta(0)$ based on the reader's own beliefs about the smoothness of
$\theta(h)$.

When forming this confidence band, one can choose a trimming range
$[\underline h,\overline h]$ that is wide enough to include different
suggestions in the literature about the appropriate amount of trimming, such as
$[0,0.1]$ or $[0,0.2]$. This allows the reader to pick their preferred trimming
amount as well as assess the sensitivity of the results to the amount of
trimming.

To describe the adjustment to critical values in this setting, let
$\hat{\theta}(h)$ be an efficient estimator of $\theta(h)$ (in the sense of
satisfying condition~\eqref{cate_inf_func_eq} in
Appendix~\ref{sec:techn-deta-appl}), and let $se(h)$ denote its standard error.
Let $N(h)$ be the number of untrimmed observations for a given $h$ (i.e.\ number
of observations $i$ such that $X_i\in \mathcal{X}_{h}$). In contrast to the
previous applications, assume that $\underline h$ and $\overline h$ are fixed.
If $e(X_i)$ is close to zero or one with high probability, the variance bound
for the ATE, $\theta(0)$, may be infinite, and a sequence of trimming points
$h_n\to 0$ can be used to obtain estimators that converge to the ATE at a slower
than root-$n$ rate \citep[see][]{khan_irregular_2010}. We expect that our
results can be extended to this case under appropriate regularity conditions,
but we leave this question for future research. We form our uniform confidence
band as
\begin{equation}\label{ate_unconf_t_hat_eq}
  \Big\{\hat\theta(h)\pm \cvt{1-\alpha}(\hat t,k_{\text{uniform}})
  \cdot se(h)\Big|
  h\in[\underline h,\overline h]\Big\},
  \quad\text{where}\quad
  \hat t
  =\frac{se(\underline h)^2N(\underline h)^2}{se(\overline h)^2N(\overline h)^2},
\end{equation}
and $k_{\text{uniform}}$ denotes the uniform kernel. In
Theorem~\ref{ate_unconf_thm} in Appendix~\ref{sec:techn-deta-appl}, we show that
this confidence band is asymptotically valid under appropriate regularity
conditions. The critical value given above comes from an approximation by a
scaled Brownian motion where the ``effective sample size'' is proportional to a
quantity that can be estimated by $se(h)^2N(h)^2$. See proof of
Theorem~\ref{ate_unconf_thm} in Supplemental Appendix~\ref{unconf_proof_sec} for
details.

\subsection{LATEs for different sets of compliers}\label{late_sec}

We observe $(Z_i,D_i,Y_i)$ where $Z_i$ is an exogenous instrument shifting a
binary treatment variable $D_{i}$, and $Y_{i}$ is an outcome variable. Let
$[\underline z,\overline z]$ be the support of $Z_i$, and assume, for
simplicity, that $\underline{z}$ and $\overline z$ are finite (this does not
involve much loss in generality, since $Z_i$ can always be transformed to the
unit interval by redefining $Z_i$ as its percentile rank). Suppose that
$P(D_i=1\mid Z_i=z)$ is increasing in $z$, and for
$h\leq (\overline z-\underline z)/2$ define
\begin{equation*}
  \theta(h)
  =\frac{E(Y_i\mid Z_i\in [\overline{z}-h,\overline{z}])
    -E(Y_i|Z_i\in [\underline{z},\underline{z}+h])}
  {P(D_i=1\mid Z_i\in [\overline{z}-h,\overline{z}])
    -P(D_i=1|Z_i\in[\underline{z},\underline{z}+h])}.
\end{equation*}
Under certain exogeneity and monotonicity assumptions, $\theta(h)$ gives the
average effect for the subpopulation of ``compliers'', individuals who change
their treatment status if their instrument shifts from
$Z_{i}\in[\underline{z},\underline{z}+h]$ to
$Z_{i}\in [\overline{z}-h,\overline{z}]$. In the literature, this is called the
``local average treatment effect'', or LATE
\citep[see][]{imbens_identification_1994,heckman_structural_2005,heckman_understanding_2006}.
It can be estimated with the sample analogue
\begin{align*}
\hat\theta(h)=\frac{\frac{1}{\#\{Z_i\in[\overline z-h,\overline z]\}}\sum_{Z_i\in[\overline z-h,\overline z]}Y_i-\frac{1}{\#\{Z_i\in[\underline z,\underline z+h]\}}\sum_{Z_i\in[\underline z,\underline z+h]}Y_i}
  {\frac{1}{\#\{Z_i\in[\overline z-h,\overline z]\}}\sum_{Z_i\in[\overline z-h,\overline z]}D_i-\frac{1}{\#\{Z_i\in[\underline z,\underline z+h]\}}\sum_{Z_i\in[\underline z,\underline z+h]}D_i},
\end{align*}
where $\#\mathcal{A}$ denotes the number of elements in a set $\mathcal{A}$. The
estimator $\hat{\theta}(h)$ is numerically identical to the instrumental
variables estimator for $\beta$ in the equation
$Y_i=\alpha+D_i\beta+\varepsilon$, where the sample is restricted to
observations with
$Z_i\in [\underline z,\underline z+h]\cup [\overline z-h,\overline z]$ and the
instrument is $I(Z_i\ge \overline z-h)$. Let $\hat\sigma^2(h)/h$ be the robust
variance estimate for $\sqrt{n}(\hat\beta-\beta)$ from this IV regression, so
that $\hat\sigma(h)/\sqrt{nh}=se(h)$ is the standard error for $\hat\theta(h)$.

The parameter $\theta(0)=\lim_{h\to 0}\theta(0)$ is typically of particular
interest since it corresponds to the LATE for the largest subpopulation for
which the LATE is identified
\citep[see][]{frolich_nonparametric_2007,heckman_structural_2005,heckman_understanding_2006}.
In finite samples one faces a trade-off similar to that in the trimmed ATE
application in Section~\ref{te_unconf_sec}: increasing $h$ increases the
precision of the estimate, but decreases the size of the complier subpopulation
associated with the estimand.

Theorem~\ref{late_thm} in Appendix~\ref{sec:techn-deta-appl} shows that under
appropriate regularity conditions, the confidence band
$[\hat{\theta}(h)\pm
\cvt{1-\alpha}(\overline{h}/\underline{h},k_{\text{uniform}})
se(h)]_{h\in[\underline{h},\overline{h}]}$, where $k_{\text{uniform}}$ denotes
the uniform kernel, is a valid confidence band for $\theta(h)$. This result
follows from the fact that $\hat\theta(h)$ is composed of kernel-based
estimators with the uniform kernel (e.g.\
$\frac{1}{\#\{Z_i\in[\underline z,\underline z+h]\}}\sum_{Z_i\in[\underline
  z,\underline z+h]}Y_i$ is a uniform kernel estimate of
$E[Y_i\mid Z_i=\underline z]$). This confidence band provides a simple way of
summarizing the estimates of $\theta(h)$ for a range of values of $h$ and their
statistical accuracy, while formally taking into account that one has looked at
multiple estimates. This allows the reader to assess treatment effect
heterogeneity by examining how $\theta(h)$ varies with $h$, or obtain a
confidence interval for $\theta(0)$ based on their own beliefs about the
smoothness of $\theta(h)$.

In addition to the trimmed ATE and LATE applications, similar extensions are
possible to other econometric models that are ``identified at infinity''
\citep[see, among
others][]{chamberlain_asymptotic_1986,heckman_varieties_1990,andrews_semiparametric_1998}.
In the interest of brevity, we do not pursue such extensions here.

\section{Empirical illustrations}\label{empirical_sec}
\subsection{U.S. House elections}\label{sec:u.s.-house-elections}
Our first empirical example is based on \citet{lee08}, who is interested in the
effect of an incumbency advantage in U.S.~House elections. Given the inherent
uncertainty in final vote counts, the party that wins is essentially randomized
in elections that are decided by a narrow margin, so that the incumbency
advantage can be identified using a sharp regression discontinuity design.

In particular, the running variable $X_{i}$ is the Democratic margin of victory
in a given election $i$. The outcome variable $Y_{i}$ is the Democratic vote
share in the next election. The parameter $\theta(0)$ is then the incumbency
advantage for Democrats---the impact of being the current incumbent party in a
congressional district on the probability of winning the next election. There
are $6,558$ observations in this dataset, spanning House elections between 1946
and 1998.

To analyze the data, \citet{lee08} uses a global fourth degree polynomial, which
yields a point estimate of 7.7\%. However, global polynomial estimates may give
large weights to observations far away from the threshold and be sensitive to
the degree of the polynomial \citep{gi14}. We therefore reanalyze the data using
local linear and local quadratic regression with a triangular kernel. We
consider bandwidths between $2$ and $40$, which includes the
\citet[IK]{imbens_optimal_2012} optimal bandwidth selector for local linear
regression, equal to $29.4$. Figure~\ref{fig:lee-example} plots the results.
Because the IK bandwidth is designed to minimize the mean squared error of the
local linear estimator, as discussed in Section~\ref{reg_discont_sec}, the bias
at bandwidths of this order is not asymptotically negligible. Panel (a) of
Figure~\ref{fig:lee-example} should therefore be interpreted as a confidence
band for $\theta(h)$. As discussed in that section, one can interpret the local
quadratic estimator as implementing the \citet{cct14} bias-correction method, so
that panel (b) can be  interpreted as giving results for $\theta(0)$.

The incumbency effect remains positive and significant over the entire range,
even after using the corrected critical value, and after implementing the
\citet{cct14} bias correction. At the IK bandwidth, the confidence interval is
given by $(4.49,8.87)$ for the local quadratic (bias-corrected) estimator. Our
adjustment widens it slightly to $(3.82,9.54)$. These results suggest that the
estimates are very robust to the choice of bandwidth.

\subsection{Progresa / Oportunidades}\label{sec:progr--oport}
Our second empirical example examines the effect of the Oportunidades
anti-poverty conditional cash transfer program in Mexico, using a dataset from
\citet[CCT]{cct14}. The program started in 1998 under the name of Progresa in
rural areas, and expanded to urban areas in 2003. The program is designed to
target poverty by providing cash payments to families in exchange for regular
school attendance, health clinic visits, and nutritional support. The transfer
constituted a significant contribution to the income of eligible families.

We focus on the program treatment effect in the urban areas. Here, unlike in the
rural areas, the program was first offered in neighborhoods with the highest
density of poor households. In order to accurately target the program to poor
households, household eligibility to participate in the program was based on a
pre-intervention household poverty index. This eligibility assignment rule
naturally leads to sharp (intention-to-treat) regression-discontinuity design.

As in CCT, we focus on the effect of the program on food and non-food
consumption expenditures two years after its implementation (consumption is
measured in pesos, expressed as monthly expenditures per household member). We
normalize the poverty index so that the participation cutoff is zero. There are
2,809 households in the dataset, 691 with index $X_{i}>0$, and 2,118 controls
with $X_{i}<0$. For the effect on food consumption, the IK bandwidth selector
sets $h_{IK}=1.44$, with 95\% confidence interval around the local linear
estimator equal to $(6.7, 71.2)$, and to $(4.6, 102.7)$ for the local quadratic
estimator, suggesting a significantly positive effect. For non-food consumption,
$h_{IK}=1.09$, and the 95\% confidence intervals are given by $(1.6, 53.7)$ for
the local linear estimator, and by $(4.5,79.3)$ for the local quadratic
estimator. To examine sensitivity of these results to snooping, we plot the estimates,
along with pointwise and uniform confidence bands over a range of bandwidths in
Figures~\ref{fig:food-example} and~\ref{fig:nfood-example}. In contrast to the
previous empirical example, the figures indicate that the results are sensitive
to bandwidth choice: the uniform bands contain zero over the entire range
plotted for both outcomes.

\subsection{Right heart catheterization}
Our final example uses data from \citet{connors96rhc} to examine the effect of
receiving right heart catheterization (RHC) on 30-day mortality. The data
contain information on 5,735 adult patients who were critically ill upon
admission to the hospital ICU, 2,184 treated and 3,551 controls. The treatment,
an indicator for receiving RHC withing 24 hours of admission, is assumed to be
as good as randomly conditional on 72 covariates (see \citet{connors96rhc} for a
detailed description).

The original analysis by \citet{connors96rhc} matched on the propensity score
estimated by a logistic regression, with each unit matched at most once. It
found that RHC appeared to lead to lower survival than not performing RHC\@,
contradicting a popular perception among practitioners that RHC was beneficial.
To estimate the treatment effect, we follow the procedure in reanalysis of this
data by \citet{crump_dealing_2009}. First, we estimate the propensity score by
logistic regression. We then take the difference between the treated and control
units weighted by the estimated propensity score. Standard errors are computed
by the bootstrap.

Due to limited overlap, \citet{crump_dealing_2009} trim the data by setting the
trimming parameter to $h=0.1$, discarding individuals with propensity score
lower than 0.1 and higher than 0.9. To examine sensitivity of the results to the
amount of trimming, we consider a range of trimming parameters from $0$ to
$0.1$. This leads to an effective bandwidth ratio $\hat{t}=2.00$.
Figure~\ref{fig:rhc} plots the results. Without trimming, the unadjusted 95\%
confidence interval is given by $(0.027, 0.092)$. Trimming at $h=0.1$ reduces it
to $(0.031, 0.087)$. Adjusting the confidence intervals for snooping widens them
to $(0.018, 0.100)$ and $(0.024,0.094)$, respectively. Overall, the results are
stable over the trimming range, with the precision of the estimates increasing
with trimming. The conclusion that RHC negatively impacts survival is robust to
snooping, with RHC lowering the 30-day survival probability by about 6\%.

\section{Conclusion}\label{conclusion_sec}
Nonparametric estimators typically involve a choice of tuning parameter. To
ensure robustness of the results to tuning parameter choice, researchers often
examine sensitivity of the results to the value of the tuning parameter.
However, if the tuning parameter is chosen based on this sensitivity analysis,
the resulting confidence intervals may undercover even if the estimator is
unbiased.

In this paper, we addressed this problem when the estimator is kernel-based, and
the tuning parameter is a bandwidth. We showed that if one uses an adjusted
critical value instead of the usual critical value based on quantiles of a
normal distribution, the resulting confidence interval will be robust to this
form of ``bandwidth snooping''.

The adjustment only depends on the kernel and the ratio of biggest to smallest
bandwidth that the researcher has tried. Therefore, readers can easily quantify
the robustness of reported results to the bandwidth choice, as long as both a
point estimate and a standard error have been reported. Our method also allows
researchers to report the results for a range of bandwidths along with the
adjusted confidence bands as a routine robustness check, allowing readers to
select their own bandwidth.

\clearpage

\appendix

\part*{\LARGE{Appendix}}
\allowdisplaybreaks
This appendix contains the proof of Theorem~\ref{highlevel_asym_dist_thm} in the
main text, as well as auxiliary results. Appendix~\ref{sec:proof-main-result}
contains the proof of the main result. Appendix~\ref{sec:techn-deta-appl} gives
formal regularity conditions applying the main result to models considered in
Section~\ref{applications_sec}. Appendix~\ref{sensitivity_analysis_sec}
discusses the use of uniform and pointwise in $h$ confidence regions in
sensitivity analysis. Additional results and proofs are in the supplemental
appendix.

Throughout this appendix, we use the following additional notation.
For a sample $\{Z_i\}_{i=1}^n$ and a function $f$ on the sample space,
$E_{n}f(Z_i)=\frac{1}{n}\sum_{i=1}^{n}f(Z_i)$ denotes the sample mean, and
$\mathbb{G}_{n}f(Z_i)=\sqrt{n}(E_{n}-E)f(Z_i)=\sqrt{n}[E_{n}f(Z_i)-Ef(Z_i)]$
denotes the empirical process.  We use $t\vee t'$ and $t\wedge t'$ to denote elementwise maximum and minimum, respectively.  We use $e_k$ to denote the $k$th basis vector in Euclidean space (where the dimension of the space is clear from context).

\section{Proof of Main Result}\label{sec:proof-main-result}

\subsection{Equivalence Results for Extreme Value Limits}

This section proves an equivalence result for extreme value limits of the form
proved in this paper.

\begin{theorem}\label{ev_limit_thm_general}
Let $h_n^*$ and $\underline h_n$ be sequences with $\underline h_n\to 0$, $h_n^*=\mathcal{O}(1)$ and $h_n^*/\underline h_n\to\infty$, and let $\mathbb{T}_n(h)$ and $\tilde{\mathbb{T}}_n(h)$ be random processes on $\mathbb{R}$.
Suppose that
\begin{equation}\label{ev_limit_eq_general}
\sqrt{2\log\log (h_n^*/\underline h_n)}\left(\sup_{\underline h_n\le h\le h_n^*}\mathbb{T}_n(h)
  -\sqrt{2\log\log (h_n^*/\underline h_n)}\right)
-b(\log\log (h_n^*/\underline h_n))\stackrel{d}{\to} Z
\end{equation}
for some limiting variable $Z$ and $b(t)=\log c_2$ or
$b(t)=\log c_1+\log \sqrt{2t}$ for some constants $c_1$ and $c_2$.

Suppose that
\begin{equation}\label{lln_equiv_general}
\sqrt{\log\log (h_n^*/\underline h_n)}\sup_{\underline h_n\le h\le h_n^*}
\left|\mathbb{T}_n(h)-\tilde{\mathbb{T}}_n(h)\right|\stackrel{p}{\to} 0.
\end{equation}
Then (\ref{ev_limit_eq_general}) holds with $\mathbb{T}_n(h)$ replaced by $\tilde{\mathbb{T}}_n(h)$.  If, in addition, for some sequence $\overline h_n$ with $\overline h_n\ge h_n^*$,
$\log\log (h_n^*/\underline h_n)-\log\log (\overline h_n/\underline h_n)\to 0$ and, for some $\varepsilon>0$,
\begin{equation}\label{Tn_bound_eq}
\frac{\sup_{h_n^*\le h\le \overline h_n}\tilde{\mathbb{T}}_n(h)}{\sqrt{2\log\log (\overline h_n/\underline h_n)}}\le 1-\varepsilon
\text{ with probability approaching one},
\end{equation}
then (\ref{ev_limit_eq_general}) holds with $\mathbb{T}_n(h)$ replaced by $\tilde{\mathbb{T}}_n(h)$ and $h_n^*$ replaced by $\overline h_n$.
\end{theorem}
\begin{proof}
The first claim is immediate from the bound $\Bigl|\sup_{\underline h_n\le h\le h_n^*}
\mathbb{T}_n(h)-\sup_{\underline h_n\le h\le h_n^*}\tilde{\mathbb{T}}_n(h)\Bigr|\le \sup_{\underline h_n\le h\le h_n^*}
\left|\mathbb{T}_n(h)-\tilde{\mathbb{T}}_n(h)\right|$ and Slutsky's theorem.

For the second claim, note that, since~\eqref{ev_limit_eq_general} holds for
$\tilde{\mathbb{T}}_n$, $\sup_{\underline h_n\le h\le h_n^*}
\tilde{\mathbb{T}}_n(h)/\sqrt{2\log\log(\overline h_n/\underline
  h_n)}\stackrel{p}{\to} 1$ so that, with probability approaching one,
$\sup_{\underline h_n\le h\le \overline h_n}
\tilde{\mathbb{T}}_n(h)=\sup_{\underline h_n\le h\le h_n^*} \mathbb{T}_n(h)$. By
Slutsky's theorem, $a_{n}X_n-b_n\stackrel{d}{\to} Z$ implies
$a_n'X_n-b_n'\stackrel{d}{\to} Z$ so long as $b_n-b_n'\to 0$ and
$(a_n-a_n')\frac{1\vee b_n}{a_n}\to 0$ (note that $(a_n-a_n')X_n-(b_n-b_n')
=\frac{a_{n}-a_n'}{a_n} (a_{n} X_n-b_n) +\frac{b_n}{a_n}(a_n-a_n')-(b_n-b_n')$).
Applying this fact with $a_n=\sqrt{2\log\log (h_n^*/\underline h_n)}$,
$a_n'=\sqrt{2\log\log (\overline h_n/\underline h_n)}$, $b_n=2\log\log
(h_n^*/\underline h_n)+b(\log\log (h_n^*/\underline h_n))$ and $b_n'=2\log\log
(\overline h_n/\underline h_n)+b(\log\log (\overline h_n/\underline h_n))$, we
have
\begin{align*}
(a_n-a_n')\frac{1\vee b_n}{a_n}
&=\left(\sqrt{2\log\log (h_n^*/\underline h_n)}
  -\sqrt{2\log\log (\overline h_n/\underline h_n)}\right)
\textstyle\frac{2\log\log (h_n^*/\underline h_n)+b(\log\log (h_n^*/\underline h_n))}{\sqrt{2\log\log (h_n^*/\underline h_n)}}  \\
&=\left(\sqrt{2\log\log (h_n^*/\underline h_n)}-\sqrt{2\log\log (\overline h_n/\underline h_n)}\right)
\left(\sqrt{2\log\log (h_n^*/\underline h_n)}+o(1)\right)  \\
&=\frac{2\log\log (h_n^*/\underline h_n)-2\log\log (\overline h_n/\underline h_n)}{\sqrt{2\log\log (h_n^*/\underline h_n)}+\sqrt{2\log\log (\overline h_n/\underline h_n)}}
\left(\sqrt{2\log\log (h_n^*/\underline h_n)}+o(1)\right)
\to 0
\end{align*}
and $b_n-b_n'=b(\log\log (h_n^*/\underline h_n))-b(\log\log (\overline
h_n/\underline h_n))+o(1)\to 0$ since $|b(t)-b(t')|\le t-t'$ for large enough
$t$ and $t'$.
\end{proof}

To prove our main result, we apply Theorem~\ref{ev_limit_thm_general} twice.
First, we show that, under the conditions of
Theorem~\ref{highlevel_asym_dist_thm}, for some $\varepsilon>0$,
\begin{equation*}
  \frac{\sup_{h^*_n\le h\le \overline h_n} \sqrt{nh}|\hat\theta(h)-\theta(h)|/\hat\sigma(h)}{\sqrt{2\log\log(\overline h_n/\underline h_n)}}
  =\frac{\sup_{h^*_n\le h\le \overline h_n}\frac{1}{\sqrt{nh}}|\sum_{i=1}^n\psi(W_i,h)k(X_i/h)|}
  {\sqrt{2\log\log(\overline h_n/\underline h_n)}}
  +o_P(1)
  \le 1-\varepsilon
\end{equation*}
with probability approaching one,
where
\begin{equation}\label{hnstar_def_eq}
h_n^*=\exp\left[-(\log \underline h_n^{-1})^{1/K}\right]
\end{equation}
for $K$ large enough (the reasoning behind this choice of $h_n^*$ is explained below; in the case where $\overline h_n$ goes to zero more quickly than this choice of $h_n^*$, this step can be skipped).  For this choice of $h_n^*$, (\ref{lln_equiv_general}) is shown to hold with $\tilde{\mathbb{T}}_n(h)$ given by $\frac{\sqrt{nh}|\hat\theta(h)-\theta(h)|}{\hat\sigma(h)}$ and
$\mathbb{T}_n(h)$ given by $\frac{1}{\sqrt{nh}} |\sum_{i=1}^n \tilde Y_{i}k(X_i/h)|$, where
\begin{equation}\label{tilde_yi_eq}
\tilde Y_i=\frac{\psi(W_i,0)-E[\psi(W_i,0)||X_i|]}{\sqrt{var(\psi(W_i,0)||X_i|)f_{|X|}(|X_i|)\int_{0}^{\infty} k(u)^2\, du}}.
\end{equation}

Next, it is shown that~\eqref{lln_equiv_general} holds for
$\tilde{\mathbb{T}}_n(h)$ given by $\frac{1}{\sqrt{nh}} |\sum_{i=1}^n \tilde{Y}_{i}k(X_i/h)|$ and $\mathbb{T}_n(h)$ given by the absolute value of a Gaussian process with
the same covariance kernel, which can be constructed on the same sample space.
Calculating this covariance kernel, we see that
\begin{multline*}
cov\left(\frac{1}{\sqrt{nh}} \sum_{i=1}^n \tilde Y_{i}k(X_i/h),
\frac{1}{\sqrt{nh}} \sum_{i=1}^n \tilde Y_{i}k(X_i/h')\right)
=E\frac{1}{\sqrt{hh'}} E[\tilde Y_i^2||X_i|]k(|X_i|/h)k(|X_i|/h')  \\
=E\left\{\frac{1}{\sqrt{hh'}} \left[f_{|X|}(|X_i|)\int_{0}^\infty k(u)^2\, du\right]^{-1}k(|X_i|/h)k(|X_i|/h')\right\}
=\frac{\int k(x/h)k(x/h')\, dx}{\sqrt{hh'}\int k(u)^2\, du}
\end{multline*}
(here, we use the fact that $k(|X_i|/h)=k(X_i/h)$ and
$\int k(u)^2\, du=2\int_0^\infty k(u)^2\, du$, since $k$ is symmetric). The
change of variables $u=x/h'$ shows that the covariance kernel depends only on
$h'/h$, so that the Gaussian process is stationary when indexed by $t=\log h$.
The result then follows by applying a theorem for limits of stationary Gaussian
processes on increasing sets \citep[see][]{leadbetter_extremes_1983}.

The reasoning behind this choice of $h_n^*$ is as follows. With
$h_n^*=\exp[-(\log \underline h_n^{-1})^{1/K}]$, we have $h_n^*/\underline h_n
=\exp[-(\log \underline h_n^{-1})^{1/K}+(\log \underline h_n^{-1})]
=\exp\{(\log \underline h_n^{-1})[1-(\log \underline
  h_n^{-1})^{1/K-1}]\}$, so that
  \begin{equation*}
\log\log (h_n^*/\underline h_n)
=\log\{(\log \underline h_n^{-1})[1-(\log \underline h_n^{-1})^{1/K-1}]\}
=\log\log(\underline h_n^{-1})+\log [1-(\log \underline h_n^{-1})^{1/K-1}].
\end{equation*}
Since the last term converges to zero, this is equal to $\log\log(\underline
h_n^{-1})$ up to an $o(1)$ term, and the same holds for $\log\log(\overline
h_n/\underline h_n)$ as required.

To see why this choice of $h_n^*$ is useful for showing (\ref{Tn_bound_eq}),
note that, if the supremum of $\tilde{\mathbb{T}}_n(h)$ increases at the same
rate over $h^*_n\le h\le \overline h_n$ (as a function of $\overline h_n/h^*_n$)
as it does over $\underline h_n\le h\le h_n^*$ (as a function of
$h^*_n/\underline h_n$), then we will have, for some constant $C$ that does not
depend on $h_n^*$, $\sup_{h^*_n\le h\le \overline h_n}
\tilde{\mathbb{T}}_n(h)\le C\sqrt{\log\log (\overline h_n/h_n^*)}$ with
probability approaching one. Thus, (\ref{Tn_bound_eq}) will hold so long as
$\frac{\log \log (\overline h_n/h_n^*)}{\log \log (\overline h_n/\underline
  h_n)} =\frac{\log \log {h_n^*}^{-1}}{\log \log \underline h_n^{-1}}+o(1)$ can
be made arbitrarily small by making $K$ large, which we can do since $\log \log
{h_n^*}^{-1}=\log (\log \underline h_n^{-1})^{1/K} =(1/K)\log\log \underline
h_n^{-1}$.

The rest of this section uses Theorem~\ref{ev_limit_thm_general} to prove
Theorem~\ref{highlevel_asym_dist_thm}. First, we state some empirical process
bounds, which will be used later in the proof.

\subsection{Empirical Process Bounds}

This section states some empirical process bounds used later in the proof. The
proofs of these results are given in Supplemental Appendix~\ref{kern_tail_sec}
(see Lemmas~\ref{lil_bound_lemma_supp} and~\ref{lil_rate_lemma_supp}). In these
lemmas, the following conditions are assumed to hold for some finite constants
$B_f$, $B_k$ and $\overline f_X$. The function $f(w,h,t)$ is assumed to satisfy
$|f(W_i,h,t)k(X_i/h)|\le B_f$ for all $h\le \overline h$ and $t\in T$ with
probability one, and the class of functions
$\{(x,w)\mapsto f(w,h,t)k(x/h)| 0\le h\le \overline h, t\in T\}$ is contained in
some larger class $\mathcal{G}$ with polynomial covering number as defined in
Supplemental Appendix~\ref{tail_bounds_sec}. We assume that $k(x)$ is a bounded
kernel function with support $[-A,A]$ and $|k(x)|\le B_k<\infty$, and that $X_i$
is a real valued random variable with density $f_X(x)$ with
$f_X(x)\le \overline f_X<\infty$ for all $x$.

\begin{lemma}\label{lil_bound_lemma}
  Suppose that the conditions given above hold and let
  $a(h)=2\sqrt{K\log\log (1/h)}$ where $K$ is a constant depending only on
  $\mathcal{G}$ given in Lemma~\ref{kern_bound_lemma}. Then, for a constant
  $\varepsilon>0$ that depends only on $K$, $A$ and $\overline f_X$,
\begin{multline*}
P\left(|\mathbb{G}_{n}f(W_i,h,t)k(X_i/h)|
  \ge a(h) h^{1/2} B_{f}A^{1/2}\overline f_X^{1/2} \text{ some $(\log\log n)/(\varepsilon n)\le h\le \overline h$, $t\in T$}\right)  \\
\le K(\log 2)^{-2}\sum_{(2\overline h)^{-1}\le 2^k\le \infty} k^{-2}.
\end{multline*}
\end{lemma}

%Using these bounds, we obtain the following uniform bound on $\mathbb{G}_nf(W_i,h,t)k(X_i/h)$.

\begin{lemma}\label{lil_rate_lemma}
Under the conditions of Lemma~\ref{lil_bound_lemma},
\begin{equation*}
  \sup_{(\log\log n)/(\varepsilon n)\le h\le \overline h, t\in T}
  \frac{|\mathbb{G}_{n}f(W_i,h,t)k(X_i/h)|}{(\log\log h^{-1})^{1/2}h^{1/2}}
  =\mathcal{O}_P(1)
\end{equation*}
\end{lemma}
It will be useful to state a slight extension of these results. Suppose that
$f(W_i,h,t)k(X_i/h)$ converges to zero as $h\to 0$. In particular, suppose that,
for some bounded function $\ell(h)$,
\begin{equation}\label{mod_cont_assump}
  f(W_i,h,t)k(X_i/h)\le \ell(h)
\end{equation}
with probability one. Applying the above results with $f(W_i,h,t)$
replaced by $f(W_i,h,t)/\ell(h)$, we then have
\begin{equation*}
  \sup_{(\log\log n)/(\varepsilon n)\le h\le \overline h, t\in T}
  \frac{|\mathbb{G}_{n}f(W_i,h,t)k(X_i/h)|}{(\log\log h^{-1})^{1/2}h^{1/2}\ell(h)}
  =\mathcal{O}_P(1).
\end{equation*}
Thus,
\begin{align*}
  \sup_{\underline h_n\le h\le \overline h_n, t\in T}
  \frac{|\mathbb{G}_{n}f(W_i,h,t)k(X_i/h)|}{h^{1/2}} &=\mathcal{O}_P\left(
    \sup_{\underline h_n\le h\le \overline h_n}(\log\log h^{-1})^{1/2}\ell(h)
  \right)  \\
  &=\mathcal{O}_P\left( (\log\log \overline h_n^{-1})^{1/2}\ell(\overline h_n)
  \right),
\end{align*}
where the second equality holds if $(\log\log h^{-1})^{1/2}\ell(h)$ is nondecreasing in $h$.

\subsection{Replacing $\psi(W_i,h)$ with $\tilde Y_i$}

This section shows that~\eqref{Tn_bound_eq} holds for
$\tilde{\mathbb{T}}_{n}(h)=\sqrt{nh}|\hat\theta(h)-\theta(h)|/\hat\sigma(h)$,
and that~\eqref{lln_equiv_general} holds for
$\mathbb{T}_n(h)=\frac{1}{\sqrt{nh}}|\sum_{i=1}^n\tilde Y_{i} k(X_i/h)|$.

The following lemma proves~\eqref{Tn_bound_eq} for
$\sqrt{nh}|\hat\theta(h)-\theta(h)|/\hat\sigma(h)$.

\begin{lemma}\label{hstar_bound_lemma}
%note: this is a modified version of sub_set_equiv_lemma from previous version
  Suppose that the classes of functions $w\mapsto \psi(w,h)$ and
  $x\mapsto k(x/h)$ have polynomial uniform covering numbers, $\psi(w,h)k(x/h)$
  is bounded, $X_i$ has a bounded density and that $k$ is a bounded kernel
  function with support $[-A,A]$.

  Let $h_n^*$ be defined as above for some constant $K$ and let $\overline h_n$
  be a bounded sequence $\overline h_n\ge h_n^*$. Then, if $K$ is large enough,
  (\ref{Tn_bound_eq}) will hold for
  $\tilde{\mathbb{T}}_n(h)=\frac{1}{\sqrt{h}}|\mathbb{G}_n\psi(W_i,h)k(X_i/h)|$.
  Thus, under Assumption~\ref{inf_func_assump}, (\ref{Tn_bound_eq}) will hold
  for
  $\tilde{\mathbb{T}}_n(h)=\sqrt{nh}|\hat\theta(h)-\theta(h)|/\hat\sigma(h)$.
\end{lemma}
\begin{proof}
Let $C$ be such that, for any $\tilde h$,
\begin{align*}
P\left(\sup_{\underline h_n\le h\le \tilde h}\frac{1}{\sqrt{\log\log h^{-1}}\sqrt{h}}|\mathbb{G}_n\psi(W_i,h)k(X_i/h)|
  > C\right)
\le C \sum_{(2\tilde h)^{-1}\le k\le \infty} k^{-2}
\end{align*}
(this can be done by Lemma~\ref{lil_bound_lemma}).
Given $\delta>0$, let $\tilde h_{\delta}$ be such that the right hand side of this display is less than $\delta$, and let $\tilde C_\delta$ be such that
$\sup_{\tilde h_\delta\le h\le \overline h_n}\frac{1}{\sqrt{h}}|\mathbb{G}_n\psi(W_i,h)k(X_i/h)|\le \tilde C_{\delta}$ with probability at least $1-\delta$.
Then, with probability at least $1-2\delta$,
\begin{align*}
&\sup_{h_n^*\le h\le \overline h_n}\frac{1}{\sqrt{h}}|\mathbb{G}_n\psi(W_i,h)k(X_i/h)|  \\
&\le \max\left\{\sqrt{2\log\log {h_n^*}^{-1}}
  \sup_{h_n^*\le h\le \tilde h_\delta}\frac{|\mathbb{G}_n\psi(W_i,h)k(X_i/h)|}{\sqrt{\log\log h^{-1}}\sqrt{h}},
\sup_{\tilde h_\delta\le h\le \overline h_n}\frac{1}{\sqrt{h}}|\mathbb{G}_n\psi(W_i,h)k(X_i/h)|\right\}
  \\
&\le C \cdot \sqrt{2\log\log {h_n^*}^{-1}}+\tilde C_\delta
=C \cdot \sqrt{(2/K)\log\log \underline h_n^{-1}}+\tilde C_{\delta}
\le C \cdot \sqrt{(3/K)\log\log \underline h_n^{-1}}
\end{align*}
for large enough $n$.
Since $\delta$ was arbitrary, it follows that
$\frac{\sup_{h_n^*\le h\le \overline h_n}\frac{1}{\sqrt{h}}|\mathbb{G}_n\psi(W_i,h)k(X_i/h)|}{\sqrt{2\log\log \underline h_n^{-1}}}
\le C\sqrt{3/(2K)}$ with probability approaching one.  Since this can be made less than $1-\varepsilon$ by making $K$ large
(and since $\limsup_n \sqrt{2\log\log \underline h_n^{-1}}/\sqrt{2\log\log (\overline h_n/\underline h_n)}\le 1$), the result follows.
\end{proof}

We now show that~\eqref{lln_equiv_general} holds for
$\mathbb{T}_n(h)=\frac{1}{\sqrt{nh}}|\sum_{i=1}^n\tilde Y_{i}k(X_i/h)|$ and
$\tilde{\mathbb{T}}_n(h)=\sqrt{nh}|\hat\theta(h)-\theta(h)|/\hat\sigma(h)$. By
Assumption~\ref{inf_func_assump}, it suffices to show this for
$\tilde{\mathbb{T}}_n(h)=\frac{1}{\sqrt{h}}|\mathbb{G}_n\psi(W_i,h)k(X_i/h)|$. To
this end, we first prove a general result where $\mathbb{T}_n(h)$ and
$\tilde{\mathbb{T}}_n(h)$ are given by
$\frac{1}{\sqrt{nh}}|\sum_{i=1}^n\psi(W_i,h)k(X_i/h)|$ and
$\frac{1}{\sqrt{nh}}|\sum_{i=1}^n\tilde\psi(W_i,h)k(X_i/h)|$, and then verify
these conditions for $\tilde \psi(W_i,h)$ given by $\tilde Y_i$.

\begin{lemma}\label{psi_equiv_lemma}
  Suppose that the conditions of Lemma~\ref{hstar_bound_lemma} hold as stated
  and with $\psi$ replaced by $\tilde\psi$. If
  $|[\tilde{\psi}(W_i,h)-\psi(W_i,h)]k(X_i/h)|\le \ell(h)$ for some function
  $\ell(h)$ with $\lim_{h\to 0} \ell(h)\log\log h^{-1}= 0$. Then, for $h_n^*$
  given in~\eqref{hnstar_def_eq},
\begin{align*}
\sqrt{\log\log (h_n^*/\underline h_n)}\sup_{\underline h_n\le h\le\overline h_n^*}\left|\frac{1}{\sqrt{h}}\mathbb{G}_n\psi(W_i,h)k(X_i/h)-\frac{1}{\sqrt{h}}\mathbb{G}_n\tilde\psi(W_i,h)k(X_i/h)\right|
\stackrel{p}{\to} 0.
\end{align*}
\end{lemma}
\begin{proof}
  By Lemma~\ref{lil_rate_lemma} applied to
  $[\tilde{\psi}(W_i,h)-\psi(W_i,h)]k(X_i/h)/\ell(h)$, we have
\begin{align*}
\sup_{\underline h_n\le h\le h_n^*}\left|\frac{1}{\sqrt{h}}\mathbb{G}_n\psi(W_i,h)k(X_i/h)-\frac{1}{\sqrt{h}}\mathbb{G}_n\tilde\psi(W_i,h)k(X_i/h)\right|
=\mathcal{O}_P\left(\sup_{\underline h_n\le h\le  h_n^*}\ell(h)\sqrt{\log\log h^{-1}}\right).
\end{align*}
Since $\lim_{h\to 0} \ell(h)\log\log h^{-1}= 0$, we can assume without loss of generality that $\ell(h)\log\log h^{-1}$ is nondecreasing and that, therefore,
$\ell(h)\sqrt{\log\log h^{-1}}$ is nondecreasing.  Thus,
\begin{multline*}
\sqrt{\log\log (h_n^*/\underline h_n)}\sup_{\underline h_n\le h\le\overline h_n^*}\left|\frac{1}{\sqrt{h}}\mathbb{G}_n\psi(W_i,h)k(X_i/h)-\frac{1}{\sqrt{h}}\mathbb{G}_n\tilde\psi(W_i,h)k(X_i/h)\right|  \\
=\mathcal{O}_P\left(\ell(h_n^*)\sqrt{\log\log {h_n^*}^{-1}}\sqrt{\log\log (h_n^*/\underline h_n)}\right)  \\
%&=\mathcal{O}_P\left(\ell(h_n^*)\log\log {h_n^*}^{-1}\frac{\sqrt{\log\log {h_n^*}^{-1}}}{\log\log {h_n^*}^{-1}}\sqrt{\log\log (h_n^*/\underline h_n)}\right)  \\
=\mathcal{O}_P\left(\ell(h_n^*)\log\log {h_n^*}^{-1}\frac{\sqrt{\log\log (h_n^*/\underline h_n)}}{\sqrt{\log\log {h_n^*}^{-1}}}\right).
\end{multline*}
The result follows since $\ell(h_n^*)\log\log {h_n^*}^{-1}\to 0$ and
\begin{equation*}
  \frac{\sqrt{\log\log (h_n^*/\underline h_n)}}{\sqrt{\log\log {h_n^*}^{-1}}}
  \le \frac{\sqrt{\log\log \underline h_n^{-1}}}{\sqrt{\log\log {h_n^*}^{-1}}}
  =\frac{\sqrt{\log\log \underline h_n^{-1}}}{\sqrt{(1/K)\log\log
      \underline h_n^{-1}}}
  =\sqrt{K}.
\end{equation*}
\end{proof}

We now show that the conditions of Lemma~\ref{psi_equiv_lemma} hold for
$\tilde\psi(W_i,h)$ given by $\tilde Y_i$ under the conditions of
Theorem~\ref{highlevel_asym_dist_thm}.

\begin{lemma}\label{ytilde_mod_cont_lemma}
Under the conditions of Theorem~\ref{highlevel_asym_dist_thm},
$|[\psi(W_i,h)-\tilde Y_i]k(X_i/h)|\le \ell(h)$ for some function $\ell(h)$ with
$\lim_{h\to 0} \ell(h)\log\log h^{-1}=0$.
\end{lemma}
\begin{proof}
  Let $\tilde\sigma^2(x)=var[\psi(W_i,0)||X_i|=x]$,
  $a(x)=[\tilde\sigma^2(x)f_{|X|}(x)\int_{0}^{\infty} k(u)^2\, du]^{-1/2}$, and
  $\tilde \mu(x)=E[\psi(W_i,0)||X_i|=x]$. We have
\begin{align*}
&[\psi(W_i,h)-\tilde Y_i]k(X_i/h)  \\
&=[\psi(W_i,h)-\psi(W_i,0)]k(X_i/h)
+\left\{\psi(W_i,0)-a(|X_i|)\left[\psi(W_i,0)-\tilde\mu(|X_i|)\right]\right\}k(X_i/h)  \\
&=[\psi(W_i,h)-\psi(W_i,0)]k(X_i/h)
+\psi(W_i,0)[1-a(|X_i|)]k(X_i/h)
+a(|X_i|)\tilde\mu(|X_i|)k(X_i/h)
\end{align*}
The first term is bounded by a function $\ell(h)$ with $\lim_{h\to 0} \ell(h)\log\log h^{-1}=0$ by assumption.

The second term is bounded by a constant times $\sup_{0\le x\le Ah} |1-a(x)|$,
and the last term is bounded by a constant times $\sup_{0\le x\le Ah}
|\tilde\mu(x)|$ once $a(x)$ is shown to be bounded. To deal with these terms,
note that $a(0)=1$ and $\tilde\mu(0)=0$ by construction (this is shown below in
Lemma~\ref{a0_mu0_lemma}). Thus,
\begin{align*}
&\sup_{0\le x\le Ah} |1-a(x)|
=\sup_{0\le x\le Ah} |a(0)-a(x)|  \\
&=\left[\int_{0}^{\infty} k(u)^2\, du\right]^{-1/2}\sup_{0\le x\le Ah} \left|[\tilde\sigma^2(0)f_{|X|}(0)]^{-1/2}
-[\tilde\sigma^2(x)f_{|X|}(x)]^{-1/2}\right|.
\end{align*}
By continuous differentiability of $(s,t)\mapsto (st)^{-1/2}$ at
$s=\tilde\sigma^2(0)$ and $t=f_{|X|}(0)$ along with
Assumption~\ref{dgp_kern_assump}, this is bounded by a constant times
$\sup_{0\le x\le Ah}\ell(x)$ for a function $\ell(h)$ with
$\ell(h)\log\log h^{-1}\to 0$ as $h\to 0$. Since
$[\log \log h^{-1}]\sup_{0\le x\le Ah}\ell(x) \le \sup_{0\le x\le Ah}[\log \log
x^{-1}]\ell(x)$, this bound satisfies the required conditions. The last term is
bounded by a constant times $\sup_{0\le x\le Ah} |\tilde\mu(x)-\tilde\mu(0)|$,
and this term is bounded by a function $\ell(h)$ with
$\ell(h)\log\log h^{-1}\to 0$ as $h\to 0$ by assumption.
\end{proof}

The following lemma is used in the proof of Lemma~\ref{ytilde_mod_cont_lemma}.

\begin{lemma}\label{a0_mu0_lemma}
Under the conditions of Theorem~\ref{highlevel_asym_dist_thm},
$a(0)=1$ and $\tilde \mu(0)=0$,
where $a(x)$ and $\tilde \mu(x)$ are defined in Lemma~\ref{ytilde_mod_cont_lemma}.
\end{lemma}
\begin{proof}
Note that
\begin{align*}
&0=\frac{1}{h}E\psi(W_i,h)k(X_i/h)
=\frac{1}{h}E\psi(W_i,0)k(X_i/h)+\frac{1}{h}E[\psi(W_i,h)-\psi(W_i,0)]k(X_i/h)  \\
&=\tilde \mu(0) \frac{1}{h}Ek(X_i/h)
+\frac{1}{h}E(\tilde \mu(X_i)-\tilde \mu(0))k(X_i/h)
+\frac{1}{h}E[\psi(W_i,h)-\psi(W_i,0)]k(X_i/h).
\end{align*}
As $h\to 0$,
$\frac{1}{h}Ek(X_i/h)%=\int_0^\infty f_{|X|}(uh) k(u)\, du
  \to f_{|X|}(0)\int_0^\infty k(u)\, du>0$,
$\frac{1}{h}E(\tilde\mu(x)-\tilde\mu(0))k(X_i/h)\to 0$ and
$\frac{1}{h}E[\psi(W_i,h)-\psi(W_i,0)]k(X_i/h)\to 0$, so taking limits in the above display shows that $\tilde\mu(0)=0$.  Similarly,
\begin{align*}
&1=\frac{1}{h}var(\psi(W_i,h)k(X_i/h))  \\
&=\frac{1}{h}var(\psi(W_i,0)k(X_i/h))
   +\frac{1}{h}var([\psi(W_i,h)-\psi(W_i,0)]k(X_i/h))  \\
&   +\frac{2}{h}cov([\psi(W_i,h)-\psi(W_i,0)]k(X_i/h),\psi(W_i,0)k(X_i/h)).
\end{align*}
As $h\to 0$, the last two terms converge to zero, since they are bounded by $\ell(h)$ or $\ell(h)^2$ times terms of the form $Ek(X_i/h)/h$ and $Ek(X_i/h)^2/h$.  The first term is
\begin{align*}
\frac{1}{h}
\int_{0}^\infty \tilde \sigma^2(x) k(x/h)^2f_{|X|}(x)\, dx
+\frac{1}{h}var(\mu(|X_i|)k(|X_i|/h)),
\end{align*}
which converges to $\tilde\sigma^2(0)f_{|X|}(0)\int_{0}^\infty k(u)^2\, du$ as $h\to 0$ (the last term is bounded by a constant times $\ell(h)^2$).  Thus,
$\tilde \sigma^2(0)=\left(f_{|X|}(0)\int_0^\infty k(u)^2\, du\right)^{-1}$ so that,
with $a(x)$ defined above, $a(0)=1$.
\end{proof}

\subsection{Gaussian Approximation}\label{gauss_approx_sec}

This section shows that
$\frac{1}{\sqrt{h}}\mathbb{G}_n\tilde Y_{i}k(X_i/h)
=\frac{1}{\sqrt{nh}}\sum_{i=1}^n\tilde Y_{i}k(X_i/h)$ is approximated by a
Gaussian process with the same covariance kernel. The proof of this result,
given in Supplemental Appendix~\ref{gauss_approx_sec_supp}, uses an application
of a Gaussian approximation theorem of \citet{sakhanenko_convergence_1985} along
with arguments similar to those in \citet{bickel_global_1973}. It is worth
noting that other Gaussian approximation results could be used, with potentially
different regularity conditions. For example, one could use results from
\citet{chernozhukov_gaussian_2014}, which would allow us to replace the
assumption of a bounded outcome variable with assumptions bounding the higher
moments, at the expense of stronger conditions on $\underline h_n$ (this would
also require additional truncation arguments elsewhere in the proofs).

We consider a general setup with $\{(\tilde X_i,\tilde Y_i)\}_{i=1}^n$ i.i.d.,
with $\tilde X_i\ge 0$ a.s.\ such that $\tilde X_i$ has a density
$f_{\tilde X}(x)$ on $[0,\overline x]$ for some $\overline x\ge 0$, with
$f_{\tilde X}(x)$ bounded away from zero and infinity on this set. We assume
that $\tilde Y_i$ is bounded almost surely, with $E(\tilde Y_i|\tilde X_i)=0$
and $var(\tilde{Y}_i|\tilde{X}_i=x)=f_{\tilde X}(x)^{-1}$. We assume that the
kernel function $k$ has finite support $[0,A]$ and is differentiable on its
support with bounded derivative. For ease of notation, we assume in this section
that $\int k(u)^2\, du=1$. The result applies to our setup with $\tilde Y_i$
given in (\ref{tilde_yi_eq}) and $\tilde X_i$ given by $|X_i|$.

%begin material from Gaussian approximation note

Let
\begin{equation*}
\hat{\mathbb{H}}_n(h)=\frac{1}{\sqrt{nh}}\sum_{i=1}^n\tilde Y_{i} k(\tilde X_i/h).
\end{equation*}

\begin{theorem}\label{gauss_approx_thm}
Under the conditions above, there exists, for each $n$, a process $\mathbb{H}_n(h)$ such that, conditional on $(\tilde X_1,\ldots,\tilde X_n)$, $\mathbb{H}_n$ is a Gaussian process with covariance kernel
\begin{equation*}
cov\left(\mathbb{H}_n(h),\mathbb{H}_n(h')\right)
  =\frac{1}{\sqrt{h h'}}\int k(x/h)k(x/h')\, dx
\end{equation*}
and
\begin{equation*}
  \sup_{\underline h_n\le h\le \overline x/A} \left|\hat{\mathbb{H}}_n(h)-\mathbb{H}_n(h)\right|
  =\mathcal{O}_P\left(
    (n\underline h_n)^{-1/4}[\log (n \underline h_n)]^{1/2}\right)
\end{equation*}
for any sequence $\underline h_n$ with $n\underline h_n/ \log\log \underline h_n^{-1}\to\infty$.
\end{theorem}

For our purposes, we need $(n\underline h_n)^{-1/4} [\log (n\underline
h_n)]^{1/2} \cdot (\log\log \underline h_n^{-1})^{1/2}\to 0$, so that the rate
in the above theorem is $o_P(1/\sqrt{\log\log\underline h_n})$. For this, the
condition $n\underline h_n/[(\log \log n)(\log\log\log n)]^2\to\infty$
given in the conditions of Theorem~\ref{highlevel_asym_dist_thm}, is sufficient,
since this implies, for some $a_n\to\infty$, $(n \underline h_n)^{1/4}\ge a_{n}
(\log\log n)^{1/2}(\log\log\log n)^{1/2}$ and this implies, for large enough
$n$,
\begin{align*}
(n\underline h_n)^{-1/4} [\log (n\underline h_n)]^{1/2}
&\le a_n^{-1}
  \frac{\{\log [a_n (\log\log n)^{1/2}(\log\log\log n)^{1/2}]^{4}\}^{1/2}}{(\log\log n)^{-1/2}(\log\log\log n)^{-1/2}}  \\
&=a_n^{-1}
  \frac{\{4[\log a_n+ (1/2)\log\log\log n+ (1/2)\log\log\log\log n]\}^{1/2}}{(\log\log n)^{-1/2}(\log\log\log n)^{-1/2}}  \\
&\le 2 a_n^{-1}(\log a_n+1)^{1/2}(\log\log n)^{-1/2}.
\end{align*}

\subsection{Limit Theorem for the Gaussian Approximation}\label{gaussian_limit_sec}

This section derives the limiting distribution of the approximating Gaussian process as $\overline h_n/\underline h_n$ increases.

\begin{theorem}\label{gaussian_limit_thm}
Let $\mathbb{H}(h)$ be a Gaussian process with mean zero and covariance kernel
\begin{align*}
cov\left(\mathbb{H}(h),\mathbb{H}(h')\right)
  =\frac{\int k(u/h)k(u/h')\, du}{\sqrt{hh'}\int k(u)^2\, du}
  =\sqrt{\frac{h'}{h}}\frac{\int k(u (h'/h))k(u)\, du}{\int k(u)^2\, du},
\end{align*}
where $k$ is a bounded symmetric kernel with a bounded derivative and support
$[-A,A]$. Let $c_1=\frac{Ak(A)^2}{\sqrt{\pi}\int k(u)^2\, du}$,
$c_2=\frac{1}{2\pi}\sqrt{\frac{\int \left[k'(u)u+\frac{1}{2} k(u)\right]^2\,
    du}{\int k(u)^2\, du}}$, and let $b(t)=\log c_2$ if $k(A)=0$ and $b(t)=\log
c_1+\frac{1}{2}\log t$ if $k(A)\ne 0$. Let $\underline h_n$ and $\overline h_n$
be sequences with $\overline h_n/\underline h_n\to\infty$. Then
\begin{align*}
\sqrt{2\log\log (\overline h_n/\underline h_n)}\left(\sup_{\underline h_n\le h\le \overline h_n}\left|\mathbb{H}(h)\right|
  -\sqrt{2\log\log (\overline h_n/\underline h_n)}\right)
-b(\log\log (\overline h_n/\underline h_n))\stackrel{d}{\to} Z\vee Z'
\end{align*}
where $Z$ and $Z'$ are independent extreme value random variables.
\end{theorem}
\begin{proof}
  We use Theorem 12.3.5 of \citet{leadbetter_extremes_1983} applied to the
  process $\mathbb{X}(t)=\mathbb{H}(e^t)$, which is stationary,
%to do: this theorem only gives the one-sided case - need to finde two-sided version
with, in the case where $k(A)\ne 0$, $\alpha=1$ and $C=\frac{Ak(A)^2}{\int k(u)^2\, du}$ and, in the case where $k(A)=0$, $\alpha=2$ and $C=\frac{\int \left[k'(u)u+\frac{1}{2} k(u)\, du\right]^2\, du}{2\int k(u)^2\, du}$.
%to do: make sure this gives the right c_1 and c_2 with the formulas in the Leadbetter et al book with Pickands constant substituted
The calculations and verification of the conditions for this theorem follow from
elementary calculus and are given in Supplemental
Appendix~\ref{ev_calc_sec_supp}.
\end{proof}

\subsection{Proof of Theorem~\ref{highlevel_asym_dist_thm}}

We are now ready to prove Theorem~\ref{highlevel_asym_dist_thm}. Before
proceeding, we recall a result regarding absolute continuity of suprema of
Gaussian processes, which will be used in the proof.

\begin{lemma}\label{pitt_tran_lemma}
Let $\mathbb{X}(t)$ be a Gaussian process on a countable index set $\mathcal{T}$ with $P(\sup_{t\in\mathcal{T}}|\mathbb{X}(t)|<\infty)=1$ and $\inf_{t\in\mathcal{T}} var(\mathbb{X}(t))>0$.  Then $\sup_{t\in\mathcal{T}}\mathbb{X}(t)$ has an absolutely continuous distribution with bounded density.
\end{lemma}
\begin{proof}
See Proposition 3.2 in \citet{pitt_local_1979}.
\end{proof}

It follows from Lemma~\ref{pitt_tran_lemma} that the distribution of $\sup_{h\in[\underline h,\overline h]}\mathbb{H}(h)$ is absolutely continuous for any $0<\underline h\le \overline h<\infty$ (the supremum is equal to the supremum over a countable subset with probability one by continuity of the sample paths).
It also follows that $\sup_{h\in[\underline h,\overline h]}|\mathbb{H}(h)|$ is absolutely continuous, since $\sup_{h\in[\underline h,\overline h]}\mathbb{H}(h)$ and $\sup_{h\in[\underline h,\overline h]} -\mathbb{H}(h)$ are absolutely continuous and absolute continuity of $Y$ and $Z$ implies absolute continuity of $\max\{Y,Z\}$.

\begin{proof}[proof of Theorem~\ref{highlevel_asym_dist_thm}]
  By arguing along subsequences, we can assume without loss of generality that
  $\overline h_n/\underline h_n\to h^*$ for some $h^*\in [0,\infty)$ or
  $h^*=\infty$. In the first case,
\begin{align*}
  \sup_{\underline h_n\le h\le\overline h_n}\frac{\sqrt{n h}|\hat\theta(h)-\theta(h)|}
  {\hat\sigma(h)}
  =\sup_{1\le t\le\overline h_n/\underline h_n}|\mathbb{H}_n(t\underline h_n)|
  +r_n,
\end{align*}
where $r_n\stackrel{p}{\to} 0$ and $\mathbb{H}_n(h)$ is, conditional on $\{|X_i|\}_{i=1}^n$, a Gaussian process with the same distribution as $\mathbb{H}(h)$.  Since multiplying $h$ by a constant does not change the distribution of $\mathbb{H}(h)$, it follows that
\begin{align*}
  \sup_{1\le t\le\overline h_n/\underline h_n}|\mathbb{H}_n(t\underline h_n)|
  \stackrel{d}{=} \sup_{1\le h\le\overline h_n/\underline h_n}|\mathbb{H}(h)|
  \stackrel{d}{\to} \sup_{1\le h\le h^*}|\mathbb{H}(h)|,
\end{align*}
where the last step follows from stochastic equicontinuity of $\mathbb{H}(h)$ on
compact intervals.
The result then follows by continuity of the distribution of $\sup_{1\le h\le
  h^*}|\mathbb{H}(h)|$ at $\cvt{1-\alpha}(h^*,k)$
(which follows from the result in \citealt{pitt_local_1979} stated in
Lemma~\ref{pitt_tran_lemma}).

In the case where $\overline h_n/\underline h_n\to\infty$, let $h_n^*$ be given
by (\ref{hnstar_def_eq}) for some $K$ which will be chosen large enough to
satisfy conditions given below. We can assume without loss of generality that
either $\overline h_n> h_n^*$ for all $n$ large enough or that $\overline h_n\le
h_n^*$ for all $n$ large enough (again, by arguing along subsequences). In the
former case, we apply Lemma~\ref{hstar_bound_lemma} to show that
condition~\eqref{Tn_bound_eq} holds for
$\sqrt{nh}|\hat\theta(h)-\theta(h)|/\hat\sigma(h)$ so
long as $K$ is chosen large enough in the definition of $h_n^*$. Thus, by
Theorem~\ref{ev_limit_thm_general}, it suffices to consider the latter case
where $\overline h_n\le h_n^*$.

By Lemmas~\ref{psi_equiv_lemma}
and~\ref{ytilde_mod_cont_lemma},~\eqref{lln_equiv_general} holds for
$\sqrt{nh}|\hat\theta(h)-\theta(h)|/\hat\sigma(h)$ and
$\frac{1}{\sqrt{nh}}|\sum_{i=1}^n\tilde Y_{i}k(X_i/h)|$. It therefore follows from
Theorem~\ref{ev_limit_thm_general} that it suffices to consider
$\frac{1}{\sqrt{nh}}|\sum_{i=1}^n\tilde Y_{i}k(X_i/h)|$. By
Theorem~\ref{gauss_approx_thm}, this can be replaced by $|\mathbb{H}_n(h)|$, where
$\mathbb{H}_n(h)$ is the Gaussian process conditional on $\{|X_i|\}_{i=1}^n$
defined in the proof of that theorem. By Theorem~\ref{gaussian_limit_thm},
\begin{align*}
\sqrt{2\log\log (\overline h_n/\underline h_n)}\left(\sup_{\underline h_n\le h\le \overline h_n}|\mathbb{H}_n(h)|
  -\sqrt{2\log\log (\overline h_n/\underline h_n)}\right)
-b(\log\log (\overline h_n/\underline h_n))\stackrel{d}{\to} Z\vee Z'.
\end{align*}
Thus, by Theorems~\ref{ev_limit_thm_general} and~\ref{gauss_approx_thm}, the
same holds with $|\mathbb{H}_n(h)|$ replaced by
$\sqrt{nh}|\hat\theta(h)-\theta(h)|/\hat\sigma(h)$. Since
$\cvt{1-\alpha}(\overline h_n/\underline h_n,k)$ is the $1-\alpha$ quantile of a
distribution that converges in distribution to $Z\vee Z'$ by
Theorem~\ref{gauss_approx_thm}, and since the cdf of $Z\vee Z'$ is continuous, the
result follows.  The last display in the statement of the theorem follow directly from this extreme value limit.
\end{proof}

\section{Regularity conditions for applications}\label{sec:techn-deta-appl}

In this appendix, we state three theorems that give primitive regularity
conditions for applications discussed in Section~\ref{applications_sec}. Proofs
for these results are in Supplemental Appendix~\ref{applications_sec_supp}.

To state these conditions, we use the following additional notation. For a
random vector $(X_i,D_i,Y_i)$ with $X_i$ continuously distributed, let
$E(Y_i|D_i=d,X_i=\tilde x_+)=\lim_{x\downarrow \tilde x} E(Y_i|D_i=d,X_i=x)$ and
$E(Y_i|D_i=d,X_i=\tilde x_-)=\lim_{x\uparrow \tilde x} E(Y_i|D_i=d,X_i=x)$. For
a function $f:\mathbb{R}\to\mathbb{R}$ and $\tilde x\in\mathbb{R}$, let
$f(\tilde x_+)=\lim_{x\downarrow \tilde x} f(x)$ and
$f(\tilde x_-)=\lim_{x\uparrow \tilde x} f(x)$ when these limits exist. We say
that a function $f$ is right-continuous at $\tilde x$ with local modulus of
continuity $\ell(x)$ if $\|f(x)-f(\tilde x_+)\|\le \ell(\|x-\tilde x\|)$ for all
$x >\tilde x$ with $\|x-\tilde x\|$ small enough. We say that a function $f$ is
left-continuous at $\tilde x$ with local modulus of continuity $\ell(x)$ if
$\|f(x)-f(\tilde x_-)\|\le \ell(\|x-\tilde x\|)$ for all $x <\tilde x$ with
$\|x-\tilde x\|$ small enough. We say that a function $f$ is continuous at
$\tilde x$ with local modulus $\ell(x)$ if it is both left- and right-continuous
with $f(\tilde x_+)=f(\tilde x_-)=f(\tilde x)$. Note that we define left- and
right-continuity with respect to the left- and right-hand limits of the
function, so that a function may be both left- and right-continuous according to
our definition even if these limits are different (as is typically the case in
RD).

\begin{theorem}\label{reg_disc_thm}
  Consider the regression discontinuity design from
  Section~\ref{reg_discont_sec}. Suppose that
\begin{itemize}
\item[(i)] $|X_i|$ has a density $f_{|X|}(x)$ at $x=0$, $Y_i$ is bounded, and,
  for some deterministic function $\ell(t)$ with $\lim_{t\to 0}\log\log
  t^{-1}\ell(t)=0$, the functions $f_X(x)$, $var((D_i,Y_i)'|X_i=x)$, $E(Y_i|X_i=x)$
  and $E(D_i|X_i=x)$ are left- and right-continuous at $0$ with local modulus of
  continuity $\ell(t)$.

\item[(ii)] $P(D_i=1|X_i=0_+)-P(D_i=1|X_i=0_-)\ne 0$ and $var(Y_i|D_i=d,X_i=0_+)\ne 0$ or $var(Y_i|D_i=d,X_i=0_-)\ne 0$ for $d=0$ or $1$.
\end{itemize}

Then, for $\hat\theta(h)$ and $\theta(h)$ given in Section~\ref{reg_discont_sec}
and $\hat\sigma(h)$ corresponding to the Eicker-Huber-White standard error
estimator given in the supplemental appendix, if the kernel function $k^*$
satisfies part (i) of Assumption~\ref{dgp_kern_assump}, then
Assumptions~\ref{inf_func_assump} and Assumption~\ref{dgp_kern_assump} hold with
$k(u)=(\mu_{k^*,2}-\mu_{k^*,1}|u|)k^*(u)$, so long as $\overline h_n$ is bounded
by a small enough constant and $n\underline h_n/(\log\log \underline
h_n^{-1})^3\to\infty$.
\end{theorem}

Next, consider the problem of constructing uniform confidence bands for average
treatment effects under unconfoundedness as in Section~\ref{te_unconf_sec}. Let
$\hat\theta(h)$ be an estimator of $\theta(h)$ with influence function
representation
\begin{equation}\label{cate_inf_func_eq}
\sqrt{n}(\hat\theta(h)-\theta(h))
=\frac{1}{\sqrt{n}}\sum_{i=1}^n \frac{[\tilde Y_i-\theta(h)]I(X_i\in\mathcal{X}_h)}{P(X_i\in\mathcal{X}_h)}+o_P(1),
\end{equation}
where the $o_P(1)$ term is uniform over $\underline h\le h\le\overline h$ and
$\tilde Y_i:=D_i\frac{Y_i-\mu_1(X_i)}{e(X_i)}
-(1-D_i)\frac{Y_i-\mu_0(X_i)}{1-e(X_i)}+\mu_1(X_i)-\mu_0(X_i)$. See
\citet{crump_dealing_2009} for references to the literature for estimators that
satisfy this condition. This condition requires that the trimmed sample is
constructed using a known propensity score $e(x)$ (while allowing for a fully
nonparametric estimator on the trimmed sample). If instead one uses the trimming
rule $\hat{\mathcal{X}}_{h}=\{x\colon h < \hat{e}(x)\le 1-h\}$ based on an
estimated propensity score $e(x)$, we conjecture that~\eqref{cate_inf_func_eq}
will hold for $\sqrt{n}(\hat \theta(h)-\theta(\widehat{\mathcal{X}}_h))$ under
regularity conditions, where, for a set $\mathcal{X}$, $\theta(\mathcal{X})$ is
defined as $E(Y_i(1)-Y_i(0)\mid X_i\in \mathcal{X})$.\footnote{In the
  pointwise-in-$h$ case, similar results are given in the the working paper
  version \citep{crump_moving_2006} of \citet{crump_dealing_2009}; verifying
  this conjecture essentially involves verifying that their results can be
  generalized to hold uniformly over $\underline h\le h\le\overline h$.} This
will lead to uniform confidence bands for the parameter
$\theta(\widehat{\mathcal{X}}_h)$, which can be interpreted as the average
treatment effect for the random subpopulation $\widehat{\mathcal{X}}_h$. The
influence function~\eqref{cate_inf_func_eq} and the pivotal asymptotic
distribution we derive below are specific to estimators of trimmed average
treatment effects: other classes of estimators and trimming rules with different
influence function representations may not lead to a snooping adjusted critical
value based on a pivotal asymptotic distribution. Note that
$E(\tilde Y_i|X_i)=\tau(X_i)$ so that
$E(\tilde Y_i|X_i\in \mathcal{X}_h)=\theta(h)$. Let
\begin{align*}
\sigma(h)^2=var\left\{\frac{[\tilde Y_i-\theta(h)]I(X_i\in\mathcal{X}_h)}{P(X_i\in\mathcal{X}_h)}\right\}
=\frac{var\left\{[\tilde Y_i-\theta(h)]I(X_i\in\mathcal{X}_h)\right\}}{P(X_i\in\mathcal{X}_h)^2},
\end{align*}

The following theorem proves the validity of the confidence band given in
Section~\ref{te_unconf_sec}.

\begin{theorem}\label{ate_unconf_thm}
Let $0\le \underline h<\overline h< 1/2$.
Suppose that
\begin{itemize}
\item[(i)] the influence function representation (\ref{cate_inf_func_eq}) holds
  uniformly over $\underline h\le h\le \overline h$, and
  $se(h)=\hat\sigma(h)/\sqrt{n}$ where $\hat\sigma(h)$ is consistent for
  $\sigma(h)$ uniformly over $\underline h\le h\le \overline h$

\item[(ii)] $\theta(h)$ is bounded uniformly over $\underline h\le h\le
  \overline h$ and $E[\tilde Y_i^2|X_i]$ is bounded uniformly over $\underline
  h\le e(X_i)\le 1-\underline h$ and

\item[(iii)] $v(\overline h)>0$ where $v(h)=E\{[\tilde Y_i-\theta(h)]^2I(X_i\in\mathcal{X}_h)\}$.
\end{itemize}
Let $\hat t=\frac{se(\underline h)^2N(\underline h)^2}{se(\overline
  h)^2N(\overline h)^2}$ as defined in (\ref{ate_unconf_t_hat_eq}). Then
\begin{align*}
\liminf_{n}P\left(\frac{\sqrt{n}\left|\hat\theta(h)-\theta(h)\right|}{\hat\sigma(h)}
\le \cvt{1-\alpha}(\hat t,k_{\text{uniform}})
\text{ all $h\in[\underline h,\overline h]$}\right)\ge 1-\alpha,
\end{align*}
where $k_{\text{uniform}}$ is the uniform kernel.
If, in addition, $v(h)$ is continuous, the above display holds with the $\liminf$ replaced by $\lim_{n\to \infty}$ and $\ge$ replaced by $=$.
\end{theorem}

The final result gives the regularity conditions for inference on local average
treatment effects.

\begin{theorem}\label{late_thm}
  Consider the setup and notation from Section~\ref{late_sec}. Suppose that
\begin{itemize}
\item[(i)] $Z_i$ has a density $f_Z(z)$ at $z=\underline z$ and $z=\overline z$, $Y_i$ is bounded and, for some function $\ell(t)$ with $\lim_{t\to 0}\log\log t^{-1}\ell(t)=0$, $f_Z$, $var((D_i,Y_i)'|Z_i=z)$, $E(Y_i|Z_i=z)$ and $E(Z_i|Z_i=z)$ are continuous at $\underline z$ and $\overline z$ with local modulus of continuity $\ell(t)$.

\item[(ii)] $P(D_i=1|Z_i=\overline z)-P(D_i=1|Z_i=\underline z)\ne 0$ and $var(Y_i|D_i=d,z_i=\underline z)\ne 0$ or $var(Y_i|D_i=d,Z_i=\overline z)\ne 0$ for $d=0$ or $1$.
\end{itemize}

Then, for $\hat\theta(h)$, $\theta(h)$ and $\hat\sigma(h)$ given above,
Assumptions~\ref{inf_func_assump} and Assumption~\ref{dgp_kern_assump} hold with
$k(u)=I(|u|\le 1)$, so long as $\overline h_n$ is bounded by a small enough constant
and $n\underline h_n/(\log\log \underline h_n^{-1})^3\to\infty$.
\end{theorem}

\section{Specification Searches and Sensitivity Analysis}\label{sensitivity_analysis_sec}

This section discusses the use of uniform-in-the-tuning-parameter confidence
bands in sensitivity analysis and compares them to
pointwise-in-the-tuning-parameter confidence bands. The points made here apply
to any sensitivity analysis of some parameter $\theta(h)$ to a tuning parameter
$h$ \citep[e.g., $h$ may determine the subset of included covariates, as
in][]{leamer_lets_1983}, provided we have an estimator $\hat\theta(h)$ that, for
a given $h$, is approximately unbiased for $\theta(h)$.

Suppose that different readers may disagree on how $\theta(h)$ relates to
$\theta(0)$, the parameter of interest, as $h$ varies. We can report
pointwise-in-$h$ confidence sets $\mathcal{C}_{\text{pointwise}}(h)$ satisfying
\begin{equation*}
  P\left( \theta(h)\in \mathcal{C}_{\text{pointwise}}(h)
  \right)= 1-\alpha\quad \text{for  all $h\in\mathcal{H}$},
\end{equation*}
or uniform-in-$h$ confidence sets $\mathcal{C}_{\text{uniform}}(h)$ satisfying
\begin{equation*}
  P\left( \theta(h)\in \mathcal{C}_{\text{uniform}}(h) \text{ all
      $h\in\mathcal{H}$}\right)= 1-\alpha.
\end{equation*}
If each reader believes that a particular $h$ is most suitable for estimating
and performing inference on $\theta(0)$, and, if given access to the original
data, would only perform analysis based on this $h$, then the researcher can
simply report $\hat\theta(h)$ and $\mathcal{C}_{\text{pointwise}}(h)$ for a
range of values of $h$. The reader would then select an estimate
$\hat{\theta}(h)$ and a confidence set $\mathcal{C}_{\text{pointwise}}(h)$ that
correspond to their prior belief about the most appropriate $h$. The confidence
set $\mathcal{C}_{\text{pointwise}}(h)$ selected by the reader (which the reader
would have always selected regardless of the data) will have the correct
coverage for $\theta(h)$ for the given $h$.

If, however, the researcher has some liberty in choosing which $\hat\theta(h)$
to report and/or emphasize (e.g.\ by reporting some results in the abstract or
main text and others in an appendix), reporting
$\mathcal{C}_{\text{pointwise}}(h)$ can lead to undercoverage, if one interprets
coverage as ``coverage conditional on being reported/emphasized in the main
text.'' In this setting, reporting $\mathcal{C}_{\text{uniform}}(h)$ solves the
problem of undercoverage of $\theta(h)$, so long as the set $\mathcal{H}$
includes all values of $h$ considered by the researcher in choosing which
$\hat\theta(h)$ to report. This becomes particularly important when readers are
less informed about the subject matter or details of the data than the
researcher, since, in this case, readers may defer to the researcher on the
choice of $h$.

To get at these ideas in another way, let us consider some hypothesis testing
problems that a researcher might have in mind in performing a sensitivity
analysis:
\begin{align*}
H_{0,a}&\colon \theta(h)\le 0 \quad\text{for some $h\in\mathcal{H}$}, \\
H_{0,b}&\colon \theta(h)\le 0 \quad\text{for all $h\in\mathcal{H}$},\\
H_{0,c}&\colon \theta(h) \text{ has the same sign for all all $h\in\mathcal{H}$.}
\end{align*}
One may consider formalizing the notion of ``concluding that $\theta(0)$ is
greater than zero in a robust sense'' by either
\begin{equation}\label{hyp_test_1}
\text{rejecting $H_{0,a}$ (and therefore also accepting $H_{0,c}$ in the sense of rejecting its complement)}
\end{equation}
or
\begin{equation}\label{hyp_test_2}
\text{rejecting $H_{0,b}$ and failing to reject $H_{0,c}$}.
\end{equation}
Clearly, (\ref{hyp_test_1}) is a more stringent requirement than
(\ref{hyp_test_2}). Rejecting only when
$\mathcal{C}_{\text{pointwise}}(h)\subseteq (0,\infty)$ for all $h$ provides a
valid test of $H_{0,a}$ since, under $H_{0,a}$ there exists a $h_{0}$ such that
$\theta(h_{0})\leq 0$, so that
$P\left(\mathcal{C}_{\text{pointwise}}(h)\subseteq (0,\infty) \text{ all
    $h$}\right) \le P\left(\mathcal{C}_{\text{pointwise}}(h_{0})\subseteq
  (0,\infty)\right) \le
P\left(\theta(h_0)\not\in\mathcal{C}_{\text{pointwise}}(h_{0})\right)$. Thus,
under the criterion~\eqref{hyp_test_1}, it is sufficient to use pointwise-in-$h$
confidence bands. However, this approach is likely to be conservative in many
practically relevant situations: the confidence set for $\theta(0)$ is
effectively given by the union of all pointwise sets
$\mathcal{C}_{\text{pointwise}}(h)$. In our case, where $\hat\theta(h)$ is a
kernel based estimate with bandwidth $h$, such confidence interval will be very
large since pointwise confidence intervals for small $h$ can be very wide.

If, instead, one takes~\eqref{hyp_test_2} as the criterion for ``concluding that
$\theta(0)$ is greater than zero in a robust sense'', one can perform such a
test by looking at the uniform confidence band, and
concluding~\eqref{hyp_test_2} only if
$\mathcal{C}_{\text{uniform}}(h)\subseteq (0,\infty)$ for some $h$, and
$\mathcal{C}_{\text{uniform}}(h)\cap (0,\infty)\ne 0$ for all $h$. In contrast,
performing this analysis with $\mathcal{C}_{\text{pointwise}}(h)$ does not
provide a test of $H_{0,c}$ with correct size: this formulation of robustness of
the results to the tuning parameter requires a uniform-in-$h$ confidence band.
One can view this approach as a way of formulating a confidence statement for
procedures such as those proposed by \citet{imbens_regression_2008} that examine
whether the sign of of a kernel estimator changes over a range of bandwidths.

\setstretch{1.2} % one half spacing
\bibliography{../library}

\clearpage

\begin{table}[p]
  \centering
  \renewcommand{\arraystretch}{1.2}      % space between rows
  \begin{tabular}{@{}l@{\hspace{1em}}rrrrrr@{}}
&
\multicolumn{3}{c}{NW / LL (int)}&\multicolumn{3}{c}{LL (boundary)}\\
\cmidrule(rl){2-4}\cmidrule(rl){5-7}
$\overline{h}/\underline{h}$
& Unif & Tri& Epa& Unif & Tri& Epa\\
\midrule
1.0 & 1.96 & 1.96 & 1.96 & 1.96 & 1.95 & 1.96\\
1.2 & 2.24 & 2.01 & 2.03 & 2.23 & 2.03 & 2.05\\
1.4 & 2.33 & 2.05 & 2.08 & 2.33 & 2.08 & 2.11\\
1.6 & 2.40 & 2.09 & 2.12 & 2.39 & 2.12 & 2.15\\
1.8 & 2.45 & 2.11 & 2.15 & 2.44 & 2.16 & 2.19\\
2   & 2.48 & 2.14 & 2.17 & 2.48 & 2.18 & 2.22\\
3   & 2.60 & 2.22 & 2.27 & 2.60 & 2.27 & 2.32\\
4   & 2.66 & 2.26 & 2.31 & 2.66 & 2.32 & 2.37\\
5   & 2.70 & 2.30 & 2.35 & 2.71 & 2.35 & 2.41\\
6   & 2.73 & 2.32 & 2.37 & 2.73 & 2.37 & 2.43\\
7   & 2.75 & 2.34 & 2.39 & 2.76 & 2.39 & 2.45\\
8   & 2.77 & 2.35 & 2.41 & 2.78 & 2.41 & 2.47\\
9   & 2.79 & 2.37 & 2.42 & 2.79 & 2.43 & 2.48\\
10  & 2.80 & 2.38 & 2.44 & 2.81 & 2.44 & 2.50\\
20  & 2.89 & 2.45 & 2.51 & 2.89 & 2.52 & 2.58\\
50  & 2.97 & 2.53 & 2.59 & 2.98 & 2.60 & 2.66\\
100 & 3.02 & 2.57 & 2.64 & 3.02 & 2.65 & 2.71
\end{tabular}
\caption{ Critical values $\cvt{0.95}(\overline{h}/\underline{h},k)$ for
  level-5\% tests for the Uniform (Unif, $k(u)=\frac{1}{2}I(\abs{u}\leq 1)$),
  Triangular (Tri, $(1-\abs{u})I(\abs{u}\leq 1)$) and Epanechnikov (Epa,
  $3/4(1-u^{2})I(\abs{u}\leq 1)$) kernels. ``NW / LL (int)''
  refers to Nadaraya-Watson (local constant) regression in the interior or at a
  boundary, as well as local linear regression in the interior. ``LL
  (boundary)'' refers to local linear regression at a boundary (including
  regression discontinuity designs).}\label{tab:cvs-condensed}
\end{table}

\clearpage

\begin{figure}[p]
  \centering
    \input{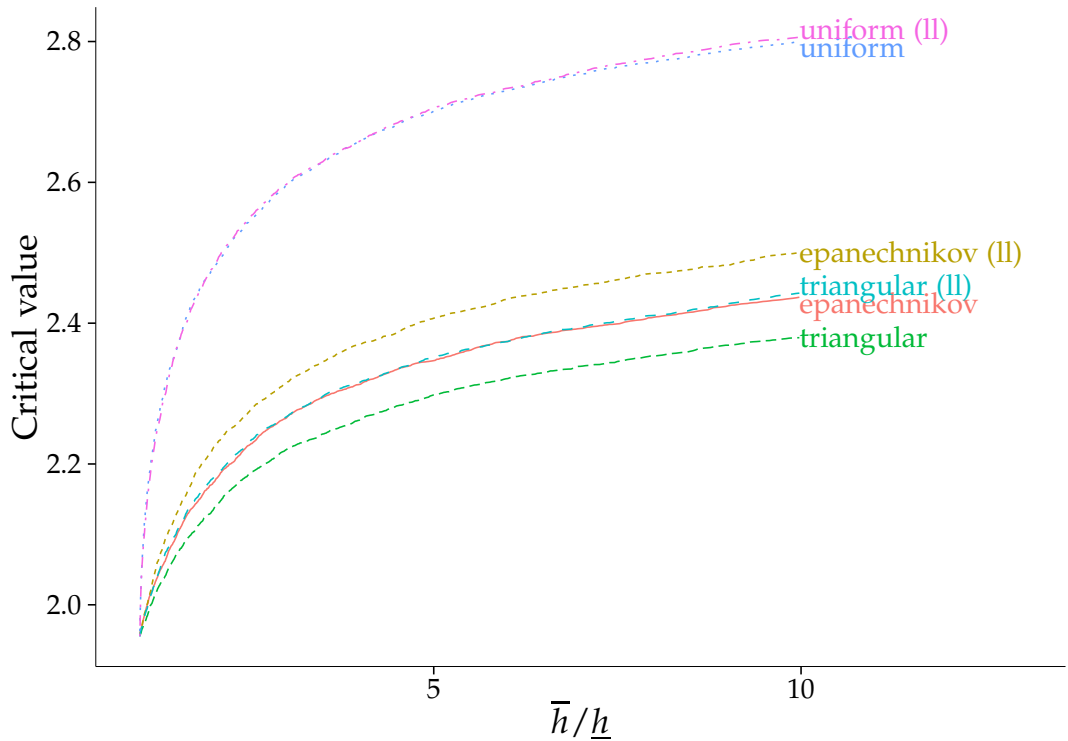}
    \caption{Two-sided 95\% critical values for different kernels. ``uniform'',
      ``triangular'' and ``epanechnikov'' refer to Nadaraya-Watson (local
      constant) regression in the interior or at a boundary as well as local
      linear regression in the interior. ``uniform (ll)'', ``triangular (ll)''
      and ``epanechnikov (ll)'' refer to local linear regression at a boundary.}\label{fig:cvs-twosided}
\end{figure}

\begin{figure}[p]
  \centering
    \input{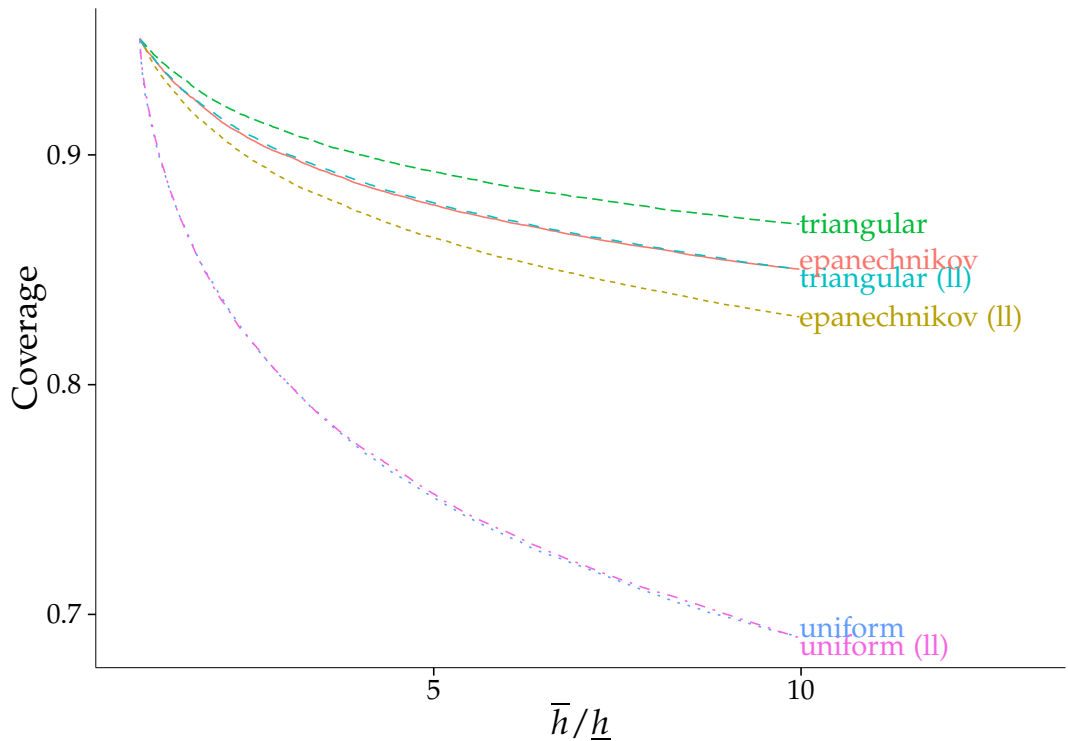}
  \caption{Coverage of unadjusted 95\% confidence bands (i.e.~using critical
    values equal to 1.96) for different kernels. ``uniform'', ``triangular'' and
    ``epanechnikov'' refer to Nadaraya-Watson (local constant) regression in the
    interior or at a boundary as well as local linear regression in the
    interior. ``uniform (ll)'', ``triangular (ll)'' and ``epanechnikov (ll)''
    refer to local linear regression at a boundary.}\label{fig:coverage-twosided}
\end{figure}

\begin{figure}[p]
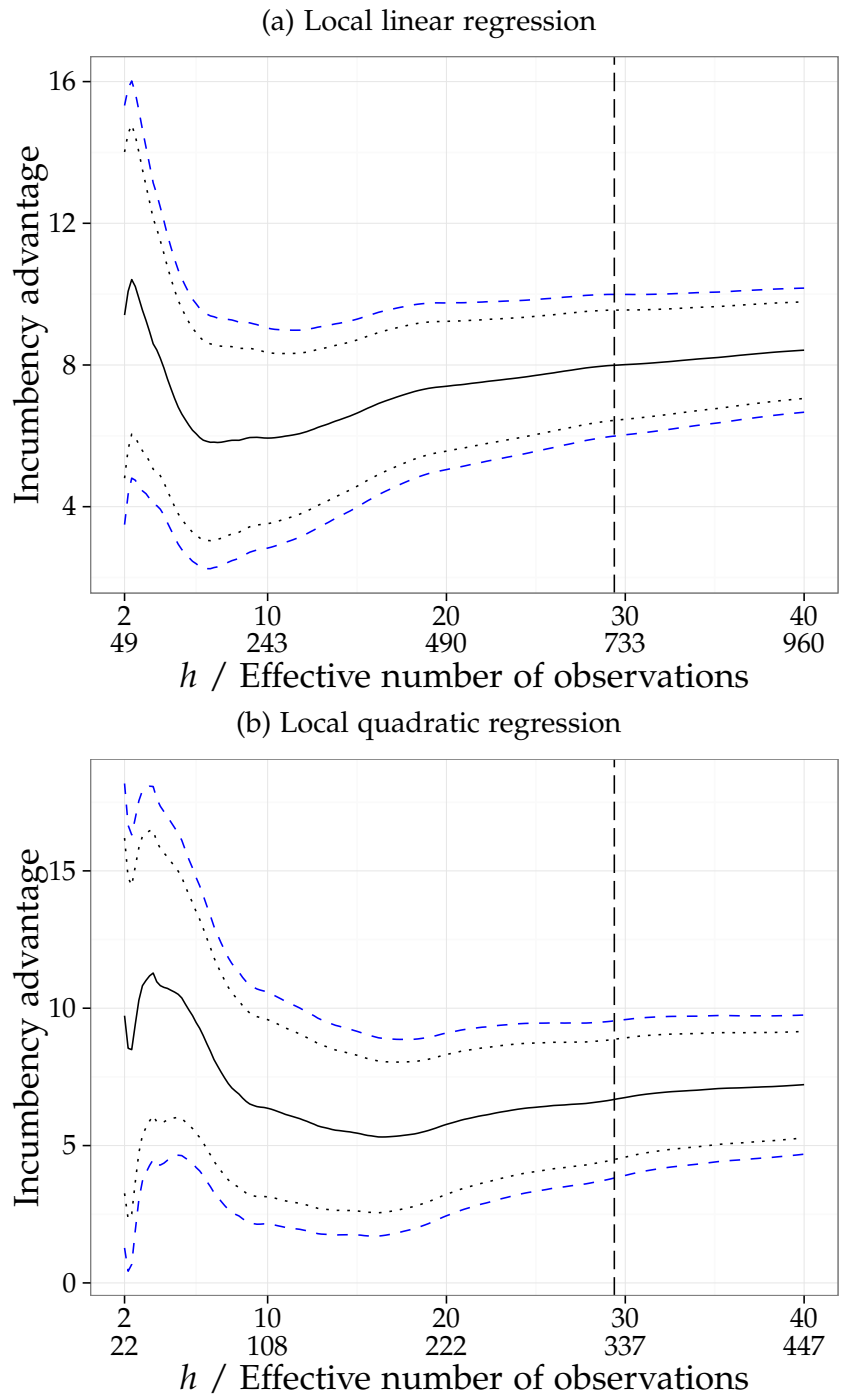

  \begin{minipage}[h]{0.97\linewidth}
    \mbox{(a) Local linear regression}
    \centering

    \input{./LeeDataN2-40.tex}
  \end{minipage}

  \begin{minipage}[h]{0.97\linewidth}
    \mbox{(b) Local quadratic regression}
    \centering

  \input{./LeeDataNQ2-40.tex}
  \end{minipage}

  \caption{Effect of incumbency on percentage vote share in the next election.
    Data are from \citet{lee08}. Local linear (panel (a)) and local quadratic
    (panel (b)) regression with triangular kernel. Point estimate
    $\hat{\theta}(h)$ (solid line), pointwise (dotted), and uniform (dashed)
    confidence bands as function of the bandwidth $h$. The range of bandwidths
    plotted is $(0.02,0.40)$, so that $\overline{h}/\underline{h}=20$, and the
    adjusted critical value is 2.52 (local linear) and 2.56 (local
    quadratic). Vertical dashed line corresponds to estimates using
    \citet{imbens_optimal_2012} bandwidth.}\label{fig:lee-example}
\end{figure}

\begin{figure}[p]
  \begin{minipage}[h]{0.97\linewidth}
    \mbox{(a) Local linear regression}
    \centering

    \input{./Prog-fl0.1-2.tex}
  \end{minipage}

  \begin{minipage}[h]{0.97\linewidth}
    \mbox{(b) Local quadratic regression}
    \centering

  \input{./Prog-fq0.1-2.tex}
  \end{minipage}

  \caption{Effect of the Oportunidades cash transfer program on food
    consumption. Data are from \citet{cct14}. Local linear (panel (a)) and local
    quadratic (panel (b)) regression with triangular kernel. Point estimate
    $\hat{\theta}(h)$ (solid line), pointwise (dotted), and uniform (dashed)
    confidence bands as function of the bandwidth $h$. The range of bandwidths
    plotted is $(0.1,2)$, so that $\overline{h}/\underline{h}=20$, and the
    adjusted critical value is 2.52 (local linear) and 2.56 (local
    quadratic). Vertical dashed line corresponds to estimates using
    \citet{imbens_optimal_2012} bandwidth. }\label{fig:food-example}
\end{figure}

\begin{figure}[p]
  \begin{minipage}[h]{0.97\linewidth}
    \mbox{(a) Local linear regression}
    \centering

    \input{./Prog-nl0.1-2.tex}
  \end{minipage}

  \begin{minipage}[h]{0.97\linewidth}
    \mbox{(b) Local quadratic regression}
    \centering

  \input{./Prog-nq0.1-2.tex}
  \end{minipage}

  \caption{Effect of the Oportunidades cash transfer program on non-food
    consumption. Data are from \citet{cct14}. Local linear (panel (a)) and local
    quadratic (panel (b)) regression with triangular kernel. Point estimate
    $\hat{\theta}(h)$ (solid line), pointwise (dotted), and uniform (dashed)
    confidence bands as function of the bandwidth $h$. The range of bandwidths
    plotted is $(0.1,2)$, so that $\overline{h}/\underline{h}=20$, and the
    adjusted critical value is 2.52 (local linear) and 2.56 (local
    quadratic). Vertical dashed line corresponds to estimates using
    \citet{imbens_optimal_2012} bandwidth. }\label{fig:nfood-example}
\end{figure}

\begin{figure}[p]
    \centering
    % Created by tikzDevice version 0.10.1 on 2016-07-20 09:56:48
% !TEX encoding = UTF-8 Unicode
\begin{tikzpicture}[x=1pt,y=1pt]
\definecolor{fillColor}{RGB}{255,255,255}
\path[use as bounding box,fill=fillColor,fill opacity=0.00] (0,0) rectangle (325.21,252.94);
\begin{scope}
\path[clip] (  0.00,  0.00) rectangle (325.21,252.94);
\definecolor{drawColor}{RGB}{255,255,255}
\definecolor{fillColor}{RGB}{255,255,255}

\path[draw=drawColor,line width= 0.6pt,line join=round,line cap=round,fill=fillColor] (  0.00,  0.00) rectangle (325.21,252.94);
\end{scope}
\begin{scope}
\path[clip] ( 48.73, 43.84) rectangle (319.21,246.94);
\definecolor{fillColor}{RGB}{255,255,255}

\path[fill=fillColor] ( 48.73, 43.84) rectangle (319.21,246.95);
\definecolor{drawColor}{gray}{0.98}

\path[draw=drawColor,line width= 0.6pt,line join=round] ( 48.73, 94.39) --
	(319.21, 94.39);

\path[draw=drawColor,line width= 0.6pt,line join=round] ( 48.73,147.70) --
	(319.21,147.70);

\path[draw=drawColor,line width= 0.6pt,line join=round] ( 48.73,201.01) --
	(319.21,201.01);

\path[draw=drawColor,line width= 0.6pt,line join=round] ( 91.76, 43.84) --
	( 91.76,246.94);

\path[draw=drawColor,line width= 0.6pt,line join=round] (153.23, 43.84) --
	(153.23,246.94);

\path[draw=drawColor,line width= 0.6pt,line join=round] (214.71, 43.84) --
	(214.71,246.94);

\path[draw=drawColor,line width= 0.6pt,line join=round] (276.18, 43.84) --
	(276.18,246.94);
\definecolor{drawColor}{gray}{0.90}

\path[draw=drawColor,line width= 0.2pt,line join=round] ( 48.73, 67.73) --
	(319.21, 67.73);

\path[draw=drawColor,line width= 0.2pt,line join=round] ( 48.73,121.04) --
	(319.21,121.04);

\path[draw=drawColor,line width= 0.2pt,line join=round] ( 48.73,174.35) --
	(319.21,174.35);

\path[draw=drawColor,line width= 0.2pt,line join=round] ( 48.73,227.66) --
	(319.21,227.66);

\path[draw=drawColor,line width= 0.2pt,line join=round] ( 61.02, 43.84) --
	( 61.02,246.94);

\path[draw=drawColor,line width= 0.2pt,line join=round] (122.50, 43.84) --
	(122.50,246.94);

\path[draw=drawColor,line width= 0.2pt,line join=round] (183.97, 43.84) --
	(183.97,246.94);

\path[draw=drawColor,line width= 0.2pt,line join=round] (245.45, 43.84) --
	(245.45,246.94);

\path[draw=drawColor,line width= 0.2pt,line join=round] (306.92, 43.84) --
	(306.92,246.94);
\definecolor{drawColor}{RGB}{0,0,0}

\path[draw=drawColor,line width= 0.6pt,line join=round] ( 61.02,140.85) --
	( 73.32,140.87) --
	( 85.61,142.02) --
	( 97.91,141.13) --
	(110.20,141.41) --
	(122.50,148.29) --
	(134.79,139.72) --
	(147.09,150.89) --
	(159.38,149.34) --
	(171.68,156.38) --
	(183.97,155.72) --
	(196.27,144.29) --
	(208.56,149.91) --
	(220.86,141.58) --
	(233.15,142.17) --
	(245.45,148.09) --
	(257.74,155.00) --
	(270.04,152.89) --
	(282.33,144.71) --
	(294.63,139.32) --
	(306.92,140.30);

\path[draw=drawColor,line width= 0.6pt,dash pattern=on 1pt off 3pt ,line join=round] ( 61.02,210.29) --
	( 73.32,210.05) --
	( 85.61,211.53) --
	( 97.91,209.29) --
	(110.20,209.00) --
	(122.50,216.35) --
	(134.79,206.46) --
	(147.09,215.17) --
	(159.38,215.24) --
	(171.68,220.72) --
	(183.97,220.03) --
	(196.27,206.64) --
	(208.56,213.52) --
	(220.86,203.38) --
	(233.15,204.13) --
	(245.45,209.28) --
	(257.74,215.63) --
	(270.04,213.09) --
	(282.33,202.88) --
	(294.63,200.71) --
	(306.92,199.96);

\path[draw=drawColor,line width= 0.6pt,dash pattern=on 1pt off 3pt ,line join=round] ( 61.02, 71.41) --
	( 73.32, 71.70) --
	( 85.61, 72.51) --
	( 97.91, 72.97) --
	(110.20, 73.82) --
	(122.50, 80.23) --
	(134.79, 72.98) --
	(147.09, 86.60) --
	(159.38, 83.45) --
	(171.68, 92.03) --
	(183.97, 91.40) --
	(196.27, 81.93) --
	(208.56, 86.31) --
	(220.86, 79.78) --
	(233.15, 80.21) --
	(245.45, 86.91) --
	(257.74, 94.36) --
	(270.04, 92.69) --
	(282.33, 86.54) --
	(294.63, 77.93) --
	(306.92, 80.65);
\definecolor{drawColor}{RGB}{0,0,255}

\path[draw=drawColor,line width= 0.6pt,dash pattern=on 4pt off 4pt ,line join=round] ( 61.02,228.62) --
	( 73.32,228.31) --
	( 85.61,229.88) --
	( 97.91,227.28) --
	(110.20,226.84) --
	(122.50,234.32) --
	(134.79,224.08) --
	(147.09,232.15) --
	(159.38,232.63) --
	(171.68,237.71) --
	(183.97,237.01) --
	(196.27,223.11) --
	(208.56,230.31) --
	(220.86,219.70) --
	(233.15,220.49) --
	(245.45,225.43) --
	(257.74,231.64) --
	(270.04,228.98) --
	(282.33,218.24) --
	(294.63,216.92) --
	(306.92,215.71);

\path[draw=drawColor,line width= 0.6pt,dash pattern=on 4pt off 4pt ,line join=round] ( 61.02, 53.08) --
	( 73.32, 53.43) --
	( 85.61, 54.16) --
	( 97.91, 54.98) --
	(110.20, 55.98) --
	(122.50, 62.26) --
	(134.79, 55.36) --
	(147.09, 69.63) --
	(159.38, 66.05) --
	(171.68, 75.04) --
	(183.97, 74.42) --
	(196.27, 65.47) --
	(208.56, 69.51) --
	(220.86, 63.47) --
	(233.15, 63.85) --
	(245.45, 70.75) --
	(257.74, 78.35) --
	(270.04, 76.80) --
	(282.33, 71.18) --
	(294.63, 61.72) --
	(306.92, 64.90);
\definecolor{drawColor}{gray}{0.50}

\path[draw=drawColor,line width= 0.6pt,line join=round,line cap=round] ( 48.73, 43.84) rectangle (319.21,246.95);
\end{scope}
\begin{scope}
\path[clip] (  0.00,  0.00) rectangle (325.21,252.94);
\definecolor{drawColor}{RGB}{0,0,0}

\node[text=drawColor,anchor=base east,inner sep=0pt, outer sep=0pt, scale=  0.96] at ( 43.33, 64.42) {0.025};

\node[text=drawColor,anchor=base east,inner sep=0pt, outer sep=0pt, scale=  0.96] at ( 43.33,117.74) {0.050};

\node[text=drawColor,anchor=base east,inner sep=0pt, outer sep=0pt, scale=  0.96] at ( 43.33,171.05) {0.075};

\node[text=drawColor,anchor=base east,inner sep=0pt, outer sep=0pt, scale=  0.96] at ( 43.33,224.36) {0.100};
\end{scope}
\begin{scope}
\path[clip] (  0.00,  0.00) rectangle (325.21,252.94);
\definecolor{drawColor}{RGB}{0,0,0}

\path[draw=drawColor,line width= 0.6pt,line join=round] ( 45.73, 67.73) --
	( 48.73, 67.73);

\path[draw=drawColor,line width= 0.6pt,line join=round] ( 45.73,121.04) --
	( 48.73,121.04);

\path[draw=drawColor,line width= 0.6pt,line join=round] ( 45.73,174.35) --
	( 48.73,174.35);

\path[draw=drawColor,line width= 0.6pt,line join=round] ( 45.73,227.66) --
	( 48.73,227.66);
\end{scope}
\begin{scope}
\path[clip] (  0.00,  0.00) rectangle (325.21,252.94);
\definecolor{drawColor}{RGB}{0,0,0}

\path[draw=drawColor,line width= 0.6pt,line join=round] ( 61.02, 40.84) --
	( 61.02, 43.84);

\path[draw=drawColor,line width= 0.6pt,line join=round] (122.50, 40.84) --
	(122.50, 43.84);

\path[draw=drawColor,line width= 0.6pt,line join=round] (183.97, 40.84) --
	(183.97, 43.84);

\path[draw=drawColor,line width= 0.6pt,line join=round] (245.45, 40.84) --
	(245.45, 43.84);

\path[draw=drawColor,line width= 0.6pt,line join=round] (306.92, 40.84) --
	(306.92, 43.84);
\end{scope}
\begin{scope}
\path[clip] (  0.00,  0.00) rectangle (325.21,252.94);
\definecolor{drawColor}{RGB}{0,0,0}

\node[text=drawColor,anchor=base,inner sep=0pt, outer sep=0pt, scale=  0.96] at ( 61.02, 31.83) {0};

\node[text=drawColor,anchor=base,inner sep=0pt, outer sep=0pt, scale=  0.96] at ( 61.02, 21.46) {5735};

\node[text=drawColor,anchor=base,inner sep=0pt, outer sep=0pt, scale=  0.96] at (122.50, 31.83) {0.025};

\node[text=drawColor,anchor=base,inner sep=0pt, outer sep=0pt, scale=  0.96] at (122.50, 21.46) {5607};

\node[text=drawColor,anchor=base,inner sep=0pt, outer sep=0pt, scale=  0.96] at (183.97, 31.83) {0.05};

\node[text=drawColor,anchor=base,inner sep=0pt, outer sep=0pt, scale=  0.96] at (183.97, 21.46) {5336};

\node[text=drawColor,anchor=base,inner sep=0pt, outer sep=0pt, scale=  0.96] at (245.45, 31.83) {0.075};

\node[text=drawColor,anchor=base,inner sep=0pt, outer sep=0pt, scale=  0.96] at (245.45, 21.46) {5020};

\node[text=drawColor,anchor=base,inner sep=0pt, outer sep=0pt, scale=  0.96] at (306.92, 31.83) {0.1};

\node[text=drawColor,anchor=base,inner sep=0pt, outer sep=0pt, scale=  0.96] at (306.92, 21.46) {4728};
\end{scope}
\begin{scope}
\path[clip] (  0.00,  0.00) rectangle (325.21,252.94);
\definecolor{drawColor}{RGB}{0,0,0}

\node[text=drawColor,anchor=base,inner sep=0pt, outer sep=0pt, scale=  1.20] at (183.97,  8.40) {$h$ / Number of observations};
\end{scope}
\begin{scope}
\path[clip] (  0.00,  0.00) rectangle (325.21,252.94);
\definecolor{drawColor}{RGB}{0,0,0}

\node[text=drawColor,rotate= 90.00,anchor=base,inner sep=0pt, outer sep=0pt, scale=  1.20] at ( 16.66,145.39) {Effect of RHC on 30-day mortality};
\end{scope}
\end{tikzpicture}
    \caption{Effect of Right Heart Catheterization on 30-day morality. Data are
      from\citet{connors96rhc}. Point estimate $\hat{\theta}(h)$ (solid line),
      pointwise (dotted), and uniform (dashed) confidence bands as function of
      the trimming parameter $h$, plotted over the range
      $h\in[0,0.1]$.}\label{fig:rhc}
\end{figure}
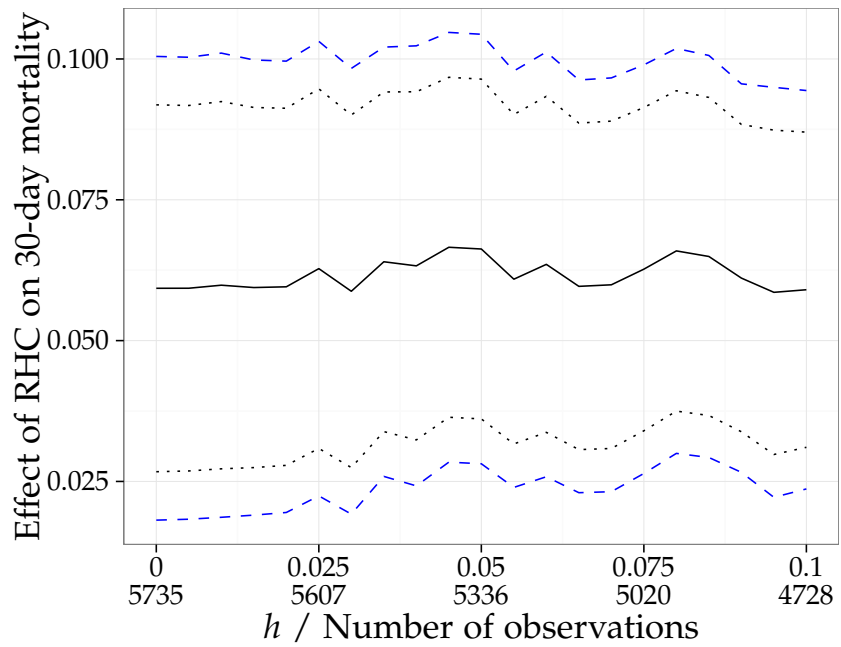

\end{document}

\maketitle

This supplement is organized as follows. Section~\ref{aux_sec_sup} contains
auxiliary results used in Appendix~\ref{sec:proof-main-result} of the main text.
Section~\ref{loc_poly_bound_sec} contains auxiliary results on local polynomial
regression. Section~\ref{applications_sec_supp} proves theorems in
Appendix~\ref{sec:techn-deta-appl}. Section~\ref{sec:critical-values} derives
critical values for one-sided confidence intervals and gives tables of one- and
two-sided critical values. Finally, Section~\ref{mc_sec} presents the results of
a Monte Carlo study.

The following additional notation, which is also used in the appendix in the
main text, is used throughout this supplement. For a sample $\{Z_i\}_{i=1}^n$
and a function $f$ on the sample space,
$E_{n}f(Z_i)=\frac{1}{n}\sum_{i=1}^{n}f(Z_i)$ denotes the sample mean, and
$\mathbb{G}_{n}f(Z_i)=\sqrt{n}(E_n-E)f(Z_i)=\sqrt{n}[E_{n}f(Z_i)-Ef(Z_i)]$
denotes the empirical process. We use $t\vee t'$ and $t\wedge t'$ to denote
elementwise maximum and minimum, respectively. We use $e_k$ to denote the $k$th
basis vector in Euclidean space (where the dimension of the space is clear from
context).

\section{Auxiliary Results}\label{aux_sec_sup}
\allowdisplaybreaks
This section contains auxiliary results that are used in the proof of
Theorem~\ref{highlevel_asym_dist_thm} in Appendix~\ref{sec:proof-main-result} of
the main text, and in the proofs of the results from
Appendix~\ref{sec:techn-deta-appl} of the main text given later in this
supplement.

\subsection{Tail Bounds for Empirical Processes}\label{tail_bounds_sec}

We state some tail bounds based on an inequality of \citet{talagrand_new_1996}
and other empirical process results. Throughout this section, we consider a
class of functions $\mathcal{G}$ on the sample space $\mathbb{R}^{d_Z}$ with an
i.i.d.\ sample of random variables $Z_1,\ldots, Z_n$. We assume throughout that
$\mathcal{G}$ has a polynomial covering number in the sense that, for some
$B,W$, $N_1(\delta,Q,\mathcal{G})\le B\varepsilon^{-W}$ for all finitely
discrete probability measures $Q$, where $N_1$ is defined in, e.g.,
\citet{pollard_convergence_1984}, p. 25.

\begin{lemma}\label{talagrand_ineq_lemma}
  Let $\tilde{\mathcal{G}}$ be a subset of $\mathcal{G}$ such that, for some
  envelope function $G$ and constant $\overline g$, $|g(Z_i)|\le G(Z_i)\le
  \overline g$ a.s.\ for all $g\in\tilde{\mathcal{G}}$. Then, for some constant
  $K$ that depends only on $\mathcal{G}$,
\begin{equation*}
P\left(
\sup_{g\in\tilde{\mathcal{G}}} \left|\mathbb{G}_{n}g(Z_i)\right|
  \ge K \sqrt{E[G(Z_i)^2]} + t \right)
  \le K\exp\left(-\frac{1}{K}\frac{t^2}{E[G(Z_i)^2]+\overline g \left\{\sqrt{E[G(Z_i)^2]}+t\right\}/\sqrt{n}}
\right)
\end{equation*}
\end{lemma}
\begin{proof}
  We apply a result of \citet{talagrand_new_1996} as stated in Equation~(3) of
  \citet{massart_about_2000}. The quantity $v$ from that version of the bound
  is, in our setting, given by $v=E\sup_{g\in\tilde{\mathcal{G}}} \sum_{i=1}^n
  [g(Z_i)-Eg(Z_i)]^2$ which, as shown in \citet[p.~882]{massart_about_2000}, is
  bounded by \citep[see also][]{klein_concentration_2005}
  \begin{equation*}
    n\sup_{g\in\tilde{\mathcal{G}}} E\{[g(Z_i)-Eg(Z_i)]^2\} +32
    \overline{g} E\sup_{g\in\tilde{\mathcal{G}}} \sum_{i=1}^n [g(Z_i)-Eg(Z_i)].
\end{equation*}
%note: 32 comes from multiplying 16 by 2 since 2 overline g is an envelope for g-Eg
By Theorem 2.14.1 in \citet{van_der_vaart_weak_1996},
%note: this assumes finite entropy integral J.   This depends on N_2 covering number, which is polynomial when N_1 covering number is polynomial by Lemma 36 on p.34 of Pollard (1984).  Since the N_1 covering number is polynomial, the entropy integral is finite.
\begin{equation}\label{vw_2.14.1_bound}
E\sup_{g\in\tilde{\mathcal{G}}} \sum_{i=1}^n [g(Z_i)-Eg(Z_i)]\le \sqrt{n} K_1 \sqrt{E[G(Z_i)]^2},
\end{equation}
for a constant $K_1$ that depends only on $\mathcal{G}$.
Combined with the fact that $E\{[g(Z_i)-Eg(Z_i)]^2\}\le E[G(Z_i)^2]$, this gives the bound
\begin{equation*}
v\le n E[G(Z_i)^2] + 32\overline g K_1\sqrt{n} \sqrt{E[G(Z_i)]^2}.
\end{equation*}
Applying the bound from equation (3) of \citet{massart_about_2000} with these quantities gives
\begin{multline*}
P\left(\sqrt{n} \sup_{g\in\tilde{\mathcal{G}}} \mathbb{G}_n g(Z_i)
  \ge K_1 \sqrt{n}\sqrt{E[G(Z_i)]^2}+ r\right)  \\
\le P\left(\sqrt{n} \sup_{g\in\tilde{\mathcal{G}}} \mathbb{G}_n g(Z_i)
  \ge E\sup_{g\in\tilde{\mathcal{G}}} \sum_{i=1}^n [g(Z_i)-Eg(Z_i)]+r\right)  \\
\le K_2 \exp\left(-\frac{1}{K_2}\frac{r^2}{n E[G(Z_i)^2] + 32\overline g K_1\sqrt{n} \sqrt{E[G(Z_i)]^2} +\overline g r}\right),
\end{multline*}
where the first inequality follows from (\ref{vw_2.14.1_bound}).
Substituting $r=\sqrt{n}t$ gives
\begin{multline*}
P\left(\sup_{g\in\tilde{\mathcal{G}}} \mathbb{G}_n g(Z_i)
  \ge K_1 \sqrt{E[G(Z_i)]^2}+ t\right)  \\
%\le K_2 \exp\left(-\frac{1}{K_2}\frac{nt^2}{n E[G(Z_i)^2] + 32\overline g K_1\sqrt{n} %\sqrt{E[G(Z_i)]^2} +\overline g \sqrt{n}t}\right)  \\
%&=
\le K_2 \exp\left(-\frac{1}{K_2}\frac{t^2}{E[G(Z_i)^2] + 32\overline g K_1 \sqrt{E[G(Z_i)]^2}/\sqrt{n} +\overline g t/\sqrt{n}}\right),
\end{multline*}
which gives the result after noting that replacing $K_1$ on the left hand side
as well as $K_2$ and $32 K_1K_2$ on the right hand side with a larger constant
$K$ decreases the left hand side and increases the right hand side, and applying
a symmetric bound to $\inf_{g\in\tilde{\mathcal{G}}} \mathbb{G}_n g(Z_i)$.
\end{proof}

Lemma~\ref{talagrand_ineq_lemma} gives good bounds for $t$ just larger than
$\sqrt{E[G(Z_i)]^2}$, so long as $\sqrt{E[G(Z_i)]^2}/\sqrt{n}$ is small relative
to $E[G(Z_i)]^2$ (i.e.\ so long as $E[G(Z_i)]^2n$ is large). We now state a
version of this result that is specialized to this case.

\begin{lemma}\label{talagrand_var_bound_lemma}
  Let $\tilde{\mathcal{G}}$ be a subset of $\mathcal{G}$ such that, for some
  envelope function $G$ and constant $\overline g$, $|g(Z_i)|\le G(Z_i)\le
  \overline g$ a.s.\ for all $g\in\tilde{\mathcal{G}}$. Then, for some constant
  $K$ that depends only on $\mathcal{G}$,
\begin{equation*}
  P\left(
  \sup_{g\in\tilde{\mathcal{G}}} \left|\mathbb{G}_{n}g(Z_i)\right|
  \ge \sqrt{V} a \right)
  \le K\exp\left(-\frac{a^2}{K}\right)
\end{equation*}
for all $V\ge E[G(Z_i)^2]$ and $a>0$ with $a+1\le \sqrt{V}\sqrt{n}/\overline g$.
\end{lemma}
\begin{proof}
  Substituting $t=r V^{1/2}$ into the bound from
  Lemma~\ref{talagrand_ineq_lemma} gives, letting $K_1$ be the constant $K$ from
  that lemma,
\begin{align*}
P\left(\sup_{g\in\tilde{\mathcal{G}}} \left|\mathbb{G}_{n}g(Z_i)\right|
  \ge (K_1+r)V^{1/2}  \right)
  \le K_1\exp\left(-\frac{1}{K_1}\frac{r^2 V}{V+\overline g \left\{V^{1/2}+r V^{1/2}\right\}/\sqrt{n}}
\right).
\end{align*}
For $\overline g(1+r)\le \sqrt{n}V^{1/2}$, this is bounded by
$K_1\exp\left(-\frac{r^2}{2K_1}\right)$.  Setting $a=K_1+r$ and noting that $K_1\exp\left(-\frac{(a-K_1)^2}{2K_1}\right)\le K_2\exp\left(-\frac{a^2}{K_2}\right)$ for a large enough constant $K_2$
(and that $\overline g(1+a)\le \sqrt{n}V^{1/2}$ implies $\overline g(1+a-K_1)\le \sqrt{n}V^{1/2}$) gives the result.
\end{proof}

\subsection{Tail Bounds for Kernel Estimators}\label{kern_tail_sec}

We specialize some of the results of Section~\ref{tail_bounds_sec} to our
setting. We are interested in functions of the form $g(x,w)=f(w,h,t)k(x/h)$,
where $h$ varies over positive real numbers and $t$ varies over some index set
$T$.
% (i.e. we will apply our results with $f(w,h)=\phi(W_i,\theta(h))-\phi(W_i\theta(0))$).

We assume throughout the section that $k(x)$ is a bounded kernel function with support $[-A,A]$, with $k(x)\le B_k<\infty$ for all $k$.
%, and that the class $x\mapsto k(x/h)$ has polynomial covering number.
We also assume that $X_i$ is a real valued random variable with with a density $f_X(x)$ with $f_X(x)\le \overline f_X<\infty$ all $x$.
%note: need to add verification of polynomial covering number of products of classes or mention

\begin{lemma}\label{kern_bound_lemma}
Suppose that $\{(x,w)\mapsto f(w,h,t)k(x/h)|0\le h\le \overline h,t\in T\}$ is contained in some larger class $\mathcal{G}$ with polynomial covering number, and that, for some constant $B_f$, $|f(W_i,h,t)k(X_i/h)|\le B_f$ for all $h\le \overline h$ and $t\in T$ with probability one.  Then, for some constant $K$ that depends only on $\mathcal{G}$,
%nondecreasing deterministic function $b(h)$, $f(W_i,h)\le b(h)$ with probability one.
\begin{align*}
  P\left(\sup_{0\le h\le \overline h,t\in T} |\mathbb{G}_{n}f(W_i,h,t)k(X_i/h)|
    \ge a B_{f}A^{1/2}\overline f_X^{1/2}\overline h^{1/2}\right) \le
  K\exp(-\frac{a^2}{K})
\end{align*}
for all $a>0$ with $a+1\le A^{1/2}\overline f_X^{1/2}\overline h^{1/2}n^{1/2}$.
\end{lemma}
\begin{proof}
  The result follows from Lemma~\ref{talagrand_var_bound_lemma}, since
  $B_{f}I(|X_i|\le A \overline h)$ is an envelope function for
  $f(W_i,h,t)k(X_i/h)$ as $h$ and $t$ vary over this set.
\end{proof}

\begin{lemma}\label{lil_bound_lemma_supp}
  Suppose that the conditions of Lemma~\ref{kern_bound_lemma} hold, and let
  $a(h)=2\sqrt{K\log\log (1/h)}$ where $K$ is the constant from
  Lemma~\ref{kern_bound_lemma}. Then, for a constant $\varepsilon>0$ that
  depends only on $K$, $A$ and $\overline f_X$,
\begin{align*}
&P\left(|\mathbb{G}_{n}f(W_i,h,t)k(X_i/h)|
  \ge a(h) h^{1/2} B_{f}A^{1/2}\overline f_X^{1/2} \text{ some $(\log\log n)/(\varepsilon n)\le h\le \overline h$, $t\in T$}\right)  \\
&\le K(\log 2)^{-2}\sum_{(2\overline h)^{-1}\le 2^k\le \infty} k^{-2}.
\end{align*}
\end{lemma}
\begin{proof}
  Let $\mathcal{H}^k=(2^{-(k+1)},2^{-k})$. Applying Lemma~\ref{kern_bound_lemma}
  to this set, we have
\begin{align*}
&P\left(|\mathbb{G}_{n}f(W_i,h,t)k(X_i/h)|
  \ge a(h) h^{1/2} B_{f}A^{1/2}\overline f_X^{1/2} \text{ some $h\in \mathcal{H}^k$, $t\in T$}\right)  \\
&\le P\left(\sup_{0\le h\le 2^k,t\in T} |\mathbb{G}_{n}f(W_i,h,t)k(X_i/h)|
  \ge a(2^{-k}) 2^{-(k+1)/2} B_{f}A^{1/2}\overline f_X^{1/2}\right)  \\
&\le K\exp\left(-\frac{[a(2^{-k}) 2^{-1/2}]^2}{K}\right)
%=K\exp\left(-\frac{1}{2K}4 K\log\log 2^{k}\right)
=K\exp\left(-2\log\log 2^{k}\right)
=K\exp\left(-2\log(k\log 2)\right)
=K[k\log 2]^{-2}
\end{align*}
so long as $2^{-1/2}a(2^{-k})+1\le A^{1/2}\overline f_X^{1/2}2^{-k/2}n^{1/2}$, where the first inequality follows since $a(h)\ge a(2^{-k})$ and $h\ge 2^{-(k+1)}$ for $h\in\mathcal{H}^k$.

Now, $2^{-1/2}a(2^{-k})+1\le A^{1/2}\overline f_X^{1/2}2^{-k/2}n^{1/2}$ will
hold iff. $[2^{-1/2}a(2^{-k})+1]2^{k/2}\le A^{1/2}\overline f_X^{1/2}n^{1/2}$.
If $2^{k}\le \varepsilon n/\log\log n$ for some $\varepsilon>0$, we will have
$a(2^{-k})\le 2\sqrt{K\log\log [\varepsilon n/\log\log n]}$, so that
$[2^{-1/2}a(2^{-k})+1]2^{k/2}\le \{2^{-1/2}\cdot 2\sqrt{K\log\log [\varepsilon
  n/\log\log n]}+1\}\sqrt{\varepsilon n/\log\log n}$. For large enough $n$, this
is bounded by $4\sqrt{K\varepsilon n}$, which is less than
$A^{1/2}\overline f_X^{1/2}n^{1/2}$ for $\varepsilon$ small enough as required.

Thus, for $\varepsilon$ defined above,
\begin{multline*}
  P\left(|\mathbb{G}_{n}f(W_i,h,t)k(X_i/h)|
    \ge a(h) h^{1/2} B_{f}A^{1/2}\overline f_X^{1/2} \text{ some $(\log\log n)/(\varepsilon n)\le h\le \overline h$, $t\in T$}\right)  \\
  \le \sum_{(2\overline h)^{-1}\le 2^k\le 2\varepsilon n/\log\log n}
  P\left(\sup_{0\le h\le 2^k,t\in T} |\mathbb{G}_{n}f(W_i,h,t)k(X_i/h)|
    \ge a(2^{-k}) 2^{-(k+1)/2} B_{f}A^{1/2}\overline f_X^{1/2}\right)  \\
  \le K(\log 2)^{-2}\sum_{(2\overline h)^{-1}\le 2^k\le 2\varepsilon n/\log\log
    n} k^{-2},
\end{multline*}
which gives the result.
\end{proof}

Using these bounds, we obtain the following uniform bound on
$\mathbb{G}_{n}f(W_i,h,t)k(X_i/h)$.

\begin{lemma}\label{lil_rate_lemma_supp}
Under the conditions of Lemma~\ref{lil_bound_lemma_supp},
\begin{equation*}
\sup_{(\log\log n)/(\varepsilon n)\le h\le \overline h, t\in T} \frac{|\mathbb{G}_{n}f(W_i,h,t)k(X_i/h)|}{(\log\log h^{-1})^{1/2}h^{1/2}}
=\mathcal{O}_P(1).
\end{equation*}
\end{lemma}
\begin{proof}
Given $\varepsilon>0$, we can apply Lemma~\ref{lil_bound_lemma_supp} to find a $\delta>0$ such that
\begin{equation*}
\sup_{(\log\log n)/(\varepsilon n)\le h\le \delta, t\in T} \frac{|\mathbb{G}_{n}f(W_i,h,t)k(X_i/h)|}{(\log\log h^{-1})^{1/2}h^{1/2}}
<2\sqrt{2K} B_{f}A^{1/2} \overline f_X^{1/2}
\end{equation*}
with probability at least $1-K(\log 2)^{-2}\sum_{(2\delta)^{-1}\le 2^k\le
  \infty}k^{-2}>1-\varepsilon/2$. For this choice of $\delta$,
\begin{equation*}
  \sup_{\delta \le
    h\le \overline h, t\in T} \frac{|\mathbb{G}_{n}f(W_i,h,t)k(X_i/h)|}{(\log\log
    h^{-1})^{1/2}h^{1/2}}=\mathcal{O}_P(1)
\end{equation*}
by Lemma~\ref{kern_bound_lemma}. Thus, choosing $C$ large enough so that $C\ge
2\sqrt{2K} B_{f}A^{1/2} \overline f_X^{1/2}$ and
\begin{equation*}
  \sup_{\delta \le h\le \overline h, t\in T}
  \frac{|\mathbb{G}_{n}f(W_i,h,t)k(X_i/h)|}{(\log\log h^{-1})^{1/2}h^{1/2}}\le C
\end{equation*}
with probability at least $1-\varepsilon/2$ asymptotically, we have
\begin{equation*}
  \sup_{(\log\log n)/(\varepsilon n)\le h\le \overline h, t\in T}
  \frac{|\mathbb{G}_{n}f(W_i,h,t)k(X_i/h)|}{(\log\log h^{-1})^{1/2}h^{1/2}}\le C
\end{equation*}
with probability at least $1-\varepsilon$ asymptotically.
\end{proof}

\subsection{Gaussian Approximation}\label{gauss_approx_sec_supp}

This section proves Theorem~\ref{gauss_approx_thm} in
Appendix~\ref{gauss_approx_sec}, which gives a Gaussian process approximation
for the process $\hat{\mathbb{H}}_n(h)$ defined in that section.

For convenience, we repeat the setup here. We show that
$\frac{1}{\sqrt{h}}\mathbb{G}_n\tilde Y_{i}k(X_i/h)
=\frac{1}{\sqrt{nh}}\sum_{i=1}^n\tilde Y_{i}k(X_i/h)$ is approximated by a
Gaussian process with the same covariance kernel. We consider a general setup
with $\{(\tilde X_i,\tilde Y_i)\}_{i=1}^n$ i.i.d., with $\tilde X_i\ge 0$ a.s.\
such that $\tilde X_i$ has a density $f_{\tilde X}(x)$ on $[0,\overline x]$ for
some $\overline x\ge 0$, with $f_{\tilde X}(x)$ bounded away from zero and
infinity on this set. We assume that $\tilde Y_i$ is bounded almost surely, with
$E(\tilde Y_i|\tilde X_i)=0$ and
$var(\tilde Y_i|\tilde X_i=x)=f_{\tilde{X}}(x)^{-1}$. We assume that the kernel
function $k$ has finite support $[0,A]$ and is differentiable on its support
with bounded derivative. For ease of notation, we assume in this section that
$\int k(u)^2\, du=1$. The result applies to our setup with $\tilde Y_i$ given in
(\ref{tilde_yi_eq}) in Appendix~\ref{sec:proof-main-result} in the main text and
$\tilde X_i$ given by $|X_i|$.

%begin material from Gaussian approximation note

Let
\begin{align*}
\hat{\mathbb{H}}_n(h)=\frac{1}{\sqrt{nh}}\sum_{i=1}^n\tilde Y_{i}k(\tilde X_i/h).
\end{align*}

\newtheorem*{thm:gaus_approx_thm}{Theorem~\ref{gauss_approx_thm}}
\begin{thm:gaus_approx_thm}%\label{gauss_approx_thm}
Under the conditions above, there exists, for each $n$, a process $\mathbb{H}_n(h)$ such that, conditional on $(\tilde X_1,\ldots,\tilde X_n)$, $\mathbb{H}_n$ is a Gaussian process with covariance kernel
\begin{align*}
cov\left(\mathbb{H}_n(h),\mathbb{H}_n(h')\right)
  =\frac{1}{\sqrt{h h'}}\int k(x/h)k(x/h')\, dx
\end{align*}
and
\begin{align*}
  \sup_{\underline h_n\le h\le \overline x/A} \left|\hat{\mathbb{H}}_n(h)-\mathbb{H}_n(h)\right|
  =\mathcal{O}_P\left(
  (n\underline h_n)^{-1/4}[\log (n \underline h_n)]^{1/2}\right)
\end{align*}
for any sequence $\underline h_n$ with $n\underline h_n/ \log\log \underline h_n^{-1}\to\infty$.
\end{thm:gaus_approx_thm}

We now prove the result. Let
$\hat G(x)=\frac{1}{n}\sum_{\tilde X_i\le x}\tilde Y_i$. With this notation, we
can write the process $\hat{\mathbb{H}}_n(h)$ as
\begin{equation*}
\hat{\mathbb{H}}_n(h)
=\frac{1}{\sqrt{nh}}\sum_{i=1}^n\tilde Y_{i}k(\tilde X_i/h)
=\frac{\sqrt{n}}{\sqrt{h}}\int k(x/h)\, d\hat G(x).
\end{equation*}
Let $\hat g(x)=\frac{1}{n}\sum_{\tilde X_i\le x} f_{\tilde X}(\tilde X_i)^{-1}$.
In Lemma~\ref{sakhanenko_approx_lemma} below, a process $\mathbb{B}_n(t)$ is
constructed that is a Brownian motion conditional on $\tilde X_1,\ldots,
\tilde{X}_n$ such that $\mathbb{B}_n(n\hat g(x))$ is, with high probability
conditional on $\tilde X_1,\ldots,\tilde X_n$, close to $n\hat G(x)$. By showing
that $\hat g(x)$ is close to $x$ with high probability and using properties of
the fluctuation of the Brownian motion, it is then shown that
$\mathbb{B}_n(n\hat g(x))$ can be approximated by $\mathbb{B}_n(nx)$, so that
$\hat{\mathbb{H}}_n(h)$ is approximated by the corresponding process with
$\hat{G}(x)$ replaced by $\mathbb{B}_n(nx)/n$.

Formally, let $\mathbb{B}_n(t)$ be given by the (conditional) Brownian motion in
Lemma~\ref{sakhanenko_approx_lemma} below, and define
\begin{align*}
\mathbb{H}_n(h)=\frac{1}{\sqrt{nh}}\int k(x/h)\, d \mathbb{B}_n(nx).
\end{align*}
Note that $\mathbb{H}_n(h)=\frac{1}{\sqrt{h}}\int k(x/h)\, d \tilde{\mathbb{B}}_n(x)$
(where $\tilde{\mathbb{B}}_n(x)=\mathbb{B}_n(nx)/\sqrt{n}$ is another Brownian motion conditional on $\tilde X_1,\ldots,\tilde X_n$), so that, conditional on $(\tilde X_1,\ldots,\tilde X_n)$, $\mathbb{H}_n$ is a Gaussian process with the desired covariance kernel.
%to do: add formal justification

Let $R_{1,n}(x)=n\hat G(x)-\mathbb{B}_n(n\hat g(x))$ and
$R_{2,n}(x)=\mathbb{B}_n(n\hat g(x))-\mathbb{B}_n(nx)$.
Then
\begin{equation*}
\hat{\mathbb{H}}_n(h)-\mathbb{H}_n(h)
=\frac{1}{\sqrt{nh}}\int k(x/h)\, d R_{1,n}(x)
+\frac{1}{\sqrt{nh}}\int k(x/h)\, d R_{2,n}(x).
\end{equation*}
Using the integration by parts formula, we have, for $j=1,2$ and $Ah\le \overline x$,
\begin{equation*}
\frac{1}{\sqrt{nh}}\int k(x/h)\, d R_{j,n}(x)
%=\frac{1}{\sqrt{nh}}\int_{x=0}^{Ah}k(x/h)\, d R_{j,n}(x)  \\
%&
=\frac{R_{j,n}(Ah)k(A)}{\sqrt{nh}}
  -\frac{1}{\sqrt{nh}}\int_{x=0}^{Ah}R_{j,n}(x)k'(x/h)\frac{1}{h}\, dx
\end{equation*}
The first term is bounded by $\frac{\left| R_{j,n}(Ah)\right|k(A)}{\sqrt{nh}}$, and the second term is bounded by
\begin{equation*}
\frac{A}{\sqrt{nh}}\left(\sup_{0\le x\le Ah}\left|R_{j,n}(x)\right|\right)
\left(\sup_{0\le u\le A}\left|k'(u)\right|\right)
\end{equation*}
\citep[see][for a similar derivation]{bickel_global_1973}. By boundedness of
$k'(u)$, it follows that both terms are bounded by a constant times
$\frac{1}{\sqrt{nh}}\sup_{0\le x\le Ah} \left|R_{j,n}(x)\right|$, so that
\begin{equation*}
  \sup_{\underline h_n\le h\le \overline x/A}\left|\hat{\mathbb{H}}_n(h)
    -\mathbb{H}_n(h)\right|
    \le K\sup_{\underline h_n\le h\le \overline x/A}\sum_{j=1}^2\sup_{0\le x\le Ah} \frac{\left|R_{j,n}(x)\right|}{\sqrt{nh}}
    \le K\sum_{j=1}^2\sup_{0\le x\le \overline x} \frac{\left|R_{j,n}(x)\right|}{\sqrt{n[(x/A)\vee \underline h_n]}}.
\end{equation*}
for some constant $K$.  Thus, the result will follow if we can show that $\sup_{0\le x\le \overline x}\frac{\left|R_{1,n}(x)\right|}{\sqrt{n(x\vee \underline h_n)}}$ and $\sup_{0\le x\le \overline x}\frac{\left|R_{2,n}(x)\right|}{\sqrt{n(x\vee \underline h_n)}}$ converge to zero at the required rate.

We first construct $\mathbb{B}_n(t)$ and show that $\sup_{0\le x\le A\overline
  x/A}\frac{\left|R_{1,n}(x)\right|}{\sqrt{n(x\vee \underline h_n)}}$ converges
to zero quickly enough with this construction, using an approximation of
Sakhanenko. Denote the the empirical cdf of $\tilde X_i$ by
$\hat{F}_{\tilde{X}}(x)=\frac{1}{n}\sum_{i=1}^n I(\tilde X_i\le x)$, and let
$\tilde X_{(k)}$ be the $k$th smallest value of $\tilde X_i$.

\begin{lemma}\label{sakhanenko_approx_lemma}
  Under the conditions of Theorem~\ref{gauss_approx_thm}, one can construct
  variables $Z_1,\ldots,Z_n$ such that $Z_i|(\tilde X_1,\ldots,\tilde X_n)\sim
  N(0,f_{\tilde X}(\tilde X_i)^{-1})$ and
\begin{equation*}
P\left(
\left|\sum_{\tilde X_i\le x}Z_i-\sum_{\tilde X_i\le x}\tilde Y_i\right|> K\log \left[n\hat F_{\tilde X}(x)+2\right] \text{ some } 0\le x\le \overline x  \bigg|\tilde X_1,\ldots,\tilde X_n\right)
\le \varepsilon(K)
\end{equation*}
with probability one,
where $\varepsilon(K)$ is a deterministic function with $\varepsilon(K)\to 0$ as $K\to\infty$.
\end{lemma}
\begin{proof}
Using a result of \citet{sakhanenko_convergence_1985} as stated in
Theorem A of \citet{shao_strong_1995}, we can construct $Z_1,\ldots,Z_n$ such that
\begin{align*}
E\exp\left(\lambda A \sup_{0\le x\le \tilde X_{(k)}} \left|\sum_{\tilde X_i\le x}Z_i-\sum_{\tilde X_i\le x}\tilde Y_i\right|  \bigg|\tilde X_1,\ldots,\tilde X_n\right)
\le 1+\lambda\sum_{\tilde X_i\le \tilde X_{(k)}} f_{\tilde X}(\tilde X_i)^{-1}
\end{align*}
where $A$ is a universal constant and $\lambda$ is any constant such that
$\lambda E[\exp(\lambda |\tilde Y_i|)|\tilde Y_i|^3|\tilde X_i]\le E[\tilde Y_i^2|\tilde X_i]$.  Let $\overline Y$ be a bound for $\tilde Y_i$.  Then
$\lambda E[\exp(\lambda |\tilde Y_i|)|\tilde Y_i|^3|\tilde X_i]
\le \lambda \exp(\lambda \overline Y)\overline Y E[|\tilde Y_i|^2|\tilde X_i]$, so the inequality holds for any $\lambda$ with $\lambda \exp(\lambda \overline Y)\overline Y \le 1$.  From now on, we fix $\lambda>0$ so that this inequality holds.

Letting $\underline f_{\tilde X}$ be a lower bound for $f_{\tilde X}(x)$ over $0\le x\le\overline x$ and applying Markov's inequality, the above bound gives
\begin{align*}
&P\left(\lambda A \sup_{0\le x\le \tilde X_{(k)}} \left|\sum_{\tilde X_i\le x}Z_i-\sum_{\tilde X_i\le x}\tilde Y_i\right|> t  \bigg|\tilde X_1,\ldots,\tilde X_n\right)  \\
&\le \exp(-t)E\exp\left(\lambda A \sup_{0\le x\le \tilde X_{(k)}} \left|\sum_{\tilde X_i\le x}Z_i-\sum_{\tilde X_i\le x}\tilde Y_i\right|  \bigg|\tilde X_1,\ldots,\tilde X_n\right)
\le \exp(-t)(1+\lambda\underline f_{\tilde X}^{-1} k).
\end{align*}
Thus,
\begin{align*}
&P\left(
\left|\sum_{\tilde X_i\le x}Z_i-\sum_{\tilde X_i\le x}\tilde Y_i\right|> K\log \left[\sum_{i=1}^n I(\tilde X_i\le x)+2\right] \text{ some } 0\le x\le \overline x  \bigg|\tilde X_1,\ldots,\tilde X_n\right)  \\
&\le P\left(\sup_{0\le x\le \tilde X_{(k)}} \left|\sum_{\tilde X_i\le x}Z_i-\sum_{\tilde X_i\le x}\tilde Y_i\right|> K\log k \text{ some } 2\le k\le n  \bigg|\tilde X_1,\ldots,\tilde X_n\right)  \\
&\le \sum_{k=2}^n P\left(\lambda A\sup_{0\le x\le \tilde X_{(k)}} \left|\sum_{\tilde X_i\le x}Z_i-\sum_{\tilde X_i\le x}\tilde Y_i\right|\ge \lambda A K\log k \bigg|\tilde X_1,\ldots,\tilde X_n\right)  \\
&\le %\sum_{k=2}^n\exp\left(-\lambda A K\log k\right)(1+\lambda\underline f_{\tilde X}^{-1}k)=
\sum_{k=2}^{n}k^{-\lambda A K}(1+\lambda\underline f_{\tilde X}^{-1}k)
\le \sum_{k=2}^\infty k^{-\lambda A K}(1+\lambda\underline f_{\tilde X}^{-1}k),
\end{align*}
which can be made arbitrarily small by making $K$ large.
\end{proof}

Embedding $\sum_{\tilde X_i\le x} Z_i$ in a Brownian motion, we can restate the above construction as follows:
with probability at least $1-K(\varepsilon)$ conditional on $\tilde X_1,\ldots, \tilde X_n$,
\begin{align*}
\left|n\hat G(x)-\mathbb{B}_n(n \hat g(x))
\right|\le K\log [n\hat F_{\tilde X}(x)+2]
\text{ all } 0\le x\le \overline x
\end{align*}
where $\mathbb{B}_n(t)=\mathbb{B}_n(t;\tilde X_1,\ldots,\tilde X_n)$ is a Brownian motion conditional on $\tilde X_1,\ldots,\tilde X_n$.  Let $\overline f_{\tilde X}$ be an upper bound for the density of $\tilde X_i$ on $[0,\overline x]$.

\begin{lemma}
Under the conditions of Theorem~\ref{gauss_approx_thm}, for any $\eta>0$,
\begin{align*}
\hat F_{\tilde X}(x)\le \overline f_{\tilde X}\cdot (1+\eta) (x\vee\underline h_n)
\end{align*}
for all $0\le x \le \overline x$ with probability approaching one.
\end{lemma}
\begin{proof}
By Lemma~\ref{lil_rate_lemma_supp},
%to do: should probably just cite lil for empirical processes if it is easy to find
\begin{align*}
\sup_{\underline h_n\le x\le \overline x} \frac{\sqrt{n}|\hat F_{\tilde X}(x)-F_{\tilde X}(x)|}{\sqrt{x\log\log x^{-1}}}
=\mathcal{O}_P(1).
\end{align*}
Thus,
\begin{align*}
&\sup_{\underline h_n\le x\le \overline x} \frac{|\hat F_{\tilde X}(x)-F_{\tilde X}(x)|}{x}
=\sup_{\underline h_n\le x\le \overline x} \frac{\sqrt{n}|\hat F_{\tilde X}(x)-F_{\tilde X}(x)|}{\sqrt{x\log\log x^{-1}}}
  \frac{\sqrt{x\log\log x^{-1}}}{\sqrt{n}x}  \\
&=\mathcal{O}_P\left(\sup_{\underline h_n\le x\le \overline x}\frac{\sqrt{\log\log x^{-1}}}{\sqrt{nx}}\right)
=\mathcal{O}_P\left(\frac{\sqrt{\log\log \underline h_n^{-1}}}{\sqrt{n \underline h_n}}\right)
=o_P(1)
\end{align*}
where the last step follows since
$n\underline h_n/ \log\log \underline h_n^{-1}\to\infty$. Thus, for any
$\eta>0$, we have, with probability approaching one,
\begin{align*}
\hat F_{\tilde X}(x)\le \hat F_{\tilde X}(x\vee \underline h_n)\le
F_{\tilde X}(x\vee \underline h_n)+(\eta \overline f_{\tilde X})(x\vee\underline h_n)
\le \overline f_{\tilde X}\cdot (1+\eta)(x\vee \underline h_n)
\end{align*}
for all $x$.
\end{proof}

Combining these two lemmas, we have, for large enough $n$,
\begin{multline*}
\limsup_n P\left(
\left|n\hat G(x)-\mathbb{B}_n(n\hat g(x))\right|> K\log \left[2n\overline f_{\tilde X} (x\vee \underline h_n)+2\right] \text{ some } 0\le x\le \overline x\right)  \\
\le \varepsilon(K)+\limsup_n P\left(
\hat F_{\tilde X}(x)> \overline f_{\tilde X}\cdot 2 (x\vee\underline h_n)
\right)\le \varepsilon(K).
\end{multline*}
Since this can be made arbitrarily small by making $K$ large, it follows that
\begin{align*}
\sup_{0\le x\le \overline x}\frac{\left|n\hat G(x)-\mathbb{B}_n(n\hat g(x))\right|}{\sqrt{n(x\vee \underline h_n)}}
=\mathcal{O}_P\left(
\sup_{0\le x\le \overline x} \frac{\log \left[2n\overline f_{\tilde X} (x\vee \underline h_n)+2\right]}{\sqrt{n(x\vee \underline h_n)}}\right)
=\mathcal{O}_P\left(\frac{\log (n\underline h_n)}{\sqrt{n\underline h_n}}\right),
\end{align*}
which gives the required rate for $R_{1,n}(x)$.

Define the function $LL(x)=\log\log x$ for $\log\log x\ge 1$ and $LL(x)=1$ otherwise.
Given $K$, let $B_n(K)$ be the event that
\begin{equation*}
|n\hat g(x)-nx|\le K\sqrt{n(x\vee \underline h_n) LL(x/\underline h_n)}  \text{ all $0\le x\le \overline x$},
\end{equation*}
%to do: figure out form of event B_n(K), finish proof with this form of the event
and let $C_n(K)$ be the event that
\begin{align*}
\left|\mathbb{B}_n(t')-\mathbb{B}_n(t)\right|
%\le K \sqrt{|t'-t|\cdot |\log |t'-t||} \text{ all $0\le t<\infty$}.
\le K \sqrt{(|t'-t|\vee 1)\cdot \log (t\vee t'\vee 2)} \text{ all $0\le t,t'<\infty$}.
%\le K \sqrt{(|t'-t|\vee \sqrt{n\underline h_n})\cdot \log [(t/\sqrt{n\underline h_n})\vee (t'/\sqrt{n\underline h_n})\vee 2]} \text{ all $0\le t,t'<\infty$}.
\end{align*}

\begin{lemma}
On the event $B_n(K)\cap C_n(K)$, for large enough $n$,
\begin{align*}
&\frac{|R_{2,n}(x)|}{\sqrt{n(x\vee \underline h_n)}}
\le K^{3/2}[n(x\vee \underline h_n)]^{-1/4}\{LL(x/\underline h_n)\}^{1/4}
\cdot \{\log 2+\log [n(x\vee \underline h_n)]\}^{1/2}  \\
&\le K^{3/2}(n\underline h_n)^{-1/4}
\cdot \{\log 2+\log [n\underline h_n]\}^{1/2}
\end{align*}
%to do: fix this display, make sure x^{-1/4}LL(x) is decreasing in x
for all $0\le x\le \overline x$.
\end{lemma}
\begin{proof}
On this event, for all $0\le x\le \overline x$ and large enough $n$,
\begin{align*}
|R_{2,n}(x)|&=|\mathbb{B}_n(n\hat g(x))-\mathbb{B}_n(nx)|
\le \sup_{|t-nx|\le K\sqrt{n(x\vee \underline h_n) LL(x/\underline h_n)}}
  |\mathbb{B}_n(t)-\mathbb{B}_n(nx)|  \\
&\le \sup_{|t-nx|\le K\sqrt{n(x\vee \underline h_n) LL(x/\underline h_n)}}
  K \sqrt{(|t-nx|\vee 1)\cdot \log [t\vee (nx)\vee 2]}  \\
&\le K \sqrt{K\sqrt{n(x\vee \underline h_n) LL(x/\underline h_n)}\cdot \log [2n(x\vee \underline h_n)]}  \\
&=K^{3/2}n^{1/4}(x\vee \underline h_n)^{1/4}\{LL(x/\underline h_n)\}^{1/4}
\cdot \{\log 2+\log [n(x\vee \underline h_n)]\}^{1/2}.
\end{align*}
\end{proof}

\begin{lemma}
Under the conditions of Theorem~\ref{gauss_approx_thm}, for any $\varepsilon>0$, there exists a $K$ such that
$P(B_n(K))\ge 1-\varepsilon$
for large enough $n$.
\end{lemma}
\begin{proof}
Let
$\mathcal{X}^k=(2^k \underline h_n,2^{k+1}\underline h_n]\cap [0,\overline x]$.  We have, for $k\ge 2$,
\begin{align*}
&P\left(|n\hat g(x)-nx|> K\sqrt{n(x\vee \underline h_n) LL(x/\underline h_n)} \text{ some } x\in\mathcal{X}^k\right)  \\
%&=P\left(|n\hat g(x)-nx|> K\sqrt{n x\cdot LL(x/\underline h_n)} \text{ some } x\in\mathcal{X}^k\right)  \\
%&=P\left(\sqrt{n}|\hat g(x)-x|> K\sqrt{x\cdot LL(x/\underline h_n)} \text{ some } x\in\mathcal{X}^k\right)  \\
&=P\left(|\mathbb{G}_{n}f(\tilde X_i)^{-1}I(\tilde X_i\le x)|> K\sqrt{x\cdot LL(x/\underline h_n)} \text{ some } x\in\mathcal{X}^k\right)  \\
&\le P\left(\sup_{x\in\mathcal{X}^k}|\mathbb{G}_{n}f(\tilde X_i)^{-1}I(\tilde X_i\le x)|> K\sqrt{2^k\underline h_n\cdot LL(2^k)} \right)  \\
&\le C\exp\left(-\frac{K^2 LL(2^k)}{C}\right)
\le C \exp\left(-\frac{K^2}{C} \log \log (2^k)\right)
=C [k \log 2]^{-\frac{K^2}{C}}
\end{align*}
for some constant $C$ by Lemma~\ref{kern_bound_lemma}.
Thus,
\begin{align*}
&P\left(|n\hat g(x)-nx|> K\sqrt{n(x\vee \underline h_n) LL(x/\underline h_n)} \text{ some } 4\underline h_n\le x\le\overline x\right)
\le C \sum_{k=2}^{\infty} [k\log 2]^{-K^2/C}
\end{align*}
which can be made arbitrarily small by making $K$ large.  Note also that
\begin{align*}
&P\left(|n\hat g(x)-nx|> K\sqrt{n(x\vee \underline h_n) LL(x/\underline h_n)} \text{ some } 0\le x\le 4\underline h_n\right)  \\
&\le P\left(\sup_{0\le x\le 4\underline h_n} |\mathbb{G}_{n}f(\tilde X_i)^{-1}I(\tilde X_i\le x)|
>K\sqrt{\underline h_n}\right),
\end{align*}
which can also be made arbitrarily small by choosing $K$ large by Lemma~\ref{kern_bound_lemma}.
Combining these bounds gives the result.
\end{proof}

\begin{lemma}
  Under the conditions of Theorem~\ref{gauss_approx_thm}, for any
  $\varepsilon>0$, there exists a $K$ such that with probability one for all
  $n$, $P(C_n(K)|\tilde X_1,\ldots,\tilde X_n)\ge 1-\varepsilon$.
\end{lemma}
\begin{proof}
We have
\begin{multline*}
  1-P(C_n(K)|\tilde X_1,\ldots,\tilde X_n)\\
  =P\left(\left|\mathbb{B}_n(t')-\mathbb{B}_n(t)\right|> K\sqrt{(|t-t'|\vee
      1)\cdot \log (t\vee t'\vee 2)}
    \text{ some } 0\le t,t'< \infty\right)  \\
  =P\left(\left|\mathbb{B}_n(t+s)-\mathbb{B}_n(t)\right|> K\sqrt{(s\vee 1)\cdot
      \log [(t+s)\vee 2]}
    \text{ some } 0\le s,t< \infty\right)  \\
  \le \sum_{k=0}^\infty\sum_{\ell=0}^\infty
  P\left(\left|\mathbb{B}_n(t+s)-\mathbb{B}_n(t)\right|> K\sqrt{(s\vee 1)\cdot
      \log [(t+s)\vee 2]} \text{ some } (s,t)\in \mathcal{S}_{k,\ell}\right)
\end{multline*}
\allowdisplaybreaks
where $\mathcal{S}_{k,\ell}=\{(s,t)|\ell \le s\le \ell+1,(\ell\vee 1) k\le t\le (\ell\vee 1)(k+1)\}$.  Note that
\begin{align*}
&P\left(\left|\mathbb{B}_n(t+s)-\mathbb{B}_n(t)\right|>
  K\sqrt{(s\vee 1)\cdot \log [(t+s)\vee 2]}
\text{ some } (s,t)\in \mathcal{S}_{k,\ell}\right)  \\
&\le P\left(\left|\mathbb{B}_n(t+s)-\mathbb{B}_n(t)\right|>
  K\sqrt{(\ell \vee 1)\cdot \log \{[(\ell\vee 1)k+\ell]\vee 2\}}
\text{ some } (s,t)\in \mathcal{S}_{k,\ell}\right)  \\
&=P\left(\left|\mathbb{B}_n(t+s)-\mathbb{B}_n(t)\right|>
  K\sqrt{(\ell \vee 1)\cdot \log \{[(\ell\vee 1)k+\ell]\vee 2\}}
\text{ some } (s,t)\in \mathcal{S}_{0,\ell}\right)  \\
&\le P\left(\left|\mathbb{B}_n(t)\right|>
  (K/2)\sqrt{(\ell \vee 1)\cdot \log \{[(\ell\vee 1)k+\ell]\vee 2\}}
\text{ some } 0\le t\le (\ell\vee 1)+\ell+1\right)  \\
&\le 4 P\left(\left|\mathbb{B}_n((\ell\vee 1)+\ell+1)\right|>
  (K/2)\sqrt{(\ell \vee 1)\cdot \log \{[(\ell\vee 1)k+\ell]\vee 2\}} \right)  \\
&\le 4 \cdot \frac{1}{\sqrt{2\pi}} \cdot
\exp\left(-\frac{1}{2}\frac{(K/2)^2(\ell \vee 1)\cdot \log \{[(\ell\vee 1)k+\ell]\vee 2\}}{(\ell\vee 1)+\ell+1}\right)  \\
&\le 4 \cdot \frac{1}{\sqrt{2\pi}} \cdot
\exp\left(-\frac{(K/2)^2\log \{[(\ell\vee 1)k+\ell]\vee 2\}}{6}\right)
= 4 \cdot \frac{1}{\sqrt{2\pi}} \cdot
  \{[(\ell\vee 1)k+\ell]\vee 2\}^{-K^2/24}.
\end{align*}
The third line follows since $\mathbb{B}_n(t)$ has the same distribution as $\mathbb{B}_n(t+(\ell\vee 1)k)$.  The fourth line follows since, if $|\mathbb{B}_n(t+s)-\mathbb{B}_n(t)|> C$ for some $C$ and $(s,t)\in\mathcal{S}_{0,\ell}$, we must have $|\mathbb{B}_n(t)|>C/2$ for some $0\le t\le (\ell\vee 1)+\ell+1$.  The fifth line follows from the reflection principle for the Brownian motion
\citep[see Theorem 2.21 in][]{morters_brownian_2010}.
%(see Proposition 9.1.10, p.345 of Dembo's STAT 310 notes).
The sixth line uses the fact that $P(Z\ge x)\le \frac{1}{\sqrt{2\pi}}\exp(-x^2/2)$ for $x\ge 1$ and $Z\sim N(0,1)$.

Thus,
\begin{align*}
&P\left(\left|\mathbb{B}_n(t')-\mathbb{B}_n(t)\right|>
  K\sqrt{(|t-t'|\vee 1)\cdot \log (t\vee t'\vee 1)}
\text{ some } 0\le t,t'< \infty\right)  \\
&\le \sum_{k=0}^\infty\sum_{\ell=0}^{\infty} 4 \cdot \frac{1}{\sqrt{2\pi}} \cdot
  \{[(\ell\vee 1)k+\ell]\vee 2\}^{-K^2/24}.
\end{align*}
This can be made arbitrarily small by making $K$ large.
%to do: check
\end{proof}

Theorem~\ref{gauss_approx_thm} now follows since, for any constant $\varepsilon>0$, there is a constant $K$ such that $\sup_{\underline h_n\le h\le \bar x/A}|\hat{\mathbb{H}}_n(h)-\mathbb{H}_n(h)|$ is less than $K\{(\log n\underline h_n)(n\underline h_n)^{-1/2}+(n\underline h_n)^{-1/4}[\log (n\underline h_n)]^{1/4}\}$ with probability at least $1-\varepsilon$ asymptotically.

%end material from Gaussian approximation note

\subsection{Calculations for Extreme Value Limit}\label{ev_calc_sec_supp}

This section provides the calculations for the asymptotic distribution derived
in Theorem~\ref{gaussian_limit_thm} in Section~\ref{gaussian_limit_sec} of the
appendix.

As described in the proof of Theorem~\ref{gaussian_limit_thm},
we use Theorem 12.3.5 of \citet{leadbetter_extremes_1983} applied to the process $\mathbb{X}(t)=\mathbb{H}(e^t)$, which is stationary,
%to do: this theorem only gives the one-sided case - need to finde two-sided version
with, in the case where $k(A)\ne 0$, $\alpha=1$ and $C=\frac{Ak(A)^2}{\int k(u)^2\, du}$ and, in the case where $k(A)=0$, $\alpha=2$ and $C=\frac{\int \left[k'(u)u+\frac{1}{2} k(u)\, du\right]^2\, du}{2\int k(u)^2\, du}$.
%to do: make sure this gives the right c_1 and c_2 with the formulas in the Leadbetter et al book with Pickands constant substituted
In the notation of that theorem, we have
\begin{equation*}
r(t)=cov\left(\mathbb{X}(s),\mathbb{X}(s+t)\right)
=\frac{e^{\frac{1}{2}t}\int k(ue^{t})k(u)\, du}{\int k(u)^2\,du}.
\end{equation*}
Since $r(t)$ is bounded by a constant times $e^{\frac{1}{2}t}\cdot e^{-t}$, the condition $r(t)\log t\stackrel{t\to\infty}{\to} 0$ holds, so it remains to verify that
$r(t)=1-C|t|^\alpha+o(|t|^\alpha)$ with $\alpha$ and $C$ given above.

Since $k(ue^{t})k(u)$ has a continuous derivative with respect to $t$ on its
support, which for $t\geq 0$ is $[-Ae^{-t},Ae^{-t}]$, it follows by Leibniz's
rule and symmetry of $k$ that, for $t\geq 0$ $\frac{d}{dt}\int k(ue^{t})k(u)\,
du=-2Ae^{-t}k(A)k(Ae^{-t})+\int k'(ue^t)k(u)ue^t\, du$ for $t\ge 0$. Thus, for
$t\ge 0$,
\begin{multline*}
  \frac{d}{dt_+}r(t)=\frac{e^{\frac{1}{2}t}\frac{d}{dt_+}\int k(ue^{t})k(u)\, du+\frac{1}{2}e^{\frac{1}{2}t}\int k(ue^{t})k(u)\, du}{\int k(u)^2\, du}  \\
  =\frac{e^{\frac{1}{2}t}\left[-2Ae^{-t}k(A)k(Ae^{-t})+\int k'(ue^t)k(u)ue^t\,
      du\right] +\frac{1}{2}e^{\frac{1}{2}t}\int k(ue^{t})k(u)\, du}{\int
    k(u)^2\, du}.
\end{multline*}
Thus,
\begin{equation*}
\frac{d}{dt_+}r(t)\bigg|_{t=0}
=\frac{-2Ak(A)^2+\int k'(u)k(u)u\, du
+\frac{1}{2}\int k(u)^2\, du}{\int k(u)^2\, du}
=\frac{-Ak(A)^2}{\int k(u)^2\, du}
\end{equation*}
where the last step follows by noting that, applying integration by parts with $k(u)u$ playing the part of $u$ and $k'(u)du$ playing the part of $dv$,
\begin{multline*}
  \int k(u)k'(u)u\, du
  =\left[k(u)^2u\right]_{-A}^A-\int k(u)[k(u)+k'(u)u]\, du  \\
  =2k(A)^2A-\int k(u)^2\, du-\int k(u)k'(u)u\, du
\end{multline*}
so that $\int k(u)k'(u)u\, du=k(A)^2A-\frac{1}{2}\int k(u)^2\, du$.  For the case where $k(A)\ne 0$, it follows from this and a symmetric argument for $t\le 0$ that
$r(t)=1-C|t|-o(|t|)$ for $C=\frac{Ak(A)^2}{\int k(u)^2\, du}$ as required.

For the case where $k(A)=0$, applying Leibniz's rule as above shows that $r(t)$ is differentiable with,
\begin{equation*}
r'(t)=e^{\frac{1}{2}t}\frac{\int k'(ue^t)k(u)ue^t\, du
+\frac{1}{2}\int k(ue^{t})k(u)\, du}{\int k(u)^2\, du}.
\end{equation*}
Thus, $r'(0)=0$ (using the integration by parts identity above)
and $r(t)$ is twice differentiable with
\begin{equation*}
r''(t)=
e^{\frac{1}{2}t}\frac{\frac{d}{dt}\int k'(ue^t)k(u)ue^t\, du
+\frac{1}{2}\left(\frac{d}{dt}\int k(ue^{t})k(u)\, du
+\int k'(ue^t)k(u)ue^t\, du
+\int k(ue^{t})k(u)\, du\right)
}{\int k(u)^2\, du}.
\end{equation*}
We have
\begin{multline*}
  \frac{d}{dt}\int k'(ue^t)k(u)ue^t\, du
  =\frac{d}{dt}\int k'(v)k(ve^{-t})ve^{-t}\, dv  \\
  =\int k'(v)k'(ve^{-t})(-ve^{-t})ve^{-t}\, dv -\int k'(v)k(ve^{-t})ve^{-t}\, dv
\end{multline*}
and
$\frac{d}{dt}\int k(ue^{t})k(u)\, du
=\int k'(ue^t)k(u)ue^t\, du$,
so this gives
\begin{multline*}
r''(t)=
e^{\frac{1}{2}t}\frac{-\int k'(v)k'(ue^{-t})u^2e^{-2t}\, du
-\frac{1}{2}\int k'(ue^t)k(u)ue^t\, du}{\int k(u)^2\, du}  \\
+\frac{1}{2}e^{\frac{1}{2}t}\frac{\int k'(ue^t)k(u)ue^t\, du
+\frac{1}{2}\int k(ue^{t})k(u)\, du}{\int k(u)^2\, du}.
\end{multline*}
Thus,
\begin{equation*}
r''(0)
=\frac{-\int [k'(u)u]^2\, du
+\frac{1}{4}\int k(u)^2\, du}{\int k(u)^2\, du}.
\end{equation*}
Since, by the integration by parts argument above,
$\frac{1}{4}\int k(u)^2\, du=\frac{1}{2}\int k(u)^2\, du-\frac{1}{4}\int k(u)^2\, du=
-\int k(u)k'(u)u\, du-\frac{1}{4}\int k(u)^2\, du$,
this is equal to
\begin{equation*}
\frac{-\int [k'(u)u]^2\, du
-\int k(u)k'(u)u\, du-\frac{1}{4}\int k(u)^2\, du}{\int k(u)^2\, du}
=-\frac{\int \left[k'(u)u+\frac{1}{2} k(u)\right]^2\, du}{\int k(u)^2\, du}
\end{equation*}
which gives the required expansion with $C$ given by one half of the negative of the above display and $\alpha=2$.

%to do: check

%note: subsections below are used for verifying the conditions of the main theorem, rather than in proving it - may want to put in separate section

\subsection{Delta Method}\label{delta_method_sec}

We state some results that allow us to obtain influence function representations
with the necessary uniform rate for differentiable functions of estimators.
These results amount to applying the delta method to our setting and keeping
track of the uniform rates.

Let $\hat\beta(h)$ be an estimator of a parameter $\beta(h)\in\mathbb{R}^{d_\beta}$ with influence function representation
\begin{equation*}
\sqrt{nh}(\hat\beta(h)-\beta(h))
  =\frac{1}{\sqrt{nh}} \sum_{i=1}^n\psi_\beta(W_i,h)k(X_i/h) + R_{1,n}(h)
\end{equation*}
for some function $\psi_\beta$ and a kernel function $k$, where
$\psi_\beta(W_i,h)k(X_i/h)$ has mean zero and
$\sup_{\underline h_n\le h\le \overline h}|R_{1,n}(h)|=o_P(1/\sqrt{\log\log \underline h_n^{-1}})$.
Let $g$ be a function from $\mathbb{R}^{d_\beta}$ to $\mathbb{R}^{d_\theta}$ and consider the parameter $\theta(h)=g(\beta(h))$ and the estimator $\hat\theta(h)=g(\hat\beta(h))$.

Let $\hat V_\beta(h)$ be an estimate of $V_\beta(h)=\frac{1}{h}E\psi_\beta(W_i,h)\psi_\beta(W_i,h)'k(X_i/h)^2$,
the (pointwise in $h$) asymptotic variance of $\hat\beta(h)$.  A natural estimator of the asymptotic variance $V_\theta(h)$ of $\hat\theta$ is
\begin{align*}
\hat V_\theta(h)=D_g(\hat\beta(h))'\hat V_\beta(h)D_g(\hat\beta(h))'.
\end{align*}

\begin{lemma}\label{delta_method_lemma}
Suppose that $\beta(h)$ is bounded uniformly over $h\le \overline h_n$ where $\overline h_n=\mathcal{O}(1)$ and
\begin{itemize}
\item[(i)] For large enough $n$, $g$ is differentiable on an open set containing the range of $\beta(h)$ over $h\le \overline h_n$, with Lipschitz continuous derivative $D_g$.

\item[(ii)] $\psi_\beta$ and $k$ are bounded, $k$ has finite support, and the class of functions $(w,x)\mapsto \psi_\beta(w,h)k(x/h)$ has polynomial uniform covering number.

\item[(iii)] $|X_i|$ has a bounded density on $[0,\overline h_n]$ for large enough $n$.
\end{itemize}
Then, if $n\underline h_n/(\log\log n)^3\to\infty$,
\begin{align*}
\sup_{\underline h_n\le h\le \overline h_n}\left|\sqrt{nh}(\hat\theta(h)-\theta(h))
-\frac{1}{\sqrt{nh}} \sum_{i=1}^{n}D_g(\beta(h))\psi_\beta(W_i,h)k(X_i/h)\right|
=o_P\left(1\Big/\sqrt{\log\log \underline h_n^{-1}}\right).
\end{align*}
If, in addition, $\sup_{\underline h_n\le h\le\overline h_n}\|\hat V_\beta(h)-V_\beta(h)\|\stackrel{p}{\to} 0$, then,
for some constant $K$ and some $R_{n,2}(h)$ with
$\sup_{\underline h_n\le h\le\overline h_n} \frac{\sqrt{nh}}{\sqrt{\log\log h^{-1}}}|R_{n,2}(h)|=\mathcal{O}_P(1)$,
\begin{align*}
\left\|\hat V_\theta(h)-V_\theta(h)\right\|
\le K \left\|\hat V_\beta(h)-V_\beta(h)\right\| + R_{n,2}(h)
\end{align*}
for all $\underline h_n\le h\le \overline h_n$ with probability approaching one.

\end{lemma}
\begin{proof}
By a first order Taylor expansion, we have, for some $\beta^*(h)$ with $\|\beta^*(h)-\beta(h)\|\le\|\hat\beta(h)-\beta(h)\|$,
\begin{align*}
&\sqrt{nh}(\hat\theta(h)-\theta(h))
=\sqrt{nh}(g(\hat\beta(h))-g(\beta(h)))
=\sqrt{nh}D_g(\beta^*(h))(\hat\beta(h)-\beta(h))  \\
&=D_g(\beta^*(h))\frac{1}{\sqrt{nh}} \sum_{i=1}^n\psi_\beta(W_i,h)k(X_i/h)
  + D_g(\beta^*(h))R_{1,n}(h)  \\
&=\frac{1}{\sqrt{nh}} \sum_{i=1}^{n}D_g(\beta(h))\psi_\beta(W_i,h)k(X_i/h)
  +[D_g(\beta^*(h))-D_g(\beta(h))]\frac{1}{\sqrt{nh}} \sum_{i=1}^n\psi_\beta(W_i,h)k(X_i/h)  \\
&  + D_g(\beta^*(h))R_{1,n}(h)
\end{align*}
Applying Lemma~\ref{lil_rate_lemma},
%$\sup_{\underline h_n\le h\le\overline h_n}\frac{\sqrt{nh}}{\sqrt{\log\log h^{-1}}}
%  \left|\hat\beta(h)-\beta(h)\right|=\mathcal{O}_P(1)$
%and
%$\sup_{\underline h_n\le h\le\overline h_n}\frac{1}{\sqrt{\log\log h^{-1}}}
%  \left|\frac{1}{\sqrt{nh}}\sum_{i=1}^n\psi_\beta(W_i,h)k(X_i/h)\right|=\mathcal{O}_P(1)$,
$\hat\beta(h)-\beta(h)$ is $\mathcal{O}_P(\sqrt{\log\log h^{-1}}/\sqrt{nh})$ uniformly over $\underline h_n\le h\le \overline h_n$ and
$\frac{1}{\sqrt{nh}} \sum_{i=1}^n\psi_\beta(W_i,h)k(X_i/h)$ is $\mathcal{O}_P(\sqrt{\log\log h^{-1}})$ uniformly over $\underline h_n\le h\le \overline h_n$.
so that, by the Lipschitz condition on $D_g$, the second term is $\mathcal{O}_P(\log\log h^{-1}/\sqrt{nh})$ uniformly over $\underline h_n\le h\le\overline h_n$, which is $o_P(1/\sqrt{\log\log \underline h_n^{-1}})$ uniformly over $\underline h_n\le h\le \overline h_n$
since $\sqrt{n\underline h_n}/(\log\log \underline h_n)^{3/2}\to \infty$.  The last term is $o_P(1/\sqrt{\log\log h_n^{-1}})$ uniformly over $\underline h_n\le h\le \overline h_n$ by the conditions on $R_{1,n}(h)$, the uniform consistency of $\hat\beta(h)$ and the Lipschitz condition on $D_g$.

For the second claim, note that
\begin{align*}
\hat V_\theta-V_\theta
&=D_g(\hat\beta(h))\hat V_\beta(h)D_g(\hat\beta(h))'
  -D_g(\beta(h))V_\beta(h)D_g(\beta(h))'  \\
&=[D_g(\hat\beta(h))-D_g(\beta(h))]\hat V_\beta(h)D_g(\hat\beta(h))'
+D_g(\beta(h))[\hat V_\beta(h)-V_\beta(h)]D_g(\hat\beta(h))'  \\
&\quad +D_g(\beta(h))V_\beta(h)[D_g(\hat\beta(h))-D_g(\beta(h))]'.
\end{align*}
The first and last terms converge at a $\sqrt{\log\log h^{-1}}/\sqrt{nh}$ rate uniformly over $\underline h_n\le h\le \overline h_n$ by Lemma~\ref{lil_rate_lemma} and the Lipschitz continuity on $D_g$.  The second term is bounded by a constant times $\|\hat V_\beta(h)-\hat V_\beta(h)\|$ uniformly over $\underline h_n\le h\le \overline h$ with probability approaching one by the uniform consistency of $\hat\beta(h)$ and the Lipschitz continuity of $D_g$.
\end{proof}

\subsection{Sufficient Conditions Based on Non-normalized Influence Function}

In some cases, it will be easier to verify the conditions for an influence function approximation to
$\sqrt{nh}(\hat\theta(h)-\theta(h))$ rather than the normalized version
$\sqrt{nh}(\hat\theta(h)-\theta(h))/\hat\sigma(h)$.  The following lemma is useful in these cases.

\begin{lemma}\label{nonnorm_equiv_lemma}
Suppose that the following conditions hold for some $\tilde\psi(W_i,h)$.

\begin{enumerate}
\item $E\tilde\psi(W_i,h)k(X_i/h)=0$ and $k$ is bounded and symmetric with
  finite support $[-A,A]$.
\item $|X_i|$ has a density $f_{|X|}$ with $f_{|X|}(0)>0$,
  $\tilde\psi(W_i,h)k(X_i/h)$ is bounded uniformly over $h\le \underline h_n$
  and, for some deterministic function $\ell(h)$ with $\ell(h)\log\log h^{-1}\to
  0$ as $h\to 0$, the following expressions are bounded by $\ell(t)$:
  $|f_{|X|}(t)-f_{|X|}(0)|$,
  $|E\left[\tilde\psi(W_i,0)||X_i|=t\right]-E\left[\tilde\psi(W_i,0)||X_i|=0\right]|$,
  $|var\left[\tilde\psi(W_i,0)||X_i|=t\right]-var\left[\tilde\psi(W_i,0)||X_i|=0\right]|$
  and $|(\tilde\psi(W_i,t)-\tilde\psi(W_i,0))k(X_i/h)|$.
\end{enumerate}

Let $\sigma^2(h)=\frac{1}{h}var(\tilde\psi(W_i,h)k(X_i/h))$ for $h>0$ and let
$\sigma^2(0)=var\left[\tilde\psi(W_i,0)||X_i|=0\right]f_{|X|}(0)\cdot \int_{u=0}^\infty
k(u)^2\, du$. Let $\psi(W_i,h)=\tilde\psi(W_i,h)/\sigma(h)$ so that
$\frac{1}{h}var[\psi(W_i,h)k(X_i/h)]=1$. Suppose that
$var\left[\tilde\psi(W_i,0)||X_i|=0\right]>0$. Then the above assumptions hold
with $\tilde\psi$ replaced by $\psi$ for $h$ small enough and with $\ell(t)$
possibly redefined.
\end{lemma}
\begin{proof}
  First, note that the only condition we need to verify is the one that involves
  $|[\psi(W_i,h)-\psi(W_i,0)]k(X_i/h)|$, since the remaining
  conditions are only changed by multiplication by a constant when $\tilde\psi$
  is replaced by $\psi$. Note that
\begin{multline*}
  \sigma^2(h)-\frac{1}{h}var(\tilde\psi(W_i,0)k(X_i/h))
  =\frac{1}{h}var(\tilde\psi(W_i,h)k(X_i/h))-\frac{1}{h}var(\tilde\psi(W_i,0)k(X_i/h))  =\\
  \frac{1}{h}var\{[\tilde\psi(W_i,h)-\tilde\psi(W_i,0)]k(X_i/h)\}
  +2\frac{1}{h}cov\{[\tilde\psi(W_i,h)-\tilde\psi(W_i,0)]k(X_i/h),\tilde\psi(W_i,0)k(X_i/h)\}.
\end{multline*}
Since $|(\tilde\psi(W_i,h)-\tilde\psi(W_i,0))k(X_i/h)|\le \ell(h)I(|X_i|\le
Ah)$, $\tilde\psi(W_i,h)k(X_i/h)$ and $\tilde\psi(W_i,0)k(X_i/h)$ are bounded,
the last two terms are bounded by a constant times
$\ell(h)\frac{1}{h}EI(|X_i|\le Ah)$, which is bounded by a constant times
$\ell(h)$ by the assumption on the density of $|X_i|$.

Thus, let us consider
\begin{multline*}
  \frac{1}{h}var(\tilde\psi(W_i,0)k(X_i/h))  \\
  =\frac{1}{h}\int_{x=0}^\infty
  var\left[\tilde\psi(W_i,0)||X_i|=x\right]k(x/h)^2 f_{|X|}(x)\, dx +\frac{1}{h}
  var\left\{E\left[\tilde\psi(W_i,0)||X_i|\right]k(X_i/h)\right\}.
\end{multline*}
Arguing as in the proof of Lemma~\ref{a0_mu0_lemma} (using the fact that
$E\tilde\psi(W_i,h)k(X_i/h)=0$ and taking limits), it can be seen that
$E\left[\tilde\psi(W_i,0)||X_i|=0\right]=0$ under these conditions. Thus, the
last term is bounded by $\ell(Ah)^2\frac{1}{h}Ek(X_i/h)^2$. The first term is
equal to $var(\tilde\psi(W_i,0)\mid |X_i|=0)f_{|X|}(0)\cdot \int_{u=0}^\infty
k(u)^2\, du$ plus a term that is bounded by a constant times $\ell(Ah)$.

It follows that, letting $\sigma^2(0)=var\left[\tilde\psi(W_i,0)||X_i|=0\right]f_{|X|}(0)\int_{u=0}^\infty k(u)^2\, du$ as defined above, we have,
for some constant $K$,
$|\sigma^2(h)-\sigma^2(0)|\le K\ell(Ah)$.  Thus,
\begin{multline*}
\left|[\psi(W_i,h)-\psi(W_i,0)]k(X_i/h)\right|  \\
\le \frac{1}{\sigma(0)}\left|[\tilde\psi(W_i,h)-\tilde\psi(W_i,0)]k(X_i/h)\right|
+\left|\tilde\psi(W_i,h)k(X_i/h)\right|\cdot \left|\frac{1}{\sigma(h)}-\frac{1}{\sigma(0)}
\right|.
\end{multline*}
The first term is bounded by a constant times $\ell(h)$ by assumption.  The last term is bounded by a constant times $|\sigma^2(h)-\sigma^2(0)|$, which is bounded by a constant times $\ell(Ah)$ as shown above.
\end{proof}

\section{Local polynomial estimators: regression discontinuity/estimation at the boundary}\label{loc_poly_bound_sec}

This section gives primitive conditions for smooth functions of estimates based
on local polynomial estimates at the boundary, or at a discontinuity in the
regression function. The results are used in Section~\ref{applications_sec_supp}
below to verify the conditions of Theorem~\ref{highlevel_asym_dist_thm} for the
applications in Section~\ref{applications_sec} in the main text. Throughout this
section, we consider a setup with $\{(X_i,Y_i')'\}_{i=1}^n$ i.i.d.\ with $X_i$ a
real valued random variable and $Y_i$ taking values in $\mathbb{R}^{d_Y}$. We
consider smooth functions of the left and right hand limits of the regression
function at a point, which we normalize to be zero.

Let
$(\hat\beta_{u,j,1}(h),\hat\beta_{u,j,2}(h)/h,\ldots,\hat\beta_{u,j,r+1}(h)/h^{r})$
be the coefficients of an $r$th order local polynomial estimate of
$E[Y_{i,j}|X_i=0_{+}]$ based on the subsample with $X_{i}\geq 0$ with a kernel
function $k^{*}$. Similarly, let
$(\hat\beta_{\ell,j,1}(h),\hat\beta_{\ell,j,2}(h)/h,\ldots,\hat\beta_{\ell,j,r+1}(h)/h^{r})$
be the coefficients of an $r$th order local polynomial estimate of
$E[Y_{i,j}|X_i=0_{-}]$ based on the subsample with $X_{i}< 0$, where the
polynomial is taken in $|X_i|$ rather than $X_i$ (this amounts to multiplying
even elements of $\beta_{\ell,j}$ by $-1$). The scaling by powers of $h$ is used
to handle the different rates of convergence of the different coefficients. Let
$p(x)=(1,x,x^2,\ldots,x^{r})'$, and define
$\hat\beta_{u,j}=(\hat\beta_{u,j,1}(h),\hat\beta_{u,j,2}(h),\ldots,\hat\beta_{u,j,r+1}(h))$
and
$\hat\beta_{\ell,j}=(\hat\beta_{\ell,j,1}(h),\hat\beta_{\ell,j,2}(h),\ldots,\hat\beta_{\ell,j,r+1}(h))$.
Let $p(x)=(1,x,x^2,\ldots,x^{r})'$. Then $\hat\beta_{u,j}$ minimizes
\begin{align*}
  \sum_{i=1}^n (Y_{i,j}-p(|X_{i}/h|)'\beta_{u,j})^2 I(X_i\ge 0) k^*(X_i/h)
\end{align*}
and
$\hat\beta_{\ell,j}$ minimizes
\begin{align*}
  \sum_{i=1}^n (Y_{i,j}-p(\abs{X_{i}/h})'\beta_{u,j})^2 I(X_i < 0) k^*(X_i/h).
\end{align*}
Define
\begin{align*}
  \Gamma_u(h)&=\textstyle\frac{1}{h}Ep(|X_i/h|)p(|X_i/h|)'k^*(X_i/h)I(X_i\ge 0),\\
  \Gamma_\ell(h)&=\textstyle\frac{1}{h}Ep(|X_i/h|)p(|X_i/h|)'k^*(X_i/h)I(X_i< 0),\\
  \hat\Gamma_u(h)&=\textstyle\frac{1}{nh}\sum_{i=1}^{n} p(|X_i/h|)p(|X_i/h|)'k^*(X_i/h)I(X_i\ge 0)&\text{ and}\\
  \hat\Gamma_\ell(h)&=\textstyle\frac{1}{nh}\sum_{i=1}^{n}p(|X_i/h|)p(|X_i/h|)'k^*(X_i/h)I(X_i<
  0).
\end{align*}
Let $\mu_{k^*,\ell}=\int_{0}^\infty u^\ell k^*(u)\, du$, and
let $M$ be the matrix with $i,j$th element given by $\mu_{k^*,i+j-2}$.

Let $\hat\alpha_u(h)=(\hat\beta_{1,1,u}(h),\ldots\hat\beta_{1,d_Y,u}(h))'$ and
$\hat\alpha_\ell(h)=(\hat\beta_{1,1,\ell}(h),\ldots\hat\beta_{1,d_Y,\ell}(h))'$,
and similarly for $\alpha_u(h)$ and $\alpha_{\ell}(h)$ (i.e. $\alpha_u$ and
$\alpha_\ell$ contain the constant terms in the local polynomial regressions for
each $j$). Let $\hat\alpha(h)=(\hat\alpha_u(h)',\hat\alpha_\ell(h)')$ and
$\alpha(h)=(\alpha_u(h)',\alpha_\ell(h)')$. We are interested in
$\theta(h)=g(\alpha(h))$ for a differentiable function $g$ from
$\mathbb{R}^{2d_Y}$ to $\mathbb{R}$, and an estimator
$\hat\theta(h)=\hat g(\alpha(h))$. We consider standard errors defined by the
delta method applied to the robust covariance matrix formula obtained by
treating the local linear regressions as a system of $2d_Y$ weighted least
squares regressions. Let $\nu_u(h)=e_1'\Gamma_u(h)^{-1}$ and let
$\nu_\ell(h)=e_1'\Gamma_\ell(h)^{-1}$. Let
$\hat\nu_u(h)=e_1'\hat\Gamma_u(h)^{-1}$ and let
$\nu_\ell(h)=e_1'\hat\Gamma_\ell(h)^{-1}$. Let $\psi_\alpha(X_i,Y_i,h)$ be the
$(2d_Y)\times 1$ random vector with $j$th element given by
\begin{equation*}
\psi_{\alpha,j}(X_i,Y_i,h)=
\begin{cases}
  \nu_u(h)p(|X_i/h|)
  [Y_{i,j}-p(|X_i/h|)'\beta_{u,j}(h)]I(X_i\ge 0) & \text{if $j=1,\dotsc, d_{Y}$,} \\
  \nu_\ell(h)p(|X_i/h|) [Y_{i,j-d_Y}-p(|X_i/h|)'\beta_{\ell,j-d_Y}(h)]I(X_i< 0)
  & \text{if $j=d_{Y}+1,\dotsc,2d_{Y}$.}
\end{cases}
\end{equation*}
Let $\hat\psi_{\alpha}(X_i,Y_i,h)$ be defined analogously,
\begin{equation*}
\hat\psi_{\alpha,j}(X_i,Y_i,h)=
\begin{cases}
  \hat\nu_u(h)p(|X_i/h|)
  [Y_{i,j}-p(|X_i/h|)'\hat\beta_{u,j}(h)]I(X_i\ge 0) & \text{if $j=1,\dotsc, d_Y$,} \\
  \hat\nu_\ell(h)p(|X_i/h|)
  [Y_{i,j-d_Y}-p(|X_i/h|)'\hat\beta_{\ell,j-d_Y}(h)]I(X_i< 0) & \text{if
    $j=d_Y+1,\dotsc, 2d_Y$.}
\end{cases}
\end{equation*}
Let
\begin{align*}
V_{\alpha}(h)=\frac{1}{h}E \psi_\alpha(X_i,Y_i,h)\psi_\alpha(X_i,Y_i,h)'k^*(X_i/h)^2
\end{align*}
and let
\begin{align*}
\hat V_{\alpha}(h)=\frac{1}{h}E_n \hat\psi_\alpha(X_i,Y_i,h)\hat\psi_\alpha(X_i,Y_i,h)'k^*(X_i/h)^2.
\end{align*}
Let $\hat\sigma(h)=D_g(\hat\alpha(h))\hat V_\alpha(h)D_g(\hat\alpha(h))'$, and
$\sigma(h)=D_g(\alpha(h)) V_\alpha(h)D_g(\alpha(h))'$, where $D_g$ is the
derivative of $g$.

We make the following assumption throughout this section.
In the following assumption, $\ell(t)$ is an arbitrary nondecreasing function satisfying $\lim_{t\downarrow 0}\ell(t)\log\log t^{-1}=0$.

\begin{assumption}\label{loc_poly_assump}
\begin{itemize}
\item[(i)] $X_i$ has a density $f_X(x)$ with
$|f_X(x)-f_{X,-}|\le \ell(x)$ for $x<0$ and $|f_X(x)-f_{X,+}|\le \ell(x)$ for some $f_{X,+}>0$ and $f_{X,-}>0$.

\item[(ii)] $Y_i$ is bounded and, for some matrices $\Sigma_-$ and $\Sigma_+$ and vectors $\tilde\mu_-$ and $\tilde\mu_+$, $\tilde \Sigma(x)=var(Y_i|X_i=x)$ and $\tilde\mu(x)=E(Y_i|X_i=x)$ satisfy
$\|\tilde \Sigma(x)-\Sigma_+\|\le \ell(x)$ and
$\|\tilde \mu(x)-\tilde \mu_+\|\le \ell(x)$
for $x>0$ and
$\|\tilde \Sigma(x)-\Sigma_-\|\le \ell(x)$ and
$\|\tilde \mu(x)-\tilde \mu_-\|\le \ell(x)$
for $x<0$.
%note: (ii) may be restritive in LATE application if propensity score goes to 0 or 1 - consider relaxing

\item[(iii)] $k^*$ is symmetric with finite support $[-A,A]$, is bounded with a bounded, uniformly continuous first derivative on $(0,A)$, and satisfies $\int k(u)\, du\ne 0$, and the matrix $M$ is invertible.

\item[(iv)] $D_g$ is bounded and is Lipschitz continuous on
an open set containing
the range of $\alpha(h)$ over
$\overline h_n$ for $n$ large enough.
%$h\le \overline h$ for small enough $\overline h$.

\item[(v)] $D_{g,u}(\alpha(0))\tilde\Sigma_+ D_{g,u}(\alpha(0))>0$ or
  $D_{g,\ell}(\alpha(0))\tilde\Sigma_- D_{g,u}(\ell)>0$.

\item[(vi)] $\overline h_n=\mathcal{O}(1)$ and $n\underline h_n/(\log\log n)^3\to\infty$.
\end{itemize}
\end{assumption}

\begin{theorem}\label{loc_poly_bound_thm}
  Under Assumption~\ref{loc_poly_assump}, Assumptions~\ref{inf_func_assump} and
  Assumption~\ref{dgp_kern_assump} hold with $k(u)=e_1'M^{-1}p(|u|)k^*(u)$ and
  $\psi$ defined below so long as $n\underline h_n/(\log\log \underline
  h_n^{-1})^3\to \infty$ and $\overline h_n$ is small enough for large $n$.
\end{theorem}
\begin{comment}

Let $(\hat\beta_{1}(h),\hat\beta_{2}(h)/h,\ldots,\hat\beta_{r+1}(h)/h^{r})$ be the coefficients of an $r$th order local polynomial estimate of $E[Y_{i}|X_i=x]_+$, and let
$\hat\beta(h)=(\hat\beta_{1}(h),\hat\beta_{2}(h),\ldots,\hat\beta_{r}(h))$, so that $\hat\beta(h)$ minimizes
\begin{align*}
\sum_{i=1}^n (Y_{i,j}-\beta_{1}-\beta_{2}h X_i-\beta_{3}(h X_i)^2
  -\ldots -\beta_{r+1}(h X_i)^{r})^2I(X_i\ge 0)k^*(W_i/h)
\end{align*}
(the scaling is used to handle the different rates of convergence of the different coefficients).  Let $\beta(h)$ minimize the population version of the above display:
\begin{align*}
E (Y_{i,j}-\beta_{1}-\beta_{2}h X_i-\beta_{3}(h X_i)^2
  -\ldots -\beta_{r+1}(h X_i)^{r})^2I(X_i\ge 0)k^*(W_i/h).
\end{align*}
Let $p(x)=(1,x,\ldots,x^{r})'$. Let
$\Gamma(h)=\frac{1}{h}Ep(X_i/h)p(X_i/h)'k^*(X_i/h)I(X_i\ge 0)$, and let
$\hat\Gamma(h)=\frac{1}{nh}\sum_{i=1}^n p(X_i/h)p(X_i/h)'k^*(X_i/h)I(X_i\ge 0)$.
With this notation,
$\hat\beta(h)=\hat\Gamma(h)^{-1}E_n\frac{1}{h}E_n p(X_i/h)Y_i k^*(X_i/h)I(X_i\ge
0)$
\end{comment}

Throughout, we assume that $\overline h_n$ is small enough so that
$\|\Gamma_u(h)^{-1}\|$ and $\|\Gamma_\ell(h)^{-1}\|$ are bounded uniformly over
$h\le\overline h_n$ for large enough $n$ (this will hold for small enough
$\overline h_n$ by Lemma~\ref{loc_poly_gamma_lemma} below).

\begin{lemma}\label{lpoly_inf_func_lemma}
  Suppose that Assumption~\ref{loc_poly_assump} holds. Then
\begin{align*}
  \sup_{\underline h_n\le h\le\overline h_n}\frac{\sqrt{nh}}{\sqrt{\log\log
      h^{-1}}} \left\| \hat\Gamma_u(h)-\Gamma_u(h)
  \right\|&=\mathcal{O}_P(1),\\
  \sup_{\underline h_n\le h\le\overline h_n}\frac{\sqrt{nh}}{\sqrt{\log\log
      h^{-1}}} \left\| \hat\Gamma_u(h)^{-1}-\Gamma_u(h)^{-1}
  \right\|&=\mathcal{O}_P(1),\\
{\sup_{\underline h_n\le h\le\overline h_n} \frac{nh}{\log\log h^{-1}} \bigg\|
    \hat\beta_{u,j}(h)-\beta_{u,j}(h)\hspace{18em}}& \\
  -\frac{1}{h} E_n\Gamma_u(h)^{-1}p(X_i/h)k^*(X_i/h)
    [Y_i-p(X_i/h)'\beta(h)]I(X_i\ge 0) \bigg\|&=\mathcal{O}_P(1),
\end{align*}
and
\begin{equation*}
\sup_{\underline h_n\le h\le\overline h_n}\frac{\sqrt{nh}}{\sqrt{\log\log h^{-1}}}
\left\|
\hat\beta_{u,j}(h)-\beta_{u,j}(h)
\right\|=\mathcal{O}_P(1)
\end{equation*}
for each $j$.  The same holds with $I(X_i\ge 0)$ replaced by $I(X_i<0)$, $\Gamma_u$ replaced by $\Gamma_\ell$, $\hat\Gamma_u$ replaced by $\hat\Gamma_\ell$, etc.
\end{lemma}
\begin{proof}
The first display follows from Lemma~\ref{lil_rate_lemma}.
For the second display, note that
$\hat\Gamma(h)^{-1}-\Gamma(h)^{-1}
=-\hat\Gamma(h)^{-1}(\hat\Gamma(h)-\Gamma(h))\Gamma(h)^{-1}$,
so $\|\hat\Gamma(h)^{-1}-\Gamma(h)^{-1}\|
\le \|\hat\Gamma(h)^{-1}\|\|\hat\Gamma(h)-\Gamma(h)\|\|\Gamma(h)^{-1}\|$.
$\|\Gamma(h)^{-1}\|$ is bounded by assumption and
$\|\hat\Gamma(h)^{-1}\|$ is $\mathcal{O}_P(1)$ uniformly over $\underline h_n\le h\le \overline h_n$ by this and the first display in the lemma.
For the third display, note that
\begin{equation*}
  \hat\beta_{u,j}(h)-\beta_{u,j}(h)
  =\hat\Gamma_u(h)^{-1}\frac{1}{h}E_{n}p(X_i/h)k^*(X_i/h)
  [Y_i-p(X_i/h)'\beta(h)]I(X_i\ge 0).
\end{equation*}
Thus, letting $\mathcal{B}=-\frac{1}{h}E_n\Gamma_u(h)^{-1}p(X_i/h)k^*(X_i/h)
     [Y_i-p(X_i/h)'\beta(h)]I(X_i\ge 0)$,
\begin{multline*}
\sup_{\underline h_n\le h\le\overline h_n}\frac{nh}{\log\log h^{-1}}
\left\|
\hat\beta_{u,j}(h)-\beta_{u,j}(h)
  \mathcal{B}
\right\|  \\
\le \sup_{\underline h_n\le h\le\overline h_n}\frac{\sqrt{nh}}{\sqrt{\log\log h^{-1}}}
\left\|\hat\Gamma_u(h)^{-1}-\Gamma_u(h)^{-1}\right\|  \\
\cdot \sup_{\underline h_n\le h\le\overline h_n}\frac{\sqrt{nh}}{\sqrt{\log\log h^{-1}}}
\left\|\frac{1}{h}E_{n}p(X_i/h)k^*(X_i/h)
     [Y_i-p(X_i/h)'\beta(h)]I(X_i\ge 0)\right\|.
\end{multline*}
The first term is $\mathcal{O}_P(1)$ by the second display in the lemma.  The second term is $\mathcal{O}_P(1)$ by Lemma~\ref{lil_rate_lemma}.
The last display in the lemma follows from the third display and Lemma~\ref{lil_rate_lemma}.
\end{proof}

Applying the above lemma, we obtain the following.
\begin{lemma}\label{alpha_inf_func_lemma}
Under Assumption~\ref{loc_poly_assump},
\begin{align*}
\sup_{\underline h_n\le h\le\overline h_n}\frac{nh}{\log\log h^{-1}}
\left\|\hat\alpha(h)-\alpha(h)
  -\frac{1}{nh}\sum_{i=1}^n \psi_\alpha(X_i,Y_i,h)k^*(X_i/h)\right\|
=\mathcal{O}_P(1)
\end{align*}
and
\begin{align*}
\sup_{\underline h_n\le h\le\overline h_n}
  \frac{\sqrt{nh}}{\sqrt{\log\log h^{-1}}}\left\|\hat V_\alpha(h)-V_\alpha(h)\right\|
  =\mathcal{O}_P(1).
\end{align*}
\end{lemma}
\begin{proof}
The first claim follows by Lemma~\ref{lpoly_inf_func_lemma}.
The second claim follows by
using the fact that $\hat V_\alpha(h)$ is a Lipschitz continuous function of the $\hat\beta$ and $\hat\nu$ terms and terms that can be handled with Lemma~\ref{lil_rate_lemma}.
%to do: fill in arguments (from section on linear estimators)
\end{proof}

\begin{lemma}\label{lpoly_theta_approx_lemma}
  Suppose that Assumption~\ref{loc_poly_assump} holds. Then
\begin{equation*}
\sup_{\underline h_n\le h\le\overline h_n}\sqrt{nh}
\left\|\hat\theta(h)-\theta(h)
  -\frac{1}{\sqrt{nh}}\sum_{i=1}^n D_g(\alpha(h))\psi_\alpha(X_i,Y_i,h)k^*(X_i/h)\right\|
=o_P\left(1/\sqrt{\log\log \underline h_n^{-1}}\right)
\end{equation*}
and
\begin{equation*}
\sup_{\underline h_n\le h\le\overline h_n}
  \frac{\sqrt{nh}}{\sqrt{\log\log h^{-1}}}\left\|\hat\sigma(h)-\sigma(h)\right\|
  =\mathcal{O}_P(1).
\end{equation*}
\end{lemma}
\begin{proof}
  By Lemma~\ref{alpha_inf_func_lemma},
\begin{multline*}
\sup_{\underline h_n\le h\le \overline h_n}
\left\|\sqrt{nh}\left(\hat\alpha(h)-\alpha(h)\right)
-\frac{1}{\sqrt{nh}}\sum_{i=1}^n\psi_\alpha(X_i,Y_i,h)k^*(X_i/h)\right\|  \\
=\mathcal{O}_P\left(\sup_{\underline h\le h\le\overline h_n} (\log\log h^{-1})/\sqrt{nh}\right)
=\mathcal{O}_P\left((\log\log \underline h_n^{-1})/\sqrt{n\underline h_n}\right)
=o_P\left(1/\sqrt{\log\log \underline h_n^{-1}}\right)
\end{multline*}
since $(\log\log \underline h_n^{-1})^{3/2}/\sqrt{n\underline h_n}\to 0$.  Thus, the result follows by Lemma~\ref{delta_method_lemma}.
%to do: elaborate
\end{proof}

Let $m_j(x,h)=p(x/h)'\beta_{u,j}(h)$ for $x\ge 0$ and
$m_j(x,h)=p(x/h)'\beta_{\ell,j-d_Y}(h)$ for $x< 0$.  Let $D_{g,u}(\alpha)$ be the row vector with the first $d_Y$ elements of $D_g(\alpha)$, and let $D_{g,\ell}(\alpha)$ be the row vector with the remaining $d_Y$ elements.  With this notation, we have
\begin{multline*}
D_g(\alpha(h))\psi_{\alpha}(X_i,Y_i,h)  \\
=\left\{I(X_i\ge 0)\nu_u(h)p(|X_i/h|)D_{g,u}(\alpha(h))+
I(X_i< 0)\nu_\ell(h)p(|X_i/h|)D_{g,\ell}(\alpha(h))\right\}
[Y_{i}-m(X_i,h)].
\end{multline*}
Let $\gamma_{u,j}(h)=\frac{1}{h}EY_{i,j}p(|X_i/h|)k^*(X_i/h)I(X_i\ge 0)$
and $\gamma_{u,\ell}(h)=\frac{1}{h}EY_{i,j}p(|X_i/h|)k^*(X_i/h)I(X_i< 0)$.
Let $\gamma_{u,j}(0)$ be the $(r+1)\times 1$ vector with $q$th element given by
$f_{X,+}\tilde \mu_{+,j} \mu_{k^*,q}$.
Let $\gamma_{\ell,j}(0)$ be the $(r+1)\times 1$ vector with $q$th element given by
$f_{X,-}\tilde \mu_{-,j} \mu_{k^*,q}$.
Let $\alpha(0)=(\tilde\mu_+',\tilde\mu_-')'$ (it will be shown below that $\lim_{h\to 0} \alpha(h)=\alpha(0)$).

We now verify the conditions of the main result with
$k(u)=e_1'M^{-1}p(|u|)k^*(u)$ and
\begin{equation*}
\psi(W_i,h)=
\frac{D_g(\alpha(h))\psi_\alpha(X_i,Y_i,h)}{e_1'M^{-1}p(|X_i/h|)\sigma(h)}
\end{equation*}
for $h>0$ and
\begin{equation*}
\psi(W_i,0)=\frac{1}{\sigma(0)}\left[D_{g,u}(\alpha(0))f_{X,+}^{-1}(Y_i-\mu_{+})I(X_i\ge 0)+D_{g,\ell}(\alpha(0))f_{X,-}^{-1}(Y_i-\mu_{-})I(X_i<0)\right]
\end{equation*}
where $\sigma^2(0)=\lim_{h\to 0}\sigma^2(h)$
(this choice of $\psi(W_i,0)$ will be justified by the calculations below).

\begin{lemma}\label{loc_poly_gamma_lemma}
Under Assumption~\ref{loc_poly_assump}, for some constant $K$,
\begin{align*}
\|\Gamma_{u}(h)-f_{X,+}M\|&\le K \ell(Ah),\\
\|\Gamma_{\ell}(h)-f_{X,-}M\|&\le K \ell(Ah),\\
\|\gamma_{u}(h)-\gamma_{u}(0)\|&\le K \ell(Ah),\\
\text{and}\qquad\|\gamma_{\ell}(h)-\gamma_{\ell}(0)\|&\le K \ell(Ah).
\end{align*}
\end{lemma}
\begin{proof}
We have
\begin{equation*}
  \begin{split}
    \gamma_{u,j}(h)&=\frac{1}{h}EY_{i,j}p(|X_i/h|)k^*(X_i/h)I(X_i\ge 0)
    =\frac{1}{h}\int_{x=0}^\infty \tilde\mu_j(x)p(x/h)k^*(x/h)f_X(x)\, dx  \\
    & =\int_{x=0}^\infty \tilde\mu_j(uh)p(u)k^*(u)f_X(uh)\, dx.
  \end{split}
\end{equation*}
Thus, by boundedness of $k^*$, the quantity
$\|\gamma_{u,j}(h)-\gamma_{u,j}(0)\|$ is bounded by a constant times $\sup_{0\le
  x\le Ah} |\tilde\mu_j(x)f_X(x)-\tilde\mu_{+,j}f_{X,+}|$, which is bounded by a
constant times $\ell(Ah)$ by assumption. Similarly,
\begin{align*}
\Gamma_{u,j,m}(h)&=\frac{1}{h}E (X_i/h)^{j+m-2}k^*(X_i/h)I(X_i\ge 0)
=\frac{1}{h}\int_{x=0}^\infty (x/h)^{j+m-2}k^*(x/h)f_X(x)\, dx  \\
&=\int_{x=0}^\infty u^{j+m-2}k^*(u)f_X(uh)\, du,
\end{align*}
so $|\Gamma_{u,j,m}(h)-f_{X,+}M_{j,m}|$ is bounded by a constant times
$\sup_{0\le x\le Ah} |f_X(x)-f_{X,+}|\le \ell(Ah)$.
The proof for $\Gamma_\ell$ and $\gamma_\ell$ is similar.
\end{proof}

Note that
$\beta_{u,j}(h)=\Gamma_u(h)^{-1}\gamma_{u,j}(h)
  \to \tilde\mu_{+,j} M^{-1}(1,\mu_{k^*,1},\ldots,\mu_{k^*,r})'=\tilde\mu_{+,j}(1,0,\ldots,0)'$ as $h\to 0$, where the last equality follows since
$M^{-1}(1,\mu_{k^*,1},\ldots,\mu_{k^*,r})'$ is the first column of $M^{-1}M=I_{r+1}$ (the second through $r$th elements of $\beta_{u,j}$ are given by the corresponding coefficients of the local polynomial scaled by powers of $h$, so this is a result of the fact that the coefficients of the local polynomial do not increase too quickly as $h\to 0$).  By these calculations and Lemma~\ref{loc_poly_gamma_lemma}, we obtain the following.

\begin{lemma}
Under Assumption~\ref{loc_poly_assump}, for some constant $K$ and $h$ small enough,
\begin{align*}
  \left|\beta_{u,j}(h)-\tilde\mu_{+,j}(1,0,\ldots,0)'\right| &\le K\ell(Ah),\\
  \text{and}\qquad  \left|\beta_{\ell,j}(h)-\tilde\mu_{-,j}(1,0,\ldots,0)'\right| &\le K\ell(Ah).
\end{align*}
\end{lemma}
\begin{proof}
The result is immediate from Lemma~\ref{loc_poly_gamma_lemma}, the fact that $\|\Gamma_u(h)^{-1}\|$ and $\|\Gamma_\ell(h)^{-1}\|$ are bounded uniformly over small enough $h$ (which follows from Lemma~\ref{loc_poly_gamma_lemma} and invertibility of $M$) and fact that the function that takes $\Gamma$ and $\gamma$ to $\Gamma^{-1}\gamma$ is Lipschitz over $\Gamma$ and $\gamma$ with $\Gamma^{-1}$ and $\gamma$ bounded.
\end{proof}

Note that, since $\alpha(h)$ is made up of the first component of each of the $\beta_{u,j}(h)$ and $\beta_{\ell,j}(h)$ vectors, the above lemma also implies that
$|\alpha(h)-\alpha(0)|\le K\ell(Ah)$ for $\alpha(0)$ defined above.  For convenience, let us also define $\beta_{u,j}(0)$ and $\beta_{\ell,j}(0)$ to be the limits of $\beta_{u,j}(h)$ and $\beta_{\ell,j}(h)$ derived above.

\begin{lemma}\label{nu_lemma}
Under Assumption~\ref{loc_poly_assump}, for some constant $K$ and $h$ small enough,
\begin{align*}
  \|\nu_u(h)-e_1'M^{-1}f_{X,+}^{-1}\| & \le K\ell(Ah)&
  \text{and}\qquad
  \|\nu_\ell(h)-e_1'M^{-1}f_{X,-}^{-1}\| &\le K\ell(Ah).
\end{align*}
\end{lemma}
\begin{proof}
  The result is follows immediately from Lemma~\ref{loc_poly_gamma_lemma} and
  the the fact that $\|\Gamma_u(h)^{-1}\|$ and $\|\Gamma_\ell(h)^{-1}\|$ are
  bounded over small enough $h$.
\end{proof}

\begin{lemma}\label{loc_poly_psicont_lemma}
Under Assumption~\ref{loc_poly_assump}, for some constant $K$ and $h$ small enough,
\begin{align*}
\left|\left[\sigma(h)\psi(W_i,h)-\sigma(0)\psi(W_i,0)\right]k(X_i/h)\right|
\le K\ell(Ah).
\end{align*}
\end{lemma}
\begin{proof}
We have
\begin{align*}
  &\left[\sigma(h)\psi(W_i,h)-\sigma(0)\psi(W_i,0)\right]k(X_i/h)
    =D_g(\alpha(h))\psi_\alpha(X_i,Y_i,h)k^*(X_i/h)  \\
  &-\left[D_{g,u}(\alpha(0))f_{X,+}^{-1}(Y_i-\mu_{+})I(X_i\ge
    0)+D_{g,\ell}(\alpha(0))f_{X,-}^{-1}(Y_i-\mu_{-})I(X_i<0)\right]\cdot\\
  &  e_1'M^{-1}p(|X_i/h|)k^*(X_i/h)  \\
  &=D_g(\alpha(h))\psi_\alpha(X_i,Y_i,h)k^*(X_i/h) -D_g(\alpha(0))
    \tilde{\psi}_\alpha(X_i,Y_i,h)k^*(X_i/h)
\end{align*}
where the first $d_Y$ columns of $\tilde \psi_{\alpha}(X_i,Y_i,h)$ are given by
$e_1'M^{-1}p(|X_i/h|)f_{X,+}^{-1}(Y_i-\mu_+)I(X_i\ge 0)$ and the remaining $d_Y$
columns are given by $e_1'M^{-1}p(|X_i/h|)f_{X,-}^{-1}(Y_i-\mu_-)I(X_i< 0)$.
Note that the above expression can be written as
\begin{multline*}
T(X_i/h,Y_i,\nu_u(h),\nu_\ell(h),\alpha(h),
  \{\beta_{u,j,m}(h)\}_{1\le j\le d_Y,1\le m\le r+1},\{\beta_{\ell,j,m}(h)\}_{1\le j\le d_Y,1\le m\le r+1})  \\
-T(X_i/h,Y_i,\nu_u(0),\nu_\ell(0),\alpha(0),
  \{\beta_{u,j,m}(0)\}_{1\le j\le d_Y,1\le m\le r+1},\{\beta_{\ell,j,m}(0)\}_{1\le j\le d_Y,1\le m\le r+1})
\end{multline*}
for a function $T$ that is Lipschitz in its remaining arguments uniformly over $X_i/h,Y_i$ on bounded sets.  Combining this with the previous lemmas gives the result.
\end{proof}

It follows from Lemmas~\ref{loc_poly_psicont_lemma}
and~\ref{nonnorm_equiv_lemma} that the conclusion of
Lemma~\ref{loc_poly_psicont_lemma} also holds with $\sigma(h)\psi(W_i,h)$
replaced by $\psi(W_i,h)$, so long as the remaining conditions of
Lemma~\ref{nonnorm_equiv_lemma} (those involving the conditional expectation and
variance of $\psi(W_i,0)$) hold. We have
\begin{align*}
&E[\psi(W_i,0)|X_i=x]  \\
&=\frac{1}{\sigma(0)}\left\{D_{g,u}(\alpha(0))f^{-1}_{X,+}[\tilde\mu(x)-\tilde\mu_+]I(x\ge 0)
+D_{g,\ell}(\alpha(0))f^{-1}_{X,-}[\tilde\mu(x)-\tilde\mu_-]I(x<0)\right\}
\end{align*}
and
\begin{align*}
var[\psi(W_i,0)|X_i=x]
&=\frac{1}{\sigma^2(0)}\left\{D_{g,u}(\alpha(0))\tilde\Sigma(x)D_{g,u}(\alpha(0))'
f^{-2}_{X,+}I(x\ge 0)\right.  \\
&\left.+D_{g,\ell}(\alpha(0))\tilde\Sigma(x)D_{g,\ell}(\alpha(0))'
f^{-2}_{X,-}I(x< 0)\right\}
\end{align*}
By the conditions on $\tilde\mu(x)$ and $\tilde\Sigma(x)$, it follows that these expressions are left and right continuous in $x$ at $0$ with modulus $\ell(x)$ satisfying the necessary conditions.  By this and the conditions on $f_X$, it follows that the same holds for $E[\psi(W_i,0)||X_i|=x]$ and $var[\psi(W_i,0)||X_i|=x]$.  In addition, the assumptions guarantee that $var[\psi(W_i,0)||X_i|=x]$ is bounded away from zero for small $x$ so that $\sigma(0)>0$.

Thus, for $\psi(W_i,h)$ defined above,
\begin{align*}
&\sup_{\underline h_n\le h\le \overline h_n} \left\|\frac{\sqrt{nh}(\hat\theta(h)-\theta(h))}{\hat\sigma(h)}
-\frac{1}{\sqrt{nh}}\sum_{i=1}^n\psi(W_i,h)k(X_i/h)
\right\|  \\
&\le \sup_{\underline h_n\le h\le \overline h_n} \left\|\frac{\sqrt{nh}(\hat\theta(h)-\theta(h))}{\sigma(h)}
-\frac{1}{\sqrt{nh}}\sum_{i=1}^n\psi(W_i,h)k(X_i/h)
\right\|  \\
&+\sup_{\underline h_n\le h\le \overline h_n} \left\|
\sqrt{nh}(\hat\theta(h)-\theta(h))
\right\|\cdot
\left\|
\frac{1}{\sigma(h)}-\frac{1}{\hat\sigma(h)}
\right\|.
\end{align*}
By Lemma~\ref{lpoly_theta_approx_lemma}, the first term is of the order
$\mathcal{O}_P(1/\sqrt{\log\log \underline h_n^{-1}})$, and the last term is of
the order
$\mathcal{O}_P(\sqrt{\log\log \underline h_n^{-1}}\cdot {\sqrt{\log\log
    \underline h_n^{-1}}}/{\sqrt{n\underline h_n}})$. Thus, for
$(\log\log \underline h_n^{-1})^3/n\underline h_n\to 0$, both terms will be
$o_P(1/\sqrt{\log\log \underline h_n^{-1}})$ as required. This completes the
proof of Theorem~\ref{loc_poly_bound_thm}.

\subsection{Equivalent Kernels for Local Linear Regression}

Thus section gives the equivalent kernels for local polynomial regression at the
boundary and in the interior, and outlines how our results can be extended to
cover local polynomial regression at local-to-boundary points. Let
\begin{equation*}
  k(u;t)= e_{1}'M(t)^{-1}p( u)
  k^{*}(u),
\end{equation*}
where
  \begin{equation}\label{eq:M-matrix}
    M(t)=\int_{u=0}^\infty p(u-t)p(u-t)'k^{*}(u-t)\, du=
    \int_{u=-t}^{\infty}p(u)p(u)' k^{*}(u)\,du.
  \end{equation}
  Then the equivalent kernel for local polynomial regression at the boundary is
  given by $k(u;0)$. For $r=1$, we have
\begin{align*}
 e_1' M(0)^{-1}p(u)=e_{1}'\left(\begin{array}{cc}  \mu_{k^*,0} & \mu_{k^*,1}  \\
     \mu_{k^*,1} & \mu_{k^*,2}
\end{array}\right)^{-1}
\left(\begin{array}{c} 1  \\ |u|
\end{array}\right)
=\frac{\mu_{k^{*},2}-\mu_{k^{*},1}|u|}{\mu_{k^*,0}\mu_{k^*,2}-\mu_{k^*,1}^{2}}.
\end{align*}
For $r=2$, we have
\begin{equation*}
  e_{1}M^{-1}p(u)=\frac{1}{D}
  \left(\left(\mu_{k^*,4}\mu_{k^*,2}-\mu_{k^*,3}^{2}\right)
    +\left(\mu_{k^*,1}\mu_{k^*,4}-\mu_{k^*,2}\mu_{k^*,3}\right)\abs{u}
    +\left(   \mu_{k^*,2}^{2}-\mu_{k^*,1}\mu_{k^*,3}\right)u^{2}
  \right),
\end{equation*}
where $D=\det(M)=\mu_{k^*,0}(\mu_{k^*,2}\mu_{k^*,4}-\mu_{k^*,3}^{2})-\mu_{k^*,1}(\mu_{k^*,1}\mu_{k^*,4}-\mu_{k^*,2}\mu_{k^*,3})+
\mu_{k^*,2}(\mu_{k^*,1}\mu_{k^*,3}-\mu_{k^*,2}^{2})$. The moments
$\mu_{k^{*},j}$ for the uniform, triangular, and Epanechnikov kernel are given by\\[1em]
  \begin{tabular}{@{}lrrrrr@{}}
    Name & $\mu_{0}$ & $\mu_{1}$ & $\mu_{2}$ & $\mu_{3}$ & $\mu_{4}$\\
    \midrule
    Uniform & $\frac{1}{2}$ & $\frac{1}{4}$ & $\frac{1}{6}$ & $\frac{1}{8}$
    & $\frac{1}{10}$ \\
    Triangular & $\frac{1}{2}$ &$\frac{1}{6}$& $\frac{1}{12}$ & $\frac{1}{20}$ &
    $\frac{1}{30}$\\
    Epanechnikov & $\frac{1}{2}$ & $\frac{3}{16}$ & $\frac{1}{10}$ &
    $\frac{1}{16}$ & $\frac{3}{70}$ \\[1em]
  \end{tabular}

  Plugging these moments into the definitions of equivalent kernels in the two
  displays above then yields the definitions of equivalent kernels for local
  linear and local quadratic regressions. These definitions are summarized in
  Table~\ref{tab:kernel-definitions}.

  Theorem~\ref{loc_poly_bound_thm} can be extended to apply to local polynomial
  estimation in the interior, provided that the definition of the equivalent
  kernel is appropriately altered to $k(u;\infty)$ (so that the integral on the
  right-hand side of Equation~\eqref{eq:M-matrix} is over the whole real line
  rather than the interval $(0,\infty)$ as in the boundary case). Our package
  \texttt{BWSnooping} can be used to calculate the appropriate critical values
  in this case. Note that for $r=1$, the equivalent kernel and the original
  kernel coincide, so that one can use Table~\ref{tab:cvs-condensed} to look up
  the appropriate critical value.

  Finally, let us outline how our results can be extended to cover estimating a
  conditional mean at a point that is local to the boundary of the support of
  the distribution of the conditioning variable. Here we can use the
  local-to-boundary formulation of the problem as in Section 3.2.5 of
  \citet{fan_local_1996}. In particular, consider local polynomial estimation of
  $E(Y_i\mid X_i=x_0)$ where $x_0=c\underline h_n$ and the lower support point
  of the density of $X_i$ is zero. Letting $\hat{\theta}(h)$ denote the $r$th
  order local polynomial estimator based on a kernel $k^{*}$, it can be
  shown that under regularity conditions, $\sup_{h\in [\underline h_n,\overline
    h_n]} \sqrt{nh}|\hat\theta(h)-\theta(h)|/\hat\sigma(h)$ can be approximated
  by $\sup_{t\in[1,\overline{h}_{n}/\underline{h}_{n}]}|\mathbb{H}(t)|$, where
  $\mathbb{H}(t)$ is a Gaussian process with covariance function
  $cov(\mathbb{H}(s),\mathbb{H}(t))=\rho(s,t;c)$, with
  \begin{align*}
  \rho(s,t;c) =\frac{\int_{u=-c}^\infty k(u/s; c/s) k(u/t; c/t)\,
    du}{\sqrt{\int_{u=-c}^\infty k(u/s; c/s)^2\, du}
    \sqrt{\int_{u=-c}^\infty k(u/t; c/t)^2\, du}}.
\end{align*}
Note that the critical value depends only on $h/\underline h_n$ and $c$ (along
with the kernel and order of the local polynomial). Similar result obtains for
one-sided $t$-statistics.

\section{Proofs for Theorems in
  Appendix~\ref{sec:techn-deta-appl}}\label{applications_sec_supp}

\subsection{Regression Discontinuity/LATEs for Largest Sets of Compliers}

This section proves Theorems~\ref{reg_disc_thm} and~\ref{late_thm}.
First, note that the regression discontinuity and LATE applications can both be written as functions of local polynomial estimators in the above setup, with $d_Y=2$ and
$Y_{i}$ playing the role of $Y_{i,1}$ and $D_i$ playing the role of $Y_{i,2}$.  For the LATE application, we define
$X_i=-(Z_i-\underline z)I(|Z_i-\underline z|\le |Z_i-\overline z|)
  +(\overline z-Z_i)I(|Z_i-\underline z|> |Z_i-\overline z|)$.  Both of these applications fit into the setup of Section~\ref{loc_poly_bound_sec} with, letting
$\alpha(h)=(\alpha_u(h)',\alpha_\ell(h)')=(\alpha_{u,Y}(h),\alpha_{u,D}(h),\alpha_{\ell,Y}(h),\alpha_{\ell,D}(h))'$
(where we use the suggestive subscripts ``$Y$'' and ``$D$'' rather than $1$ and $2$),
$g(\alpha)=\frac{\alpha_{u,Y}-\alpha_{\ell,Y}}{\alpha_{u,D}-\alpha_{\ell,D}}$.  Then, letting $\Delta_D=\alpha_{u,D}-\alpha_{\ell,D}$, we have
\begin{align*}
D_g(\alpha)
  =\left[\begin{array}{cccc}\frac{1}{\Delta_D} & \frac{-g(\alpha)}{\Delta_D} &
\frac{-1}{\Delta_D} & \frac{g(\alpha)}{\Delta_D}\end{array}\right].
\end{align*}
This is Lipschitz continuous and bounded over bounded sets with $\alpha_{u,D}-\alpha_{\ell,D}$ bounded away from zero.

For the last condition (non-degeneracy of the conditional variance), note that
$D_{g,u}(\alpha(0))\tilde\Sigma_+ \cdot D_{g,u}(\alpha(0))
=\frac{1}{\Delta_D(0)^2}var[Y_i-g(\alpha(0))D_i|X_i=0_+]$, which will be nonzero
so long as $corr(D_i,Y_i\mid X_i=0_+)<1$ and $var(Y_i\mid X_i=0_+)>0$. A sufficient
condition for this is that $var(Y_i\mid D_i=d,X_i=0_+)>0$ is nonzero for $d=0$ or
$d=1$, and this (or the corresponding statement with $+$ replaced by $-$) holds
under the conditions of the theorem.

\subsection{Trimmed Average Treatment Effects under Unconfoundedness}\label{unconf_proof_sec}

This section proves Theorem~\ref{ate_unconf_thm}. We first give an intuitive
derivation of the critical value, which explains why it differs in this setting,
and provide the technical details at the end.

To derive the form of the correction in this case, note that, under the conditions of the theorem,
$\frac{\sqrt{n}(\hat\theta(h)-\theta(h))}{\hat\sigma(h)}$
will converge to a Gaussian process $\mathbb{G}(h)$ with covariance
\begin{equation*}
cov(\mathbb{G}(h),\mathbb{G}(h'))
=\frac{cov\left\{[\tilde Y_i-\theta(h)]I(X_i\in\mathcal{X}_h),
[\tilde Y_i-\theta(h')]I(X_i\in\mathcal{X}_{h'})\right\}}
{\sqrt{var\left\{[\tilde Y_i-\theta(h)]I(X_i\in\mathcal{X}_h)\right\}
var\left\{[\tilde Y_i-\theta(h')]I(X_i\in\mathcal{X}_{h'})\right\}}}.
\end{equation*}
Let $v(h)=var\{[\tilde Y_i-\theta(h)]I(X_i\in\mathcal{X}_h)\}$ as defined in the statement of the theorem.
Note that, for $h\ge h'$,
\begin{align*}
&cov\left\{[\tilde Y_i-\theta(h)]I(X_i\in\mathcal{X}_h),
[\tilde Y_i-\theta(h')]I(X_i\in\mathcal{X}_{h'})\right\}
=E\left\{[\tilde Y_i-\theta(h)][\tilde Y_i-\theta(h')]I(X_i\in\mathcal{X}_{h})\right\}  \\
&=E\left\{[\tilde Y_i-\theta(h)]^2I(X_i\in\mathcal{X}_{h})\right\}
+[\theta(h)-\theta(h')]E\left\{[\tilde Y_i-\theta(h)]I(X_i\in\mathcal{X}_{h})\right\}
=v(h)
\end{align*}
where the last step follows since
$E\left\{[\tilde Y_i-\theta(h)]I(X_i\in\mathcal{X}_{h})\right\}=0$. Note also
that $v(h)$ is weakly decreasing in $h$, which can be seen by noting that
$v(h)=\inf_a E\left\{\left[\tilde Y_i-a\right]^2I(X_i\in\mathcal{X}_h)\right\}$,
since $\theta(h)$ is the conditional expectation of $\tilde Y_i$ given
$X_i\in\mathcal{X}_h$. Thus,
\begin{align*}
cov(\mathbb{G}(h),\mathbb{G}(h'))=\frac{v(h\vee h')}{\sqrt{v(h)v(h')}}
=\frac{v(h)\wedge v(h')}{\sqrt{v(h)v(h')}},
\end{align*}
so
$\mathbb{G}(h)\stackrel{d}{=}\frac{\mathbb{B}(v(h))}{\sqrt{v(h)}}$
where $\mathbb{B}$ is a Brownian motion.
Thus, the distribution of
$\sup_{\underline h\le h\le \overline h} \frac{\sqrt{n}(\hat\theta(h)-\theta(h))}{\hat\sigma(h)}$
can be approximated by the distribution of
$\sup_{v(\overline h)\le t\le v(\underline h)} \frac{\mathbb{B}(t)}{\sqrt{t}}
\stackrel{d}{=} \sup_{1\le t\le v(\underline h)/v(\overline h)} \frac{\mathbb{B}(t)}{\sqrt{t}}$.
Note that
$v(h)=\sigma(h)^2P(X_i\in\mathcal{X}_h)^2$, so that
\begin{align*}
\frac{v(\underline h)}{v(\overline h)}
=\frac{\sigma(\underline h)^2P(X_i\in\mathcal{X}_{\underline h})^2}
  {\sigma(\overline h)^2P(X_i\in\mathcal{X}_{\overline h})^2}.
\end{align*}
Thus, $\hat t$ is a consistent estimator for
$\frac{v(\underline h)}{v(\overline h)}$ under the conditions of the theorem.

The formal result then obtains by noting that, by Theorem 19.5 in
\citet{van_der_vaart_asymptotic_1998},
$\frac{\sqrt{n}(\hat\theta(h)-\theta(h))}{\hat\sigma(h)} \stackrel{d}{\to}
\mathbb{G}(h)$, taken as processes over $h\in[\underline h,\overline h]$ with
the supremum norm. By the calculations above,
\begin{align*}
  \sup_{h\in[\underline h,\overline h]} \left|\mathbb{G}(h)\right|
  \stackrel{d}{=}\sup_{h\in[\underline h,\overline h]}
  \left|\frac{\mathbb{B}(v(h))}{\sqrt{v(h)}}\right|,
\end{align*}
where $\mathbb{B}$ is a Brownian motion.
The result then follows since
$\{t|v(h)=t\text{ some }h\in[\underline h,\overline h]\}\subseteq [v(\overline h),v(\underline h)]$,
and the two sets are equal if $v(h)$ is continuous.

\section{Additional details for critical values}\label{sec:critical-values}
We first give a one-sided version of Theorem~\ref{highlevel_asym_dist_thm} and
Corollary~\ref{theta0_corollary}. To state the result, recall the definitions of
$b(t,k)$ and the Gaussian process $\mathbb{H}(h)$ as defined in the statement of
Theorem~\ref{highlevel_asym_dist_thm}.

\begin{theorem}\label{th:highlevel_thm_onesided}
  Let $\cvo{1-\alpha}(t,k)$ be the $1-\alpha$ quantile of
  $\sup_{1\le h\le t}\mathbb{H}(h)$. Suppose that $\underline h_n\to 0$,
  $\overline h_n=\mathcal{O}_P(1)$, and
  $n\underline h_n/[(\log \log n)(\log\log\log n)]^2\to\infty$. Then, under
  Assumptions~\ref{inf_func_assump} and~\ref{dgp_kern_assump},
  \begin{equation*}
    P\left(
      \theta(\theta)\in \hor{\hat\theta(h)- \hat\sigma(h)\cdot
        \cvo{1-\alpha}{\overline h_n/\underline h_n}{k}/\sqrt{nh},\infty}
      \text{ all }h\in [\underline h_n\le h\le \overline h_n]
    \right)
    \stackrel{n\to\infty}{\to} 1-\alpha
  \end{equation*}
  The above display also holds with
  $\cvo{1-\alpha}(\overline h_n/\underline h_n,k)$ replaced by
  \begin{equation*}
    \frac{-\log\left(-\log (1-\alpha)\right)+b(\overline h_n/\underline h_n,
    k)}{\sqrt{2\log\log(\overline h_n/\underline h_n)}} +\sqrt{2\log\log
    (\overline h_n/\underline h_n)},
  \end{equation*}
  provided $\overline h_n/\underline h_n\to\infty$.
  If
  $\sup_{h\in [\underline h_n,\overline
    h_n]}\frac{\sqrt{nh}(\theta(h)-\theta(0))}{\hat\sigma(h)} \le
  o_P\big((\log\log (\overline h_n/\underline h_n))^{-1/2}\big)$, then
  \begin{equation*}
  \liminf_{n\to\infty }P\left(
    \theta(0)\in \hor{\hat\theta(h)- \hat\sigma(h)\cdot
      \cvo{1-\alpha}{\overline h_n/\underline h_n}{k}/\sqrt{nh},\infty}
    \text{ all }h\in [\underline h_n\le h\le \overline h_n]
  \right)
  \ge 1-\alpha.
\end{equation*}
\end{theorem}
Unlike in the two-sided case, the bias does not have to be negligible so long as
it can be signed: if $\theta(h)-\theta(0)$ is known to be weakly negative
(positive), then bias can only improve the coverage of a lower (upper) one-sided
CI (see Section~\ref{sec:adaptation_example}). The proof of
Theorem~\ref{th:highlevel_thm_onesided} is analogous to the proof of
Theorem~\ref{highlevel_asym_dist_thm} given in
Appendix~\ref{sec:proof-main-result}.

Tables~\ref{tab:cvs-ts} and~\ref{tab:cvs-os} give two- and one-sided critical
values $\cvo{1-\alpha}(\overline h_n/\underline h_n,k)$ and
$\cvt{1-\alpha}(\overline h_n/\underline h_n,k)$ for several kernel functions
$k$, $\alpha$ and a selected of values of $\overline h_n/\underline h_n$ for
90\%, 95\%, and 99\% confidence intervals. The critical values can also be
obtained using our R package \texttt{BWSnooping}, which can be downloaded from
\url{https://github.com/kolesarm/BWSnooping}.  The package also
includes critical values for local quadratic regression, and computes critical
values for other significance levels and other ratios of maximum to minimum
bandwidth $\overline{h}/\underline{h}$.

For comparison, Figure~\ref{fig:extreme-value} plots critical values based on
the extreme value approximation (given in the second part of
Theorem~\ref{highlevel_asym_dist_thm}) along with those based directly on the
Gaussian process.

\section{Monte Carlo evidence}\label{mc_sec}
We conduct a small Monte Carlo study of inference in a sharp regression
discontinuity design to further illustrate our method and to examine how well it
works in practice. In each replication, we generated a random sample
$\left\{X_{i},\varepsilon_{i}\right\}_{i=1}^{n}$, with size $n=500$,
$X_{i}=2Z_{i}-1$, where $Z_{i}$ has Beta distribution with parameters $2$ and
$4$, and $\varepsilon_{i}\sim\mathcal{N}(0,0.1295^{2})$. The regression
discontinuity point is normalized to zero. The outcome $Y_{i}$ is given by
$Y_{i}=g_{j}(X_{i})+\varepsilon_{i}$, where the regression function $g_{j}$ depends
on the design. We consider two regression functions. The first one corresponds
to a polynomial fit to the \citet{lee08} data,
\begin{equation*}
  g_{1}(x)=\begin{cases}
0.48+1.27x+7.18x^{2}+20.21x^{3}+21.54x^{4}+7.33x^{5} & \text{if $x<0$,} \\
0.52+0.84x-3.00x^{2}+7.99x^{3}-9.01x^{4}+3.56x^{5} & \text{otherwise.}
\end{cases}
\end{equation*}
This design corresponds exactly to the data generating process in
\citet[IK]{imbens_optimal_2012} and \citet[CCT]{cct14}. The second regression
function corresponds to another design in IK, and is given by
\begin{equation*}
  g_{2}(x)=0.42+0.1I(x\geq 0)+0.84x+7.99x^{3}-9.01x^{4}+3.56x^{5}.
\end{equation*}
Figure~\ref{fig:rd-muX} plots the conditional mean functions $g_{1}$ and $g_{2}$
that generate the data in Designs 1 and 2. The results for designs in which the
error term $\varepsilon_{i}$ is heteroscedastic are very similar, and reported
in an earlier version of the paper \citep{ArKo15snooping}.

In each design, we consider estimates based on local linear regression using the
uniform and the triangular kernel. We use the bandwidth selector proposed by IK
to select a baseline bandwidth, and then construct confidence bands for
estimators in bandwidth range around this baseline bandwidth. We also consider
the robust bias correction method of CCT discussed in
Section~\ref{reg_discont_sec} by running a local quadratic regression at the
same bandwidths. To define these estimators, let $p(x)=(1,x,\dotsc,x^{r})$
denote a polynomial expansion of order $r$. Given an i.i.d.~sample
$\left\{Y_{i},X_{i}\right\}_{i=1}^{n}$, the RD estimator is given by the
difference between the intercepts of polynomial linear regressions of order
$r$ with the same bandwidth on either side of the cutoff,
\begin{equation*}
  \hat{\theta}(h)=\hat{\alpha}_{u}(h)-\hat{\alpha}_{\ell}(h),
\end{equation*}
where $\hat{\alpha}_{u}(h)=e_{1}'\beta_{u}(h)$,
$\hat{\alpha}_{\ell}(h)=e_{1}'\beta_{\ell}(h)$,
\begin{align*}
  \hat{\beta}_{u}(h)&= \hat{\Gamma}_{u}(h)^{-1}\sum_{i=1}^{n}I({X_{i}\geq 0})
  k^{*}(X_{i}/h)p(\abs{X_{i}})Y_{i},\\
  \hat{\beta}_{\ell}(h)&= \hat{\Gamma}_{l}(h)^{-1}\sum_{i=1}^{n}I({X_{i}< 0})
  k^{*}(X_{i}/h)p(\abs{X_{i}})Y_{i},
\end{align*}
$k^{*}$ is a kernel, and
\begin{align*}
\hat{  \Gamma}_{u}(h)&=\sum_{i}I({X_{i}\geq 0})
  k^{*}(X_{i}/h)p(\abs{X_{i}})p(\abs{X_{i}})',\\
\hat{  \Gamma}_{\ell}(h)&=\sum_{i}I({X_{i}< 0})
  k^{*}(X_{i}/h)p(\abs{X_{i}})p(\abs{X_{i}})'.
\end{align*}
The corresponding function $\theta(h)$ is plotted in Figures~\ref{fig:rd-thetah}
and~\ref{fig:rd-thetah-quadratic} for the local linear and local quadratic
estimators.

To estimate the variance of the estimator, we use the Eicker-Huber-White (EHW)
robust variance estimator that treats the two linear linear regressions on
either side of the cutoff as a weighted linear regression. In
Theorem~\ref{reg_disc_thm} below, we show formally that using this estimator
leads to uniformly valid confidence intervals. We also consider a modification
of the EHW estimator that uses a nearest neighbor (NN) estimator to estimate
$var(Y_{i}\mid X_{i})$ in the middle part of the Eicker-Huber-White
``sandwich'', rather than using the regression residuals. This estimator was
introduced by \citet{AbIm06} and \citet{AbImZh14}, and it was studied by
\citet{cct14} in an RD context. The nearest neighbor (NN) and EHW variance
estimators have the form
\begin{equation*}
  \hat{\sigma}^{2}(h)=nh \left(\widehat{var}(\hat{\alpha}_{u}(h))
    +\widehat{var}(\hat{\alpha}_{\ell}(h))\right),
\end{equation*}
where
\begin{equation*}
  \widehat{var}(\hat{\alpha}_{u}(h))=e_{1}'
  \hat{\Gamma}_{u}(h)^{-1}\left(\sum_{i=1}^{n}
    I({X_{i}\geq 0})    \hat{\sigma}_{u}^{2}(X_{i})
    k^{*}(X_{i}/h)p(\abs{X_{i}})p(\abs{X_{i}})'
  \right) \hat{\Gamma}_{u}(h)^{-1} e_{1}
\end{equation*}
and similarly for $\widehat{var}(\hat{\alpha}_{u}(h))$, where
$\hat{\sigma}_{u}^{2}(X_{i})$ and $\hat{\sigma}_{\ell}^{2}(X_{i})$ are some
estimators of $var(Y_{i}\mid X_{i})$. The EHW estimator sets
$\hat{\sigma}_{u}^{2}(X_{i})=(Y_{i}-X_{i}'\hat{\beta}_{u})^{2}$, and the NN
estimators sets
    \begin{equation*}
      \hat{\sigma}_{u}^{2}(X_{i})=
      I(X_{i}\geq 0)\frac{J}{J+1}\left(Y_{i}-\sum_{j=1}^{J}Y_{\ell_{u,j}(i)}\right)^{2},
  \end{equation*}
  where $\ell_{u,j}(i)$ is the $j$th closest unit to $i$ among $\left\{k\neq
    i\colon X_{k}\geq 0\right\}$, and $J=3$.

Table~\ref{tab:rd-1-thetah-ik} reports empirical coverage of the confidence
bands for $\theta(h)$ for the two designs we consider. Our adjustment works well
overall, with the empirical coverage being close to 95\% for almost all
specifications, in contrast with the naive confidence bands (using the
unadjusted 1.96 critical value), which undercover. As plotted in
Figure~\ref{fig:coverage-twosided}, Theorem~\ref{highlevel_asym_dist_thm}
predicts that with $\bar{h}/\underline{h}=2$, the coverage should be 91.6\% for
the triangular kernel, and 83.9\% for the uniform kernel. When
$\bar{h}/\underline{h}=4$, the coverage of the naive confidence bands should
drop to 88.5\% and 76.8\%, respectively. The Monte Carlo results match these
predictions closely. There are a few specifications in which the adjusted
confidence bands based on EHW standard errors undercover. This happens when
small bandwidths are considered, and is due to the well-known downward bias of
EHW standard errors in small samples, so that the pointwise confidence intervals
fail to achieve nominal coverage in the first place. Since our method only
corrects for the multiple comparisons, it cannot solve this problem. Overall,
the adjusted confidence bands have coverage that is as good as the coverage of
the underlying pointwise confidence intervals.

Typically in regression discontinuity studies, the primary object of interest is
$\theta(0)$, the average treatment effect conditional on $X=0$. We therefore
also report empirical coverage of the confidence bands for $\theta(0)$ in
Table~\ref{tab:rd-1-theta0-ik}. Confidence bands around undersmoothed local
linear estimator, (that correspond to the bandwidth range
$[\hat{h}_{IK}/4,\hat{h}_{IK}/2]$) perform well, provided NN standard errors,
which perform better in small samples, are used. At larger values of the
bandwidth, $\hat{\theta}(h)$ is a biased estimator of $\theta(0)$. The pointwise
confidence intervals based on the local linear regression do not take this bias
into account, and they fail to achieve proper coverage. Consequently, although
our adjustment ensures that the coverage of the adjusted confidence band is
within the range of the pointwise confidence intervals, it still falls short of
95\% due to the pointwise confidence intervals performing poorly. On the other
hand, confidence bands around the bias-adjusted confidence intervals (that
correspond to local quadratic regression) perform well, especially when the
NN standard errors are used.

In conclusion, our adjustment performs well in terms of coverage of $\theta(h)$,
with empirical coverage close to nominal coverage, especially when combined with
NN standard errors. If our method is combined with undersmoothing (corresponding
to bandwidth ranges smaller than $\hat{h}_{IK}$), or bias-correction (such as
when the CCT method for constructing confidence intervals is used), so that the
underlying pointwise confidence intervals achieve good coverage of $\theta(0)$,
our method also achieves good coverage of $\theta(0)$.

\setstretch{1.2} % one half spacing
\bibliography{../library}

\clearpage

\begin{landscape}
\begin{table}[p]
  \centering
  \renewcommand{\arraystretch}{1.2}      % space between rows
  \begin{tabular}{@{}l@{\hspace{1em}}rrr@{\hspace{1em}}rrr@{\hspace{1em}}rrr@{\hspace{1em}}rrr@{\hspace{1em}}rrr@{\hspace{1em}}rrr@{}}
&
\multicolumn{9}{c@{\hspace{1em}}}{NW / Loc.~linear (interior)}&\multicolumn{9}{c}{Loc.~linear (boundary)}\\
\cmidrule(rl){2-10}\cmidrule(rl){11-19}
 & \multicolumn{3}{c}{Unif} & \multicolumn{3}{c}{Tri} & \multicolumn{3}{c}{Epa} &
\multicolumn{3}{c}{Unif} & \multicolumn{3}{c}{Tri} & \multicolumn{3}{c}{Epa}\\
$\overline{h}/\underline{h}$ & 0.1 & 0.05 & 0.01 & 0.1 & 0.05 & 0.01& 0.1 & 0.05 & 0.01& 0.1 & 0.05 & 0.01& 0.1 & 0.05 & 0.01& 0.1 & 0.05 & 0.01\\
\midrule
1.0 & 1.65 & 1.96 & 2.57 & 1.64 & 1.96 & 2.58 & 1.64 & 1.96 & 2.58 & 1.63 & 1.96 & 2.57 & 1.64 & 1.95 & 2.57 & 1.64 & 1.96 & 2.57\\
1.2 & 1.92 & 2.24 & 2.85 & 1.70 & 2.01 & 2.63 & 1.71 & 2.03 & 2.65 & 1.92 & 2.23 & 2.83 & 1.72 & 2.03 & 2.64 & 1.73 & 2.05 & 2.66\\
1.4 & 2.02 & 2.33 & 2.93 & 1.74 & 2.05 & 2.67 & 1.77 & 2.08 & 2.69 & 2.02 & 2.33 & 2.93 & 1.77 & 2.08 & 2.69 & 1.80 & 2.11 & 2.72\\
1.6 & 2.09 & 2.40 & 2.98 & 1.78 & 2.09 & 2.70 & 1.81 & 2.12 & 2.72 & 2.09 & 2.39 & 3.00 & 1.80 & 2.12 & 2.73 & 1.84 & 2.15 & 2.76\\
1.8 & 2.14 & 2.45 & 3.03 & 1.81 & 2.11 & 2.72 & 1.85 & 2.15 & 2.75 & 2.14 & 2.44 & 3.04 & 1.84 & 2.16 & 2.76 & 1.88 & 2.19 & 2.81\\
2   & 2.18 & 2.48 & 3.07 & 1.83 & 2.14 & 2.75 & 1.87 & 2.17 & 2.78 & 2.18 & 2.48 & 3.08 & 1.87 & 2.18 & 2.78 & 1.91 & 2.22 & 2.83\\
3   & 2.30 & 2.60 & 3.18 & 1.91 & 2.22 & 2.83 & 1.96 & 2.27 & 2.86 & 2.30 & 2.60 & 3.18 & 1.96 & 2.27 & 2.86 & 2.01 & 2.32 & 2.91\\
4   & 2.37 & 2.66 & 3.24 & 1.96 & 2.26 & 2.86 & 2.02 & 2.31 & 2.92 & 2.36 & 2.66 & 3.24 & 2.01 & 2.32 & 2.90 & 2.06 & 2.37 & 2.95\\
5   & 2.41 & 2.70 & 3.28 & 2.00 & 2.30 & 2.90 & 2.05 & 2.35 & 2.95 & 2.41 & 2.71 & 3.27 & 2.05 & 2.35 & 2.94 & 2.11 & 2.41 & 2.99\\
6   & 2.44 & 2.73 & 3.31 & 2.02 & 2.32 & 2.92 & 2.08 & 2.37 & 2.97 & 2.44 & 2.73 & 3.31 & 2.08 & 2.37 & 2.96 & 2.13 & 2.43 & 3.01\\
7   & 2.47 & 2.75 & 3.34 & 2.04 & 2.34 & 2.94 & 2.10 & 2.39 & 2.99 & 2.47 & 2.76 & 3.33 & 2.10 & 2.39 & 2.98 & 2.16 & 2.45 & 3.04\\
8   & 2.49 & 2.77 & 3.35 & 2.06 & 2.35 & 2.95 & 2.12 & 2.41 & 3.01 & 2.49 & 2.78 & 3.35 & 2.12 & 2.41 & 2.99 & 2.18 & 2.47 & 3.05\\
9   & 2.51 & 2.79 & 3.37 & 2.07 & 2.37 & 2.96 & 2.14 & 2.42 & 3.02 & 2.50 & 2.79 & 3.37 & 2.14 & 2.43 & 3.00 & 2.20 & 2.48 & 3.06\\
10  & 2.52 & 2.80 & 3.38 & 2.08 & 2.38 & 2.97 & 2.15 & 2.44 & 3.04 & 2.52 & 2.81 & 3.39 & 2.15 & 2.44 & 3.01 & 2.21 & 2.50 & 3.07\\
20  & 2.61 & 2.89 & 3.45 & 2.16 & 2.45 & 3.03 & 2.23 & 2.51 & 3.10 & 2.61 & 2.89 & 3.45 & 2.23 & 2.52 & 3.08 & 2.29 & 2.58 & 3.14\\
50  & 2.70 & 2.97 & 3.51 & 2.24 & 2.53 & 3.10 & 2.31 & 2.59 & 3.15 & 2.70 & 2.98 & 3.52 & 2.32 & 2.60 & 3.16 & 2.38 & 2.66 & 3.21\\
100 & 2.75 & 3.02 & 3.56 & 2.29 & 2.57 & 3.14 & 2.36 & 2.64 & 3.20 & 2.76 & 3.02 & 3.56 & 2.37 & 2.65 & 3.20 & 2.44 & 2.71 & 3.25\\
\end{tabular}
\caption{Critical values $\cvt{1-\alpha}(\overline{h}/\underline{h},k)$ for
  level $\alpha=0.1$, $0.05$, and $0.01$ for the Uniform (Unif,
  $k(u)=\frac{1}{2}I(\abs{u}\leq 1)$), Triangular (Tri,
  $(1-\abs{u})I(\abs{u}\leq 1)$) and Epanechnikov (Epa,
  $3/4(1-u^{2})I(\abs{u}\leq 1)$) kernels. ``NW / Loc.~linear (interior)''
  refers to Nadaraya-Watson (local constant) regression in the interior or at a
  boundary, as well as local linear regression in the interior. ``Loc.~linear
  (boundary)'' refers to local linear regression at a boundary (including
  regression discontinuity designs).}\label{tab:cvs-ts}
\end{table}
\end{landscape}

\begin{landscape}
\begin{table}[p]
  \centering
  \renewcommand{\arraystretch}{1.2}      % space between rows
  \begin{tabular}{@{}l@{\hspace{1em}}rrr@{\hspace{1em}}rrr@{\hspace{1em}}rrr@{\hspace{1em}}rrr@{\hspace{1em}}rrr@{\hspace{1em}}rrr@{}}
&
\multicolumn{9}{c@{\hspace{1em}}}{NW / Loc.~linear (interior)}&\multicolumn{9}{c}{Loc.~linear (boundary)}\\
\cmidrule(rl){2-10}\cmidrule(rl){11-19}
 & \multicolumn{3}{c}{Unif} & \multicolumn{3}{c}{Tri} & \multicolumn{3}{c}{Epa} &
\multicolumn{3}{c}{Unif} & \multicolumn{3}{c}{Tri} & \multicolumn{3}{c}{Epa}\\
$\overline{h}/\underline{h}$ & 0.1 & 0.05 & 0.01 & 0.1 & 0.05 & 0.01& 0.1 & 0.05 & 0.01& 0.1 & 0.05 & 0.01& 0.1 & 0.05 & 0.01& 0.1 & 0.05 & 0.01\\
\midrule

1.0 & 1.29 & 1.66 & 2.33 & 1.29 & 1.66 & 2.34 & 1.29 & 1.66 & 2.34 & 1.28 & 1.64 & 2.33 & 1.29 & 1.65 & 2.33 & 1.28 & 1.65 & 2.33\\
1.2 & 1.57 & 1.94 & 2.64 & 1.35 & 1.72 & 2.39 & 1.36 & 1.73 & 2.41 & 1.57 & 1.93 & 2.59 & 1.36 & 1.72 & 2.40 & 1.38 & 1.74 & 2.42\\
1.4 & 1.67 & 2.04 & 2.73 & 1.39 & 1.76 & 2.44 & 1.42 & 1.79 & 2.45 & 1.67 & 2.03 & 2.69 & 1.41 & 1.78 & 2.46 & 1.44 & 1.80 & 2.47\\
1.6 & 1.75 & 2.11 & 2.79 & 1.42 & 1.80 & 2.47 & 1.46 & 1.83 & 2.50 & 1.74 & 2.10 & 2.76 & 1.46 & 1.81 & 2.49 & 1.49 & 1.85 & 2.52\\
1.8 & 1.80 & 2.15 & 2.83 & 1.46 & 1.83 & 2.49 & 1.49 & 1.86 & 2.53 & 1.80 & 2.15 & 2.81 & 1.49 & 1.84 & 2.52 & 1.53 & 1.88 & 2.55\\
2 & 1.84 & 2.19 & 2.85 & 1.48 & 1.85 & 2.52 & 1.52 & 1.89 & 2.55 & 1.84 & 2.19 & 2.84 & 1.52 & 1.87 & 2.55 & 1.56 & 1.91 & 2.58\\
3 & 1.97 & 2.31 & 2.97 & 1.56 & 1.93 & 2.58 & 1.62 & 1.98 & 2.64 & 1.96 & 2.30 & 2.95 & 1.62 & 1.96 & 2.62 & 1.67 & 2.01 & 2.67\\
4 & 2.04 & 2.38 & 3.02 & 1.61 & 1.97 & 2.63 & 1.68 & 2.03 & 2.68 & 2.03 & 2.36 & 3.02 & 1.67 & 2.01 & 2.67 & 1.73 & 2.06 & 2.72\\
5 & 2.09 & 2.42 & 3.05 & 1.65 & 2.01 & 2.66 & 1.71 & 2.07 & 2.71 & 2.08 & 2.41 & 3.04 & 1.71 & 2.05 & 2.70 & 1.77 & 2.11 & 2.76\\
6 & 2.12 & 2.45 & 3.08 & 1.68 & 2.03 & 2.68 & 1.74 & 2.09 & 2.74 & 2.12 & 2.44 & 3.07 & 1.74 & 2.08 & 2.72 & 1.80 & 2.13 & 2.77\\
7 & 2.15 & 2.48 & 3.10 & 1.71 & 2.05 & 2.70 & 1.77 & 2.11 & 2.76 & 2.14 & 2.47 & 3.09 & 1.76 & 2.10 & 2.74 & 1.83 & 2.16 & 2.80\\
8 & 2.17 & 2.50 & 3.12 & 1.72 & 2.07 & 2.72 & 1.79 & 2.13 & 2.77 & 2.17 & 2.49 & 3.11 & 1.79 & 2.12 & 2.75 & 1.85 & 2.18 & 2.81\\
9  & 2.19 & 2.52 & 3.14 & 1.74 & 2.09 & 2.73 & 1.80 & 2.14 & 2.79 & 2.18 & 2.51 & 3.12 & 1.80 & 2.14 & 2.77 & 1.87 & 2.20 & 2.82\\
10 & 2.21 & 2.53 & 3.16 & 1.76 & 2.10 & 2.74 & 1.82 & 2.16 & 2.81 & 2.20 & 2.52 & 3.13 & 1.82 & 2.15 & 2.79 & 1.88 & 2.21 & 2.84\\
20 & 2.29 & 2.62 & 3.23 & 1.83 & 2.17 & 2.80 & 1.91 & 2.24 & 2.87 & 2.29 & 2.61 & 3.22 & 1.90 & 2.23 & 2.86 & 1.97 & 2.29 & 2.91\\
50 & 2.40 & 2.71 & 3.31 & 1.92 & 2.25 & 2.87 & 2.00 & 2.32 & 2.94 & 2.39 & 2.70 & 3.30 & 1.99 & 2.31 & 2.92 & 2.06 & 2.38 & 2.99\\
100& 2.46 & 2.77 & 3.36 & 1.98 & 2.30 & 2.92 & 2.06 & 2.37 & 2.99 & 2.45 & 2.76 & 3.35 & 2.05 & 2.37 & 2.96 & 2.12 & 2.44 & 3.03\\
\end{tabular}
\caption{One-sided critical values $\cvo{1-\alpha}(\overline{h}/\underline{h},k)$ for
  level $\alpha=0.1$, $0.05$, and $0.01$ for the Uniform (Unif,
  $k(u)=\frac{1}{2}I(\abs{u}\leq 1)$), Triangular (Tri,
  $(1-\abs{u})I(\abs{u}\leq 1)$) and Epanechnikov (Epa,
  $3/4(1-u^{2})I(\abs{u}\leq 1)$) kernels. ``NW / Loc.~linear (interior)''
  refers to Nadaraya-Watson (local constant) regression in the interior or at a
  boundary, as well as local linear regression in the interior. ``Loc.~linear
  (boundary)'' refers to local linear regression at a boundary (including
  regression discontinuity designs).}\label{tab:cvs-os}
\end{table}
\end{landscape}

\begin{table}[p]
  \centering
  \renewcommand*{\arraystretch}{1.5}
  \begin{tabular}{llll}
   Name & $k^{*}(u)$ & {Order} & $k(u)$\\
   \midrule
   \multirow{3}{*}{Uniform} &    \multirow{3}{*}{$\frac{1}{2}I(\abs{u}\leq 1)$} &
   0 & $\frac{1}{2}I(\abs{u}\leq 1)$\\
   && 1& $(4-6\abs{u})I(\abs{u}\leq 1)$\\
   &&2 &$(9-36\abs{u}+30u^2)I(\abs{u}\leq 1)$\\[1em]
   \multirow{3}{*}{Triangular} &    \multirow{3}{*}{$(1-\abs{u})_{+}$}
   & 0 & $(1-\abs{u})_{+}$\\
   && 1 & $6(1-2\abs{u})(1-\abs{u})_{+}$\\
   && 2 & $12(1-5\abs{u}+5u^{2})(1-\abs{u})_{+}$\\[1em]
   \multirow{3}{*}{Epanechnikov} &    \multirow{3}{*}{$\frac{3}{4}(1-u^{2})_{+}$}
   & 0 & $\frac{3}{4}(1-u^{2})_{+}$\\
   && 1 &$\frac{6}{19}(16-30\abs{u})(1-u^2)_{+}$\\
   && 2 & $\frac{1}{8}(85-400\abs{u}+385{u}^{2})(1-u^{2})_{+}$
  \end{tabular}
  \caption{Definitions of kernels and equivalent kernels for regression
    discontinuity / estimation at a boundary. Order refers to the
    order of the local polynomial.}\label{tab:kernel-definitions}
\end{table}

\begin{table}[p]
  \centering
  \renewcommand*{\arraystretch}{1.2}
  \begin{tabular}{@{}llllllll@{}}
    && \multicolumn{3}{@{}c}{Uniform kernel} &
    \multicolumn{3}{c@{}}{Triangular kernel}\\
    \cmidrule(rl){3-5}\cmidrule(rl){6-8}
    $(\underline{h},\bar{h})$ & $\hat{\sigma}(h)$ & Pointwise & Naive
    & Adj. & Pointwise & Naive& Adj. \\
    \midrule
    \multicolumn{3}{@{}l}{Design 1: Local Linear regression}\\
    \multirow{2}{*}{($\hat{h}_{IK}/4$, $\hat{h}_{IK}/2$)}
&  EHW     & (92.7, 94.4) & 83.7 & 93.9 & (91.8, 94.0) & 88.5 & 92.3\\
&       NN & (94.6, 95.8) & 87.3 & 95.3 & (94.2, 95.3) & 91.2 & 94.2\\[0.1ex]
    \multirow{2}{*}{($\hat{h}_{IK}/2$, $\hat{h}_{IK}$)}
&  EHW     & (94.2, 94.7) & 85.0 & 95.0 & (93.9, 94.5) & 90.5 & 94.0\\
&       NN & (95.3, 96.1) & 87.9 & 96.3 & (94.9, 95.9) & 92.3 & 95.3\\[0.1ex]
%     \multirow{2}{*}{($1/2\hat{h}_{IK}$, $3/2\hat{h}_{IK}$)}
% &   EHW  & (89.9, 94.6) & 76.6 & 92.7& (91.9, 94.5) & 86.9 & 92.7\\
% &    NN  & (90.0, 94.9) & 77.2 & 93.0& (92.0, 94.7) & 87.5 & 92.8\\[0.1ex]
    \multirow{2}{*}{($\hat{h}_{IK}/2$, $2\hat{h}_{IK}$)}
&  EHW     & (90.4, 94.7) & 74.8 & 93.4 & (92.1, 94.5) & 85.6 & 93.0\\
&       NN & (91.8, 96.1) & 77.4 & 94.4 & (93.4, 95.9) & 88.2 & 94.4\\[0.1ex]
\multicolumn{3}{@{}l}{Design 1: Local quadratic regression}\\
\multirow{2}{*}{($\hat{h}_{IK}/4$, $\hat{h}_{IK}/2$)}
&    EHW  & (89.5, 92.6) & 78.1 & 90.2& (88.5, 91.9) & 83.3 & 88.7\\
&      NN & (93.8, 94.8) & 85.2 & 94.3 & (93.2, 94.5) & 89.0 & 93.0\\[0.1ex]
\multirow{2}{*}{($\hat{h}_{IK}/2$, $\hat{h}_{IK}$)}
&    EHW  & (92.7, 94.3) & 82.7 & 93.5& (92.0, 94.0) & 88.0 & 92.4\\
&      NN & (94.8, 95.7) & 87.1 & 95.5 & (94.5, 95.4) & 91.3 & 94.6\\[0.1ex]
% \multirow{2}{*}{($\hat{h}_{IK}/2$, $3\hat{h}_{IK}/2$)}
% &   EHW  & (92.7, 95.1) & 79.1 & 93.6& (92.0, 94.8) & 85.9 & 92.5\\
% &    NN  & (93.5, 95.3) & 80.5 & 94.3& (93.3, 94.9) & 87.2 & 93.3\\
\multirow{2}{*}{($\hat{h}_{IK}/2$, $2\hat{h}_{IK}$)}
&    EHW  & (84.4, 95.1) & 68.6 & 90.1& (89.4, 94.8) & 80.1 & 89.7\\
&      NN & (87.1, 96.2) & 74.9 & 92.9 & (91.3, 96.0) & 84.5 & 92.5\\
    \multicolumn{3}{@{}l}{Design 2: Local Linear regression}\\
    \multirow{2}{*}{($\hat{h}_{IK}/4$, $\hat{h}_{IK}/2$)}
&  EHW     & (85.6, 91.4) & 73.7 & 86.5 & (83.0, 90.0) & 78.3 & 83.2\\
&       NN & (93.3, 94.3) & 84.6 & 93.9 & (92.7, 93.5) & 88.7 & 92.1\\[0.1ex]
    \multirow{2}{*}{($\hat{h}_{IK}/2$, $\hat{h}_{IK}$)}
&  EHW     & (91.3, 92.7) & 80.3 & 91.9 & (90.2, 92.1) & 86.0 & 90.3\\
&       NN & (94.0, 94.6) & 85.3 & 94.3 & (93.5, 94.2) & 90.2 & 93.2\\[0.1ex]
    \multirow{2}{*}{($\hat{h}_{IK}/2$, $2\hat{h}_{IK}$)}
&  EHW     & (88.5, 92.8) & 70.6 & 90.4 & (86.3, 92.1) & 78.8 & 87.7\\
&       NN & (85.0, 94.8) & 73.0 & 91.3 & (81.0, 94.2) & 76.0 & 85.3\\[0.1ex]
\multicolumn{3}{@{}l}{Design 2: Local quadratic regression}\\
\multirow{2}{*}{($\hat{h}_{IK}/4$, $\hat{h}_{IK}/2$)}
&   EHW  & (74.6, 86.9) & 57.6 & 72.9& (74.3, 85.4) & 65.5 & 72.0\\
&    NN  & (92.2, 93.1) & 80.7 & 91.6& (91.5, 92.4) & 85.1 & 89.6\\[0.1ex]
\multirow{2}{*}{($\hat{h}_{IK}/2$, $\hat{h}_{IK}$)}
&   EHW  & (87.5, 91.5) & 74.1 & 87.6& (85.8, 90.9) & 80.3 & 85.9\\
&    NN  & (92.2, 93.3) & 81.6 & 92.1& (91.6, 92.8) & 86.8 & 90.6\\[0.1ex]
\multirow{2}{*}{($\hat{h}_{IK}/2$, $2\hat{h}_{IK}$)}
&   EHW  & (87.5, 94.6) & 67.6 & 88.3& (85.8, 94.0) & 76.8 & 86.6\\
&    NN  & (92.2, 95.2) & 74.8 & 92.2& (91.6, 94.5) & 83.2 & 91.0\\
  \end{tabular}
  \caption{Monte Carlo study of regression discontinuity.
    Empirical coverage of $\theta(h)$ for nominal 95\% confidence bands
    around IK bandwidth. ``Pointwise'' refers to range of coverage of pointwise
    confidence intervals. ``Naive'' refers to the coverage of the naive
    confidence band that uses the unadjusted critical value equal to 1.96.
    ``Adj.'' refers to confidence bands using adjusted critical values based
    on Theorem~\ref{highlevel_asym_dist_thm}.  Variance estimators are described
    in the text. 10,000 Monte Carlo draws, 100 grid points
    for $h$.}\label{tab:rd-1-thetah-ik}
\end{table}

\begin{table}[p]
  \centering
  \renewcommand*{\arraystretch}{1.2}
  \begin{tabular}{@{}llllllll@{}}
    && \multicolumn{3}{@{}c}{Uniform kernel} &
    \multicolumn{3}{c@{}}{Triangular kernel}\\
    \cmidrule(rl){3-5}\cmidrule(rl){6-8}
    $(\underline{h},\bar{h})$ & $\hat{\sigma}(h)$ & Pointwise & Naive
    & Adj. & Pointwise & Naive& Adj. \\
    \midrule
    \multicolumn{3}{@{}l}{Design 1: Local Linear regression}\\
    \multirow{2}{*}{($\hat{h}_{IK}/4$, $\hat{h}_{IK}/2$)}
&  EHW     & (90.3, 92.7) & 79.9 & 92.0 & (90.0, 92.1) & 85.5 & 90.1\\
&       NN & (92.4, 94.7) & 84.5 & 94.4 & (92.4, 94.1) & 89.2 & 92.7\\[0.1ex]
    \multirow{2}{*}{($\hat{h}_{IK}/2$, $\hat{h}_{IK}$)}
&  EHW     & (73.7, 89.8) & 62.0 & 80.7 & (76.7, 89.5) & 74.0 & 80.9\\
&       NN & (77.1, 92.0) & 66.9 & 84.2 & (80.1, 91.9) & 78.0 & 84.1\\[0.1ex]
    \multirow{2}{*}{($\hat{h}_{IK}/2$, $2\hat{h}_{IK}$)}
&  EHW     & (73.2, 89.8) & 54.3 & 80.9 & (76.5, 89.5) & 69.4 & 81.1\\
&       NN & (76.7, 92.0) & 59.8 & 84.9 & (79.9, 91.9) & 74.2 & 84.9\\[0.1ex]
\multicolumn{3}{@{}l}{Design 1: Local quadratic regression}\\
\multirow{2}{*}{($\hat{h}_{IK}/4$, $\hat{h}_{IK}/2$)}
&    EHW  & (89.6, 92.7) & 78.3 & 90.0& (88.6, 92.2) & 83.5 & 88.6\\
&      NN & (93.8, 94.7) & 85.1 & 94.4 & (93.2, 94.3) & 88.8 & 93.0\\[0.1ex]
\multirow{2}{*}{($\hat{h}_{IK}/2$, $\hat{h}_{IK}$)}
&    EHW  & (90.3, 93.6) & 80.2 & 92.2& (89.3, 93.0) & 85.2 & 90.0\\
&      NN & (92.4, 95.3) & 84.8 & 94.7 & (91.5, 94.7) & 88.4 & 92.8\\[0.1ex]
\multirow{2}{*}{($\hat{h}_{IK}/2$, $2\hat{h}_{IK}$)}
&    EHW  & (74.8, 93.6) & 54.8 & 81.8& (78.6, 93.0) & 70.3 & 82.5\\
&      NN & (78.6, 95.3) & 62.2 & 86.3 & (82.2, 94.7) & 75.3 & 86.5\\
    \multicolumn{3}{@{}l}{Design 2: Local Linear regression}\\
    \multirow{2}{*}{($\hat{h}_{IK}/4$, $\hat{h}_{IK}/2$)}
&  EHW     & (85.6, 91.4) & 73.7 & 86.5 & (83.0, 90.1) & 78.3 & 83.2 \\
&       NN & (93.3, 94.3) & 84.6 & 93.9 & (92.7, 93.5) & 88.7 & 92.1 \\[0.1ex]
    \multirow{2}{*}{($\hat{h}_{IK}/2$, $\hat{h}_{IK}$)}
&  EHW     & (91.3, 92.6) & 80.2 & 91.8 & (90.2, 91.9) & 85.7 & 90.1 \\
&       NN & (94.0, 94.6) & 85.1 & 94.3 & (93.5, 94.1) & 89.9 & 93.1 \\[0.1ex]
    \multirow{2}{*}{($\hat{h}_{IK}/2$, $2\hat{h}_{IK}$)}
&  EHW     & (59.3, 92.6) & 47.8 & 75.3 & (53.9, 91.9) & 47.4 & 60.7 \\
&       NN & (63.1, 94.6) & 54.0 & 79.8 & (57.6, 94.1) & 52.9 & 65.5 \\[0.1ex]
\multicolumn{3}{@{}l}{Design 2: Local quadratic regression}\\
\multirow{2}{*}{($\hat{h}_{IK}/4$, $\hat{h}_{IK}/2$)}
&   EHW  & (74.6, 86.9) & 57.6 & 72.9& (74.3, 85.4) & 65.5 & 72.0\\
&    NN & (93.6, 94.8) & 84.4 & 93.6 & (92.8, 94.0) & 87.7 & 91.6\\[0.1ex]
\multirow{2}{*}{($\hat{h}_{IK}/2$, $\hat{h}_{IK}$)}
&   EHW  & (87.5, 91.5) & 74.2 & 87.6& (85.8, 90.9) & 80.3 & 85.9\\
&    NN & (93.5, 94.4) & 84.1 & 93.8 & (92.9, 93.8) & 88.4 & 92.4\\[0.1ex]
\multirow{2}{*}{($\hat{h}_{IK}/2$, $2\hat{h}_{IK}$)}
&   EHW  & (87.5, 94.3) & 67.2 & 88.0& (85.8, 93.6) & 75.9 & 86.1\\
&    NN & (93.5, 95.8) & 78.3 & 93.8 & (92.9, 95.1) & 84.6 & 92.4\\[0.1ex]
  \end{tabular}
  \caption{Monte Carlo study of regression discontinuity.
    Empirical coverage of $\theta(0)$ for nominal 95\% confidence bands
    around IK bandwidth. ``Pointwise'' refers to range of coverage of pointwise
    confidence intervals. ``Naive'' refers to the coverage of the naive
    confidence band that uses the unadjusted critical value equal to 1.96.
    ``Adj.'' refers to confidence bands using adjusted critical values based
    on Theorem~\ref{highlevel_asym_dist_thm}.  Variance estimators are described
    in the text. 10,000 Monte Carlo draws, 100 grid points
    for $h$.}\label{tab:rd-1-theta0-ik}
\end{table}

\begin{figure}[p]
  \centering
  \input{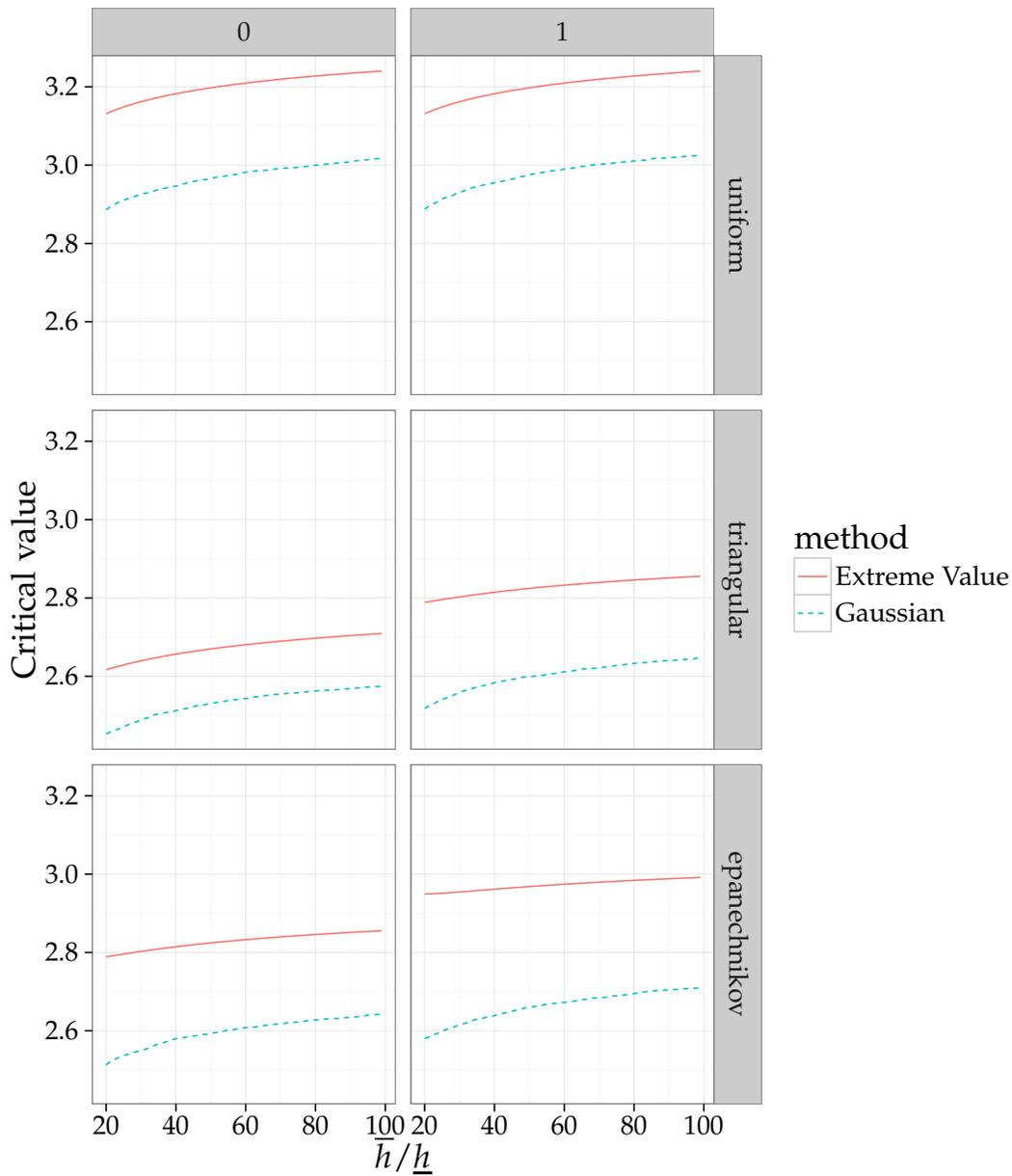}
  \caption{Comparison of critical values based on Gaussian approximation and
    extreme value approximation (i.e.\ asymptotic approximation as
    $\overline{h}/\underline{h}\to\infty$). Order ``0'' corresponds to
    Nadaraya-Watson interior or boundary regression, and to local linear
    regression in the interior, and order ``1'' to local linear regression at a
    boundary.}\label{fig:extreme-value}
\end{figure}

\begin{figure}[p]
  \centering
  \input{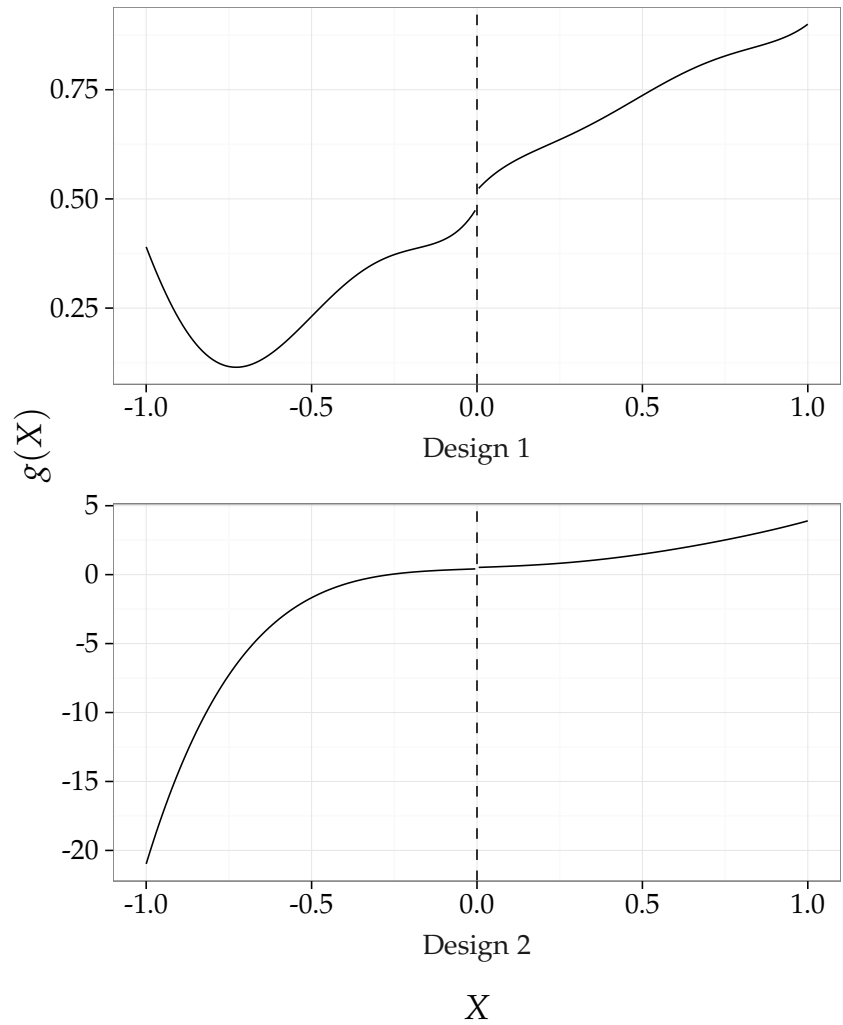}
  \caption{Monte Carlo study of regression discontinuity. Regression function
    $g(X)$ for designs 1 and 2.}\label{fig:rd-muX}
\end{figure}

\begin{figure}[p]
  \centering
  \input{./MC-thetah-1.tex}
  \caption{Monte Carlo study of regression discontinuity. Function $\theta(h)$
    for local linear regression for designs 1 and 2. Solid lines correspond to
    the triangular kernel, dotted lines to the uniform
    kernel.}\label{fig:rd-thetah}
\end{figure}

\begin{figure}[p]
  \centering
  \input{./MC-thetah-2.tex}
  \caption{Monte Carlo study of regression discontinuity. Function $\theta(h)$
    for local quadratic regression for designs 1 and 2. Solid lines correspond
    to the triangular kernel, dotted lines to the uniform
    kernel.}\label{fig:rd-thetah-quadratic}
\end{figure}